%% file: modal-k-circuits.tex
\title{\vspace*{-4ex}
	On computation with `probabilities' modulo $k$}
\author{
	Niel de Beaudrap}
\date{23 December 2014}

\documentclass[10pt,a4paper]{article}

\input{modal-k-circuits-preamble}

\preambletrue
\preliminariesfalse
\tikztrue
\scholiumfalse
\FullArticlefalse
\synopsisfalse

\begin{document}

\maketitle

\begin{abstract}
\noindent
	We propose a framework to study models of computation of indeterministic data, represented by abstract ``distributions''.
	In these distributions, probabilities are replaced by ``amplitudes'' drawn from a fixed semi-ring $S$, of which 
	the non-negative reals, the complex numbers,
	finite fields $\F_{\!\!\;p^r}$, and cyclic rings $\Z_k$ are examples.
	Varying $S$ yields different models of computation, which we may investigate to better understand the (likely) difference in power between randomised and quantum computation.
%
	The ``modal quantum states'' of Schumacher and Westmoreland~\cite{SW10} are examples of such distributions, for $S$ a finite field.
	For $S = \F_2$, Willcock and Sabry~\cite{WS11} show that UNIQUE-SAT is solvable by polynomial-time uniform circuit families consisting of invertible gates.
	We characterize the decision problems solvable by polynomial uniform circuit families,
	using either invertible or ``unitary'' transformations over cyclic rings $S = \Z_k$, or (in the case that $k$ is a prime power) finite fields $S = \F_k$.
	In particular, for $k$ a prime power, these are precisely the problems in the class $\Mod[k]\P$.
\end{abstract}


\section{Introduction}
\ifpreamble

An indeterministic computation is one in which a computational system occupies states which are not determined by the system having been in a given configuration at any earlier time.
This may occur in models of computation in which some configurations can lead to multiple possible future ones, as with nondeterministic Turing machines and randomised algorithms.
Furthermore, as in quantum computation, it may be possible to arrive at a final state which is described by a single configuration, by an evolution which is not easily described by assignments of configurations to intermediate times.
In each case, we may describe the state of a computation by a ``distribution'' over a set of \emph{possible} configurations.
An indeterministic computation is then a sequence of transformations of such distributions; and a model of indeterministic computation --- such as nondeterministic Turing machines, randomised circuits, or unitary quantum circuits --- is a means of describing a range of such computations.

Much of complexity theory is about indeterminism.
The questions ${\P \mathrel{?\!=} \BPP}$ and ${\P \mathrel{?\!=} \NP}$ each concern whether some kind of indeterminism can be simulated in polynomial time.
The question ${\BPP \mathrel{?\!\subset} \NP}$ concerns the relationship between two kinds of indeterminism; a similar question, which may be practically important, concerns whether randomized algorithms can efficiently simulate quantum algorithms.
Let $\BQP$ be the class of decision problems which can be solved with bounded error, by polytime-uniform unitary circuits which read out one bit at the end as output~\cite{NC,BV97}.
Is the containment $\BPP \subset \BQP$ strict?
Problems such as factoring and discrete logarithms are contained in \BQP~\cite{Shor97}, and are considered unlikely to be in $\BPP$, so it is usually supposed that the answer is ``yes''.
If so, could there be a simple reason why?

Past criticisms of quantum computation~\cite{Lev00,Gold05} touched on the precision of quantum amplitudes, and the exponential size of quantum state vectors (often oversimplified as ``exponential parallelism''), as extravagant resources which quantum computation exploits.
But these are features of probability vectors as well.
Both probability vectors and quantum state vectors use distributions --- functions ranging over real or complex numbers --- to describe indeterminism; and even quantum computers restricted to real-valued amplitudes may simulate arbitrary quantum computations~\cite{Shi03,Ahar03}.

Probability distributions and quantum state vectors only differ in two related ways: \textbf{(a)}~quantum states transform by reversible rotations rather than irreversible mixing operations, and \textbf{(b)}~quantum states can have coefficients which are neither positive nor zero.
In particular, transitions from different configurations 
may result in \emph{destructive interference}, in which transitions from different configurations give (partially or totally) cancelling contributions to the amplitude of a later configuration.
This hints that destructive interference may drive quantum speed-ups.

While destructive interference clearly occurs in many quantum algorithms, it may not be helpful to emphasize it as a computational effect.
If it is difficult to arrange for a quantum algorithm to produce useful interference patterns, perhaps one should consider it to be a \emph{symptom} of the computational power of quantum mechanics, rather than a meaningful ``cause''.
In quantum mechanics, destructive interference is merely a consequence of the Schr\"odinger Wave Equation: some subtler feature of Schr\"odinger evolution may be a more fruitful subject of scrutiny.

What could it mean for destructive interference, in itself, to be computationally powerful?
To explore this idea, we consider indeterministic models of computation, which differ from quantum computation but have similar forms of destructive interference.
To this end, we study ``distributions'' similar to probability vectors and quantum state vectors, but which takes values over an arbitrary semi-ring $S$, rather than $\R$ or $\C$.

Schumacher and Westmoreland~\cite{SW10,SW-2012} explore what mathematical features of quantum mechanics remain when one replaces complex amplitudes with elements of a finite field $\F_k$.
Fields have negatives for every element, so this substitution is useful for considering destructive interference.
However, finite fields have no notion of measure which could yield a consistent theory of probability: these distributions only admit weaker notions of ``possibility'', ``impossibility'', and ``necessity''.
Schumacher and Westmoreland call these \emph{modal quantum states} for this reason.
One may still define a theory of exact algorithms for these distributions, by distinguishing between outcomes which are either ``impossible'' or ``necessary''.
Ref.~\cite{SW10} shows that $\F_k$-valued distributions have exact communication protocols which are analogues of teleportation~\cite{BBCJPW-93} and superdense coding~\cite{BW-1992}.
Investigating the computational power of transformations of \mbox{$\F_2$-valued} modal distributions, 
Willcock and Sabry~\cite{WS11} describe an exact ``modal quantum'' algorithm for {UNIQUE-SAT}~\cite{VV86}.
In follow-up work by Hanson~\etal~\cite{HOSW11}, they propose a restriction to vectors with unit $\ell_2$ norm with the aim of investigating how this limits the computational power of modal quantum computations.

We propose a general circuit-like theory of \emph{modal computation}, consisting of transformations of semiring-valued distributions which represent indeterministic data.
This theory includes computation on probability distributions, quantum state-vectors, and ``modal quantum'' states~\cite{SW10} as examples.
We then apply this theory to distributions over Galois rings $R$~\cite{Wan-2003}, which include finite fields $\F_k$ and cyclic rings $\Z_k$ as special cases for prime powers $k$.
We reduce the study of these distributions to the simplest case $R = \Z_k$, and describe ways in which nondeterministic Turing machines may partially simulate computations on these distributions (and vice versa).
We thereby prove that the problems which can be decided by polynomial-size uniform circuit families in these models are those in the counting class $\Mod[k]\P$~\cite{BGH90}.
This demonstrates how interference in itself may contribute to the power of a model of computation.

\subsection{Summary of the article}

The definitions of this article are largely informed by the study of quantum computation; but we do not require the reader to be familiar with quantum computation, and develop our framework independently of it.
For readers who \emph{are} familiar with quantum computation, we may summarize our results as follows.
\begin{definition*}[sketch]
  Consider the model of computation which one obtains, by taking quantum circuits and replacing all complex coefficients in state vectors and gates by integers modulo $k$.
	State-vectors are then vectors of dimension $2^N$ over $\Z_k$, where $N$ is the number of qubits involved (the standard basis states remain unchanged); the gates are replaced by $2^h \x 2^h$ matrices over $\Z_k$.
	The circuit is \emph{$\Z_k$-invertible} if one replaces the unitary gates by matrices which are merely invertible; it is \emph{$\Z_k$-unitary} if one requires each gate $U$ to satisfy $U\:\!\trans U = \idop$.
	Then \GLP[\Z_k] is the set of decision problems which can be solved exactly by a polynomial-time uniform $\Z_k$-invertible circuit family, and \UnitaryP[\Z_k] is the set of decision problems which can be solved exactly by a polynomial-time uniform $\Z_k$-unitary circuit family.
\end{definition*}

\begin{theorem}
	\label{thm:mainResult}
	For $k$ a prime power, $\UnitaryP[\Z_k] = \GLP[\Z_k] = \Mod[k]\P$.
	Furthermore, these equalities still hold if $\Z_k$ is replaced by any Galois ring of character $k$ (such as a finite field in the case that $k$ is itself prime).
\end{theorem}
\noindent
To complete the definition above, we describe a general theory of abstract distributions, and a framework of exact and bounded-error computational models on these distributions.
When applied to distributions over $\R$ or $\C$, this framework yields familiar classes such as $\P$, $\BPP$, and $\BQP$; for distributions over $\Z_k$, we obtain the new classes $\GLP[\Z_k]$ and $\UnitaryP[\Z_k]$ instead.

Our motivation for studying $\Z_k$-valued distributions is that it is a simple example of modal computation which has destructive interference of amplitudes.
However, apart from understanding destructive interference in general, it may also lead directly to an improved understanding of the power of \emph{exact} quantum computation, through the $p$-adic integers $\Z_{(p)}$~\cite{Robert-2000}.
As the complex numbers (without its usual topology) can be recovered as an appropriate closure of $\Z_{(p)}$, exact quantum computation might be recoverable by appropriate limits of $\Z_k$-modal computation.
Results along these lines would yield new lower bounds for $\BQP$: thus, the fact that $\UnitaryP[\Z_{p^{\!\!\;r}}]\! = \Mod[p]\P$ for each constant $r \ge 1$ is noteworthy.
We conclude the article by discussing lines of research suggested by our results, concerning the power of quantum computation.

These results should be considered as exploration of abstract models of indeterministic computation, and in particular, demonstrating connections between $\Z_k$-valued distributions and classical notions of nondeterminism.
Except in special cases, we would not expect that these models could be efficiently realised or simulated by deterministic Turing machines, or even by quantum computers.

\vspace*{-1ex}
\paragraph{Structure of the article.}

Section~2 contains background in algebra and computational complexity.
Section~3 presents the general framework of modal distributions, and defines notions of exact and bounded-error modal computation.
Section~4 uses this framework to define a theory of computation on ``$R$-modal'' distributions for $R$ a Galois ring (such as a cyclic ring $\Z_{p^r}$ or finite field $\F_{\!\!\;p^r}$).
We show that for these circuits, bounded-error computation can be reduced to the study of exact algorithms, and $R$-modal circuits may be simulated by $\Z_k$-modal circuits for $k = \Char(R)$, allowing us to reduce the theory to exact computation on $\Z_k$-valued distributions.
In Section~5 we present the main result $\GLP[\Z_k] = \UnitaryP[\Z_k] = \Mod[k]\P$ for $k$ a prime power, and also characterize the power of $\GLP[\Z_k]$ and $\UnitaryP[\Z_k]$ for arbitrary $k \ge 2$.
We conclude in Section~6 with commentary and lines of further research.

\subsection{Related work}

\paragraph{Semiring-valued distributions and counting complexity.}
A similar extension of probability distributions and quantum states to semi-rings in general is presented by Beaudry, Fernandez, and Holzer~\cite{BFH-2004}.
They describe the complexity of evaluating tensor networks over a few different semi-rings, both with bounded error and unbounded error.
In the two-sided bounded error setting, they show that the complexity of evaluating tensor formulas over the boolean semi-ring $(\{0,1\}, \vee, \mathbin\&)$, non-negative rationals $\Q_+$, and arbitrary rationals $\Q$, are complete for the complexity classes $\P$, $\BPP$, and $\BQP$ respectively.
Beaudry~\etal\ also attribute these distinctions to destructive interference; our results extend this line of investigation.

The amplitudes of the distributions we consider may be expressed by Valiant's matchcircuits~\cite{Valiant-2005}, which similarly describe data in terms of tensor networks.
Our results are more concerned specifically with tensor networks which have a directed acyclic structure, and which can therefore be construed as a sequence of computational steps acting on a piece of input data.
As a result, our formalism can be presented more directly as a computational model than the framework of Refs.~\cite{Valiant-2004,Valiant-2005}.

\vspace*{-1ex}
\paragraph{Destructive interference in quantum algorithms.}
The notion that destructive interference is a crucial phenomenon for quantum speed-ups is also addressed by works concerned with the classical simulation of quantum circuits.
Van~den Nest~\cite{VanDenNest-2011} noted that while sparse matrices may generate large amounts of entanglement, they did not seem (in an informal sense) to involve much destructive interference, and on that basis demonstrated settings in which sparse operations could be probabilistically simulated on quantum states.
Stahlke~\cite{Stahlke-2014} introduces a quantitative measure of destructive interference, which allows him to describe upper bounds on the complexity of simulating a quantum circuit by Monte Carlo techniques.
The results of those articles aim at upper bounds to the complexity of simulating quantum operations, depending on limitations on the amount of interference involved.
Our results instead represent a qualitative lower bound (using complexity classes rather than run-times) on the computational power of interference in a non-quantum setting.

\vspace*{-1ex}
\paragraph{Modal quantum states and computation.}

Our results are motivated by the models of Schumacher and Westmoreland~\cite{SW10} and the result of Willcock and Sabry~\cite{WS11}.
However, certain features of Schumacher and Westmoreland's modal quantum theory~\cite{SW10}, such as different bases of measurement, do not apply to all types of modal distribution (such as probability distributions).
They are therefore absent in our treatment.
We justify this omission, and describe how to recover measurement bases in those cases where this notion is meaningful, towards the end of Section~\ref{discn:measurementBasis}.

The general framework of modal distributions which we define in Section~\ref{sec:modalDefinitions} appears not to be strongly related to the generalized framework of Ref.~\cite{SW-2012}, as we are motivated by relationships between algebra and computation rather than non-signalling correlations (as in Barrett's generalized probabilistic theories~\cite{Barrett-2007}).

\vspace*{-1ex}
\paragraph{\emph{Unrelated} work.}

	We follow Schumacher and Westmoreland~\cite{SW10} in using the word `modal' to refer to a notion of contingency more general than probability; it does not refer to interpretations of quantum mechanics~\cite{DV-1998}.

\section{Preliminaries}

We introduce some algebraic tools and review basic notions of counting complexity.
In Section~\ref{sec:basicAlgebra} we describe ``semi-rings'', which are algebraic structures that we use to generalize sets of probabilities or quantum amplitudes.
We also consider a notion of ``unitarity'' of a linear transformation on a finite field, similar to Hanson~\etal~\cite{HOSW11}, in order to study the power of models of computation involving algebraic constraints similar to those of quantum computing.
In Section~\ref{sec:Mod-k-P}, we review basic ideas of counting complexity, including the complexity class \Mod[k]\P\ for integers $k \ge 2$.

\subsection{Algebraic preliminaries}
\label{sec:basicAlgebra}

\subsubsection{Semi-rings}

We wish to study a notion of a distribution which subsumes probability distributions, quantum state vectors, and the ``modal quantum states'' of Schumacher and Westmoreland~\cite{SW10}.
We therefore require an algebraic structure, which includes the non-negative reals $\R_+$, fields such as $\C$ and $\F_k$, and cyclic rings such as $\Z_k$ as examples.
For this, we use the notion of a (commutative) \emph{semi-ring}: a set $S$ together with
\begin{itemize}[itemsep=0.5ex]
\item 
	a commutative and associative addition operation $a+b$, which has an identity element $0_S \in S$; and
\item
	a commutative and associative multiplication operation $ab$, which has an identity element $1_S \in S$; and where we also require that
\item
	multiplication distributes over addition, $a(b+c) = ab + ac$.
\end{itemize}
Our references to semirings are essentially superficial: we use them merely to generalize the examples of $\R_+$, $\C$, and $\Z_k$, and do not invoke any deep results about them.
We present the notion of a semiring only to allow us to present the framework of this article with appropriate generality.

To better appreciate the variety of semirings, we make some observations.
A semi-ring may lack multiplicative inverses for $a \ne 1_S$, and may lack \emph{negatives}: elements $-a$ for each $a \in S$ such that $-a + a = 0_S$.
A \emph{ring} is a semiring whose elements all have negatives.
In general, the multiplication in a ring or semiring may be non-commutative, but throughout this article we consider only commutative \mbox{(semi-)rings}.
Examples of semirings include:
\begin{romanum}
\item
	The integers $\Z$, which have negatives for every element (and thus form a ring), but inverses only for $\pm 1$.
\item
	The non-negative reals $\R_+$, which have inverses for $a \ne 0$ (and thus form a \emph{semi-field}), but has no negatives for $a \ne 0$.
\item
	The non-negative integers $\N$, which have neither negatives (for non-zero elements) nor inverses (for elements other than $1$).
\item
	The complex numbers $\C$, which have both negatives and inverses for its non-zero elements, and thus form a \emph{field}.
\end{romanum}
To obtain a uniform theory of computation on distributions, our most general definitions are presented in terms of semi-rings.
However, to study destructive interference, our main results concern only the special case of rings, and in particular cyclic rings $\Z_k$ for $k \ge 2$.


\subsubsection{Conjugation and inner products over semi-rings}

To address the conjecture of Hanson~\etal~\cite{HOSW11} concerning ``unitary'' transformations over finite fields, we consider how to define an appropriate generalization of unitarity.

\begin{subequations}
\label{eqn:innerProducts}
We associate an ``inner product'' function\footnote{%
	In a finite semi-ring of non-zero character, these are not inner products in the of real and complex analysis: while every $\vec v \ne \vec 0$ has a $\vec w$ such that $\langle \vec v, \vec w \rangle \ne 0$, some $\vec v \ne \vec 0$ may have 	the property $\langle \vec v, \vec v \rangle = 0$.
	Our choice of terminology of ``inner product'' for these and similar two-variable functions is standard in coding theory~\cite{vanLint-1982,GottPhD}, and chosen for the sake of brevity. 
}
${\langle \ast, \ast \rangle: S^d \x S^d \to S}$ to each semiring $S$ and each $d \ge 1$, as follows.
For $S = \R$, the usual choice is the bilinear ``dot-product'',
\begin{equation}
	\label{eqn:boringInnerProduct}
	\langle \vec v, \vec w \rangle = \sum_{j=1}^d v_j w_j = \vec v\trans \vec w  ;
\end{equation}
for $S = \C$, we instead consider a sesquilinear\footnote{%
	Let $S$ be a semiring with a ``conjugation'' operation $s \mapsto \overline s$ such that $\overline 0_S = 0_S$,\, $\overline 1_S = 1_S$,\, $\overline{r+s} = \overline{r} + \overline{s}$, $\overline{rs} = (\overline r)(\overline s)$, and $\overline{(\overline s)} = s$ for all $r,s \in S$.
	A two-argument function $F: S^d \x S^d \to S$ is \emph{sesquilinear} if $F(a\vec v_1+b\vec v_2,\vec w) = \overline a F(\vec v_1,\vec w) + \overline b F(\vec v_2,\vec w)$ and $F(\vec v, a\vec w_1 + b\vec w_2) = aF(\vec v,\vec w_1) + bF(\vec v,\vec w_2)$, for all scalars $a,b \in S$. 
	If $\overline s = s$ for all $s \in S$, then $F$ is bilinear.
}
inner product,
\begin{equation}
	\label{eqn:interestingInnerProduct}
	\langle \vec v, \vec w \rangle = \sum_{j=1}^d \bar{v}_j w_j = (\bar{\vec v})\trans \vec w,
\end{equation}
where $\overline{(a+bi)} = a-bi \in \C$ denotes the conjugation operation for $a,b \in \R$.
\end{subequations}
These inner products are both bilinear or sesquilinear, satisfy either $\langle \vec v, \vec w \rangle = \langle \vec w, \vec v \rangle$ or $\langle \vec v, \vec w \rangle = \overline{\langle \vec w, \vec v \rangle}$, and are non-degenerate in that $\langle \vec v, \vec w \rangle = 0$ for all $\vec v$ if and only if $\vec w = \vec 0$.
In analogy to the cases $S = \R$ and $S = \C$, we assume below that a semiring $S$ comes equipped with a \emph{conjugation operation} $s \mapsto \overline s$, which we define for the purposes of this article\label{discn:conjugation} as a self-inverse\footnote{%
	Our interest in self-inverse conjugation operations is motivated by the conjecture of Hanson~\etal~\cite{HOSW11}, who consider formal analogues of the sesquilinear inner product of vector spaces over $\C$.
	Note however that in Galois theory, one may easily construct rings $Q$ which have automorphisms which are \emph{not} self-inverse, which are also referred to as ``conjugation'' operations.
	For instance, for any finite field $F$ of order $k$, one may construct an extension $E$ of order $k^e$ for any $e \ge 3$, equipped with a conjugation operation $\overline s = \varphi(s) = s^k$ such that $\id_E \ne \varphi \circ \varphi$.
}
automorphism of $S$: an operation which preserves $0_S$, $1_S$, and sums/products over $S$.
(This operation may be the identity operation, as it must be for instance in the case $S = \Z_k$ for $k > 1$, and as we conventionally consider for $S = \R$.)
In the case that $S$ is a quadratic extension of a finite field (or any Galois ring), there is a canonical way to produce a non-trivial automorphism of this sort:
for details, the interested reader is referred to Appendix~\ref{apx:quadraticExtensions}.
Having specified such a conjugation operation for a given semiring $S$, we fix the inner product functions associated with $S$ to be those described by Eqn.~\eqref{eqn:interestingInnerProduct}: this reduces to the dot product of Eqn.~\eqref{eqn:boringInnerProduct} in the case that $\overline s = s$ for all $s \in S$.

Any choice of inner product induces a notion of \emph{adjoint}: an involution $M \mapsto M\herm$ on linear transformations, such that $\langle \vec v, M\vec w\rangle = \langle M\herm \vec v, \vec w\rangle$.
For inner products such as those above, 
$M\herm = (\bar M)\trans$; if $\overline s = s$ for all $s \in S$, this reduces to $M\herm = M\,\trans$.
In any case, 
we define a unitary matrix to be one such that $U\herm U = \id$, so that $\langle U\herm \vec v, U \vec w\rangle = \langle \vec v, \vec w \rangle$.
Hanson~\etal~\cite{HOSW11} consider notions of unitarity only for finite field extensions $S = \F_{p^2} \cong \F_p{[\:\!i\:\!]}$, for primes $p \equiv 3 \pmod{4}$ and where $i^2 = -1_S$, which are formally very similar to the field extension $\C = \R{[\:\!i\:\!]}$.
However, our analysis applies to all self-inverse conjugation operations, including the case that $\overline s = s$ for all $s \in S$.
Thus we may speak of ``unitary'' transformations for arbitrary semirings, though it will in some cases amount to ``orthogonality'' (for which $U\:\!\trans U = \id$).

\subsection{Counting complexity and the class $\mathsf{Mod_{\mathit k}P}$}
\label{sec:Mod-k-P}

We now review some basic ideas in counting complexity, including the class \Mod[k]\P\ and the relationship between linear transformations and \#\P. 
We assume familiarity with nondeterministic Turing machines: for introductory references see Refs.~\cite{Papadimitriou-1994,AB-2009}.

Let $\{0,1\}^\ast$ be the set of boolean strings of any finite length.
Valiant~\cite{Valiant-1979} defines \#\P\ as the set of functions $f: \{0,1\}^\ast \to \N$ for which there is a nondeterministic polynomial-time Turing machine $\mathbf N$ such that
\begin{equation}
  f(x) = \#\{ \text{computational branches of $\mathbf N$ which accept, on input $x$} \}.
\end{equation}
The problem \#SAT, of determining the number of satisfying assignments for an instance of SAT, is a prototypical problem in \#\P: for a boolean formula $\varphi$ evaluating a logical formula, it suffices to consider the number of accepting branches of a nondeterministic Turing machine which guesses satisfying assignments of variables for $\varphi$.
Functions in \#\P\ are in general hard to compute, as they represent the result of the branching of nondeterministic Turing machines which are allowed to run for polynomial time.
For instance, the class \NP\ is the set of languages $L$ for which there is a function $f \in \text\#\P$, such that $x \in L$ if and only if $f(x) \ne 0$.

The classes \Mod[k]\P\ for $k \ge 2$ were defined by Beigel, Gill, and Hertrampf~\cite{BGH90}, generalizing the class $\oplus\P = \Mod[2]\P$ defined by Papadimitriou and Zachos~\cite{PZ-1983}.
These classes attempt to capture some of the complexity of \#\P\ functions through decision problems, and are defined similarly to \NP:
\begin{definition}
  For $k \ge 2$ an integer, \Mod[k]\P\ is the set of languages $L \subset \{0,1\}^\ast$ for which there exists a function $f \in \text\#\P$ such that $x \in L$ if and only if $f(x) \not\equiv 0 \pmod{k}$.
\end{definition}


Despite having similar definitions in terms of \#\P\ functions, the relationship between \Mod[k]\P\ and \NP\ is currently unknown.
For instance, for any $k \ge 2$, it is not known whether either containment $\NP \subset \Mod[k]\P$ or $\Mod[k]\P \subset \NP$ holds.
However, one may show that $\UP \subset \Mod[k]\P$ for any $k \ge 2$, where \UP\ is the class of problems in \NP\ which are decidable by nondeterministic Turing machines which accept on at most one branch.\footnote{%
	For $L \in \UP$, by definition the characteristic function $\chi_L: \{0,1\}^\ast \to \{0,1\}$ is in \#\P, which satisfies the acceptance conditions for $\Mod[k]\P$ for any $k \ge 2$.}
In particular, both $\Mod[k]\P$ and $\NP$ contain the deterministic class \P.
Also, from their definitions in terms of \#\P\ functions, both classes of problems are contained in the class $\P^{\text\#\P}$ of problems solvable in polynomial time, by a deterministic Turing machine with access to a \#\P\ oracle.

The classes \Mod[k]\P\ have useful properties when $k$ is a prime power: for instance, 
they are closed under subroutines in that case~\cite{BGH90}.
One might then think of $\Mod[k]\P$ as representing the computational power of an abstract machine which computes answers explicitly in its working memory.
Our results describe computational models with which one may provide such a description of $\Mod[k]\P$, for $k$ a prime power.

\section{Computations on ``modal distributions''}
\label{sec:modalDefinitions}

In this section, we describe a general framework for indeterministic computation, involving \emph{modal distributions} and \emph{modal state spaces}, which subsumes randomized computation and quantum computation.
We define notions of computability and (exact and bounded-error) computational complexity for modal distributions in general, and show how to recover traditional classes for decision complexity from this framework.

We do not assume any familiarity with quantum computation.
However, we adopt several conventions (and present some algebraic machinery) which will be familiar to readers who are familiar with the subject.

\ifsynopsis
\paragraph{Synopsis of our framework.} 

The framework that we develop in this section includes a number of technical details, which we consider important to obtain a versatile and well-defined theory of distributions.
However, these details are not essential to a first reading of the results of Section~\ref{sec:characterization}.
To allow the reader to quickly skip ahead to the main results, we summarize the main ideas of this section below.
\begin{enumerate}[label=\textbf{$\boldsymbol\SectionSymbol$3.\arabic{*}}]
\item \textbf{Distributions, states, and notation.}

	Given a fixed semiring $S$, we interpret elements of $S$ as \emph{amplitudes}, which are substitutes for probabilities.
	We then consider \textbf{$S$-modal distributions} over boolean strings $x \in \{0,1\}^\ast$ as mappings $\psi: \{0,1\}^\ast \to S$ in which $\psi_x \in S$ is the amplitude of $x$.
	These are interpreted as ``distributions'' in the sense that a boolean string $x \in \{0,1\}^\ast$ is considered \textbf{possible} for $\psi$ if $\psi_x \ne 0_S$.

	We represent distributions as vectors.
	Independently distributed boolean strings $x, y \in \{0,1\}^\ast$, governed by distributions $\alpha$ and $\beta$, have a joint distribution given by $\psi = \alpha \ox \beta$, so that $\psi_{x,y} = \alpha_x \beta_y$.
	We write $\sB = S^{\{0,1\}}$ for the sake of brevity, and $\sB\sox{n} = \sB \ox \sB \ox \cdots \ox \sB$ for distributions on $n$-bit strings.

	For the sake of convenience, we adopt Dirac notation as a vector notation.
	A distribution $\psi$ on boolean strings $x \in \{0,1\}^n$ is written $\ket{\psi}$.
	Special exceptions include the null distribution, which is always written $\vec 0$, and standard basis vectors $\vec e_x \in \sB\sox{n}$, which are written as $\ket{x}$.
	Expressions such as $\ket{\alpha}\ket{\beta}\cdots\ket{\omega}$ have implicit tensor products; and we may abbreviate $\ket{x_1}\ket{x_2}\cdots\ket{x_n}$ by $\ket{x_1 x_2 \cdots x_n}$ for $x \in \{0,1\}^n$.
	The adjoint of $\ket{\psi}$ is written $\bra{\psi}$, and the inner product $\langle \phi, \psi \rangle$ is written $\bracket{\phi}{\psi}$.
	
	An \textbf{$S$-modal state space} is a subset of the set of $S$-distributions which includes all standard basis vectors $\ket{x}$ for $x \in \{0,1\}^\ast$, excludes $\vec 0$, and is closed under tensor products.
	These represent distributions which are normalized (in the sense of a probability distribution whose sum is $1$), and and allows us to define a restrictive notion of \textbf{necessity}.
	For a state $\ket{\psi}$ representing distribution over $\{0,1\}^n$, and for a particular bit position $1 \le j \le n$, we say that the $j\textsuperscript{th}$ bit is necessarily $a \in \{0,1\}$ if $\forall x \in \{0,1\}^n : (\psi_x \ne 0) \implies (x_j = a)$.
		
	We refer to these distributions as ``modal'' because they may express notions of ``possibility'' and ``necessity'', but not necessarily any stronger notion of the significance of outcomes.
		
\item
	\textbf{Transformations of states.}

	Continuing the analogy between $S$-modal distributions and probabilities, we consider only \emph{linear} transformations of distributions, so that the possible values $x \in \{0,1\}^\ast$ transform independently of each other.
	For a given choice of state-space (there may be more than one for a given $S$), we also restrict to transformations which map  states to other states.
	
	To define a model of computation, we consider ways that a set of ``primitive'' valid transformations may be combined to produce more complex transformations.
	The primitive transformations which we consider in this article are maps $\ket{\psi} \mapsto \ket{\psi}\ket{0}$ which prepare a new bit in the constant state $\ket{0}$, and \emph{gates} $M: \sB\sox{h} \to \sB\sox{\ell}$ (that is: $2^\ell \x 2^h$ matrices over $S$) for $h,\ell \ge 1$.
	To describe the action of a gate $M$ on some subset $\rA \prsubset \{1,2,\ldots,n\}$ of an $n$-bit string, we form the tensor product of $M$ with identity operations $\idop := \id_{\sB}$ on those bits which are unaffected.
	We may then form matrix-products of these tensor-product operations, which we call \textbf{modal circuits}.
	We interpret these circuits as describing time-ordered transformations of $S$-modal states, in analogy to boolean circuits.
			
\item \textbf{Circuit complexity.}

	The state spaces which we consider, for $S$ a finite field or Galois ring, are
	\begin{itemize}
	\item 
		the set of $S$-distributions $\ket{\psi}$ for which $\exists \ket{\phi} : \bracket{\phi}{\psi} = 1_S$, whose valid transformations are all left-invertible transformations; and
	\item
		the set of $S$-distributions $\ket{\psi}$ for which $\bracket{\psi}{\psi} = 1_S$, whose valid transformations include all ``unitary embeddings'' $U : \sB\sox{n} \to \sB\sox{n}$ such that $U\herm U = \idop\sox{n}$.
	\end{itemize}
	As the valid transformations are left-invertible, we consider sets of primitive gates which are themselves invertible, realising transformations $M: \sB\sox{h} \to \sB\sox{h}$.	
	We define two different circuit models based on these two types of operations: \textbf{invertible $S$-modal circuits} and \textbf{unitary $S$-modal circuits}.
	
	The complexity of a particular valid transformation depends on the primitive gate set which one uses.
	Each primitive gate is associated with a ``cost'', which is bounded above by the complexity of computing its coefficients.
	For a finite sized set of primitive gates, the cost of each gate may then be bounded by a constant.
	
	As with boolean circuits and quantum circuits, efficient algorithms are expressed in terms of \emph{polynomial-uniform circuit families} $\{ C_n \}_{n \ge 1}$.
	These are families where each circuit $C_n$ has a structure that is specified by a deterministic Turing machine which halts in in time $\poly(n)$, and each individual gate has cost at most $\poly(n)$.
	In particular, each $C_n$ acts on at most $n + m$ bits for $m \in \poly(n)$.
	
	We wish to consider the complexity of deciding languages $L$, using $S$-modal circuits, with bounded error.
	However, for $S = \F_k$, the only way to avoid an \emph{unbounded} error decision model is to restrict to \textbf{exact computation}, in which the output (given by the final bit) is necessarily the correct answer for all inputs.
	That is, for invertible or unitary $S$-modal circuits $C_n$, we require that $C_n \ket{x}\ket{0^m} = \ket{\psi_x} \ket{L(x)}$, where $L(x) \in \{0,1\}$ indicates whether $x \in L$.

	We define the complexity classes \GLP[S] and \UnitaryP[S] as the classes of problems which are exactly solvable by polynomial-uniform circuits, of invertible or unitary gates respectively, over $S$.
	The main result of this article is that when $S$ is a finite field or Galois ring of size $k$, we have $\GLP[S] = \UnitaryP[S] = \Mod[k]\P$.
\end{enumerate}
The remainder of this Section elaborates and motivates these ideas, and describes connections to probability theory and quantum information theory.
Readers who are content with this summary may skip to Section~\ref{sec:characterization} (page~\pageref{sec:characterization}).
\fi

\subsection{Modal distributions, states, and valid transformations}
\label{sec:distributions}

We define the theory of modal distributions over boolean strings: the generalization to distributions over other countable sets should be clear.
\begin{definition}
	Let $S$ be a semiring, 
	and $n \ge 0$ an integer.
	An \emph{$S$-distribution} over $\{0,1\}^n$ is a function $\psi: \{0,1\}^n \to S$. 
	We refer to 
	$\psi(x)$ 
	as \emph{the amplitude of $\psi$ at $x$.}
	We say that $x$ is \emph{possible}
	(or \emph{a possible value}) for $\psi$ if $x \in \supp(\psi)$, that is $\psi(x) \ne 0$.\footnote{%
		This definition of ``possibility'' differs slightly from that of Ref.~\cite{SW10}.
		Readers familiar with quantum information theory may be interested in the remarks on this point, on page~\pageref{discn:measurementBasis} toward the end of this Section.
	}
\end{definition}
\noindent

We interpret $\psi$ as an ensemble or result of an indeterministic process, whose outcomes are the boolean strings $\{0,1\}^n$.
Examples include \ifFullArticle indicator functions of subsets $L \subset \{0,1\}^n$, which are $\B$-distributions for the boolean semiring $\B = (\{0,1\}, \vee, \&)$; \fi histograms on $\Sigma^n$, which are $\N$-distributions; probability distributions over $\{0,1\}^n$, which are \mbox{$\R_+$-dist}\-ributions; and quantum state vectors on $\{0,1\}^n$, which are $\C$-distributions.
The modal quantum states of Schumacher and Westmoreland are $\F$-distributions for a given finite field $\F$. \ifFullArticle

We wish to consider ways that an $S$-distribution may (abstractly) represent the state of a computational system, in such a way that the valid states include ones which we conceive of as representing ``point-mass'' distributions.
\begin{definition}
  For a semiring $S$, a \emph{computational basis distribution} on an alphabet $\{0,1\}$ is a distribution $\vec e_x \in S^{\{0,1\}^n}$ such that 
  \begin{equation}
    \vec e_x(y)	=	\begin{cases} 1, & \text{if $x = y$}; \\ 0, & \text{otherwise}. \end{cases}
  \end{equation}
\end{definition}
\noindent
The term ``computational basis'' is used because the distributions $\{ \vec e_x : x \in \{0,1\}^n \}$ do form a basis for the space of distributions $\psi:\{0,1\}^n \to S$, and because we intend this distribution to represent a computational system whose state is $x \in \{0,1\}^n$ with certainty.
These semantics for $\vec e_x$ (and the sense in which $\psi: \{0,1\}^n \to S$ is ``a distribution'' over $x \in \{0,1\}^n$) are provided by Definition~\ref{def:modality}.
We first present the following definition to motivate the relevant mathematical tools, and to suggest the sense in which we wish to think of distributions.

\begin{definition} 
	Let $n \ge 0$ be an integer, $\{0,1\}$ a finite alphabet and $S$ a semiring.
	For $(\rA, \rB)$ a partition of $\{1,2,\ldots,n\}$ and $\psi : \{0,1\}^n \to S$ an $S$-distribution, we say that $\rA$ and $\rB$ are \emph{independent with respect to $\psi$} if
	\begin{equation}
		\psi(x) = \alpha(x_A) \,\beta(x_B)  \quad\text{for all $x \in \{0,1\}^n$},
	\end{equation}
	for some $\alpha: \{0,1\}^A \to S$ and $\beta: \{0,1\}^B \to S$.
	Here $x_{\rA}$ and $x_{\rB}$ are the restrictions of $x$ to the indices of ${\rA}$ and $\rB}$ respectively.
	We call $\alpha$ and $\beta$ the \emph{independent distributions} induced by $\psi$ on $\rA$ and $\rB$.
	If $\rA$ and $\rB$ are not independent with respect to $\psi$, we say that they are \emph{correlated with respect to $\psi$}.
\end{definition}
\noindent
For independent subsystems $\rA$ and $\rB$, we can without ambiguity speak of the ``state of $\rA$'' without reference to the global system.
Note that $\vec e_x$ is always a product distribution, as $\vec e_x(y) = \vec e_{x_1\!}(y_1)\,\vec e_{x_2\!}(y_2)\,\cdots\,\vec e_{x_n\!}(y_n)$ for $x,y \in \{0,1\}^n$: we may interpret this as indicating that each symbol has a definite value which is independent of the others.

\subsubsection{Distributions as vectors}
\else
\subsubsection{Vector representation and Dirac notation}

\fi
We identify $S$-distributions 
$\psi: \{0,1\}^n \to S$ with vectors\ifFullArticle\footnote{%
	Technically, these are only vectors when $S$ is a field, and the maps $\psi: \{0,1\}^n \to S$ form a ``vector space''; we abuse this distinction here.}\else\ \fi
$\psi \in S^{\{0,1\}^n}$\!, and write $\psi_x = \psi(x)$ for $x \in \{0,1\}^n$.
\iftrue 
We express distributions as column vectors with the coefficients of $\psi \in S^{\{0,1\}^n}\!$ in lexicographical order, such as
\begin{equation}
  \psi =	\begin{smatrix}	\psi_{00\cdots000} \\ \psi_{00\cdots001} \\ \psi_{00\cdots010} \\ \vdots \\ \psi_{11\cdots111}
        	\end{smatrix}.
\end{equation}
In particular, for $x \in \{0,1\}^n$, let $[x] \in \N$ be the integer with binary expansion $x$: the distribution $\vec e_x$ is a standard basis vector with a $1$ in the $[x]+1\textsuperscript{st}$ coefficient.

Standard mathematical definitions and constructions for probability distributions and quantum states~\cite{McCullagh-1987,NC} may be extended to $S$-distributions for any semiring $S$.
For instance: if $\rA,\rB \subset \{1,2,\ldots,n\}$ is a partition of $\{1,2,\ldots,n\}$ and if $\psi$ decomposes as a product $\psi(x) = \alpha(x_\rA)\beta(x_\rB)$ for all $x \in \{0,1\}^n$, for some $\alpha \in S^{\{0,1\}^\rA}$ and $\beta \in S^{\{0,1\}^\rB}$, we say that $\rA$ and $\rB$ are independently distributed for $\psi$.
As a vector, we then have $\psi = \alpha \ox \beta$, where $\ox$ is the tensor product.
(In particular, ${\vec e_x ={}} {\vec e_{x_1} \ox{}} \vec e_{x_2} \ox \cdots {{}\ox \vec e_{x_n}}$ for $x \in \{0,1\}^n$.)
Representing $\psi$ as a column vector, we may form tensor products using the Kronecker product, regardless of the particular semiring $S$:
\begin{equation}
		\alpha \ox \beta
  =
		\begin{bmatrix} \alpha_{00\cdots00} \\ \alpha_{00\cdots01} \\ \vdots \\ \alpha_{11\cdots11} \end{bmatrix} \ox \begin{bmatrix} \beta_{00\cdots00} \\ \beta_{00\cdots01} \\ \vdots \\ \beta_{11\cdots11} \end{bmatrix}
	=
		\begin{bmatrix} \alpha_{00\cdots00} \beta_{00\cdots00} \\ \alpha_{00\cdots00} \beta_{00\cdots01}  \\ \vdots \\ \alpha_{00\cdots00} \beta_{11\cdots11} \\  \alpha_{00\cdots01} \beta_{00\cdots00} \\ \alpha_{00\cdots01} \beta_{00\cdots01} \\ \vdots \\ \alpha_{11\cdots11} \beta_{11\cdots11} \end{bmatrix} 
	=
		\left[\;\;\begin{matrix} \\[-1ex] \!\!\!\alpha_{00\cdots00} \beta\!\!\! \\[1ex] \hline \\[-1ex]  \!\!\!\alpha_{00\cdots01} \beta\!\!\! \\[1ex] \hline \\[-2ex]  \vdots \\[1ex] \hline \\[-1ex] \!\!\!\alpha_{11\cdots11} \beta\!\!\! \\[2ex] \end{matrix}\;\;\right],
\end{equation}
where the subscripts of $\alpha_{x_\rA}$ and $\beta_{x_\rB}$ run over $x_\rA \in \{0,1\}^\rA$ and $x_\rB \in \{0,1\}^\rB$.
If $\rA$ and $\rB$ are not contiguous blocks of bits, the tensor product corresponds to the same vector as above, up to a permutation which maps the concatenation $x_\rA x_\rB$ to the string $x$ of which the strings $x_\rA$ and $x_\rB$ are restrictions.
This construction respects the decomposition of $\psi \in S^{\{0,1\}^n}$ over the computational basis $\{ \vec e_x : x \in \{0,1\}^n \}$.


To avoid too many subscripts of the sort seen above, we adopt Dirac notation for convenience.\footnote{%
	We prefer to use Dirac notation rather than the alternative notation of Ref.~\cite{SW10}, as it seems unnecessary to introduce a distinct vector notation for semirings $S \ne \C$.
	Our use of Dirac notation should hopefully not lead to any confusion, as we do not consider any examples of quantum states $\ket{\psi} \in \C^d$ in this article.
}
We write a typical distribution $\psi \in S^{\{0,1\}^n}$ as $\ket{\psi}$ (read as ``ket psi'').
A
A notable exception is the zero vector (of any dimension), which is always denoted $\vec 0$.
We also write computational basis states $\vec e_x$ for $x \in \{0,1\}^\ast$ differently, representing them by $\ket{x} = \ket{x_1 x_2 \cdots x_k}$.
We may write tensor products by concatenation of ``kets'', omitting $\ox$ symbols for the sake of brevity: 
\begin{equation}
  \ket{\alpha} \ox \ket{\beta} \ox \cdots \ox \ket{\omega}
  \;=\;
	\ket{\alpha}\ket{\beta}\cdots\ket{\omega}.
\end{equation}
We identify the standard basis states on strings $x \in \{0,1\}^n$ with the tensor products of the individual bit-values:
\begin{equation}
  \ket{x_1} \ox \ket{x_2} \ox \cdots \ox \ket{x_k} \;=\;	\ket{x_1}\ket{x_2}\cdots\ket{x_k} \;=\; \ket{x_1 \, x_2 \, \cdots \, x_k}.
\end{equation}
When the tensor factors describe distributions on specific ``subsystems'', corresponding to a partitioning $\rA, \rB, \ldots, \Omega$ of $\{1,2,\ldots,n\}$, we may write the subsystems on which the different factors act as subscripts: if $\ket{\alpha}$ is a distribution on $\rA$, and $\ket{\beta}$ a distribution on $\rB$,~\etc\ then we may write the joint distribution of the system as
\begin{equation}
	\ket{\alpha}_{\rA}\ket{\beta}_{\rB}\cdots\ket{\omega}_\Omega.
\end{equation}
We denote the adjoint $\ket\psi\herm$ of a distribution by $\bra{\psi}$ (read ``bra psi''), where this adjoint is defined as in Section~\ref{sec:basicAlgebra}.
We then write inner products $\langle \phi, \psi \rangle$ for vectors $\phi,\psi \in S^{\{0,1\}^n}$ in Dirac notation as a ``bra-ket'', $\bracket{\phi}{\psi}$.
In particular, for any $\ket{\psi} \in S^{\{0,1\}^n}$, we may write the coefficient $\psi_{x}$ for $x \in \{0,1\}^n$ by $\bracket{x}{\psi}$.
For example, we have $\bracket{x}{y} = 0$ for $x \ne y$, and $\bracket{x}{x} = 1$.
\else
	We adopt Dirac notation for vectors $\ket{\psi} \in S^{\{0,1\}^n}$ in the usual way: $\ket{x}$ denotes the standard basis state indexed by $x \in \{0,1\}^n$, $\bra{\psi}$ denotes $\ket{\psi}\herm$, and $\bracket{\phi}{\psi}$ denotes the inner product (over $S$) of the distributions $\ket{\phi}$ and $\ket{\psi}$.
	If $\rA$ and $\rB$ are independent subsystems, with distributions $\ket{\phi}_\rA$ and $\ket{\psi}_\rB$, their joint distribution is formed by the tensor product $\ket{\phi}_{\rA} \ox \ket{\psi}_\rB$ (written $\ket{\phi}_{\rA}\ket{\psi}_{\rB}$ for short).
\fi

\subsubsection{Distribution space and state spaces}
\smallskip

The following definition captures the full range of $S$-distributions over boolean strings $x \in \{0,1\}^\ast$, in a way which is closed under tensor products.
\begin{definition}
	\label{def:distributionSpace}
	For a semiring $S$, write $\sB = S^{\{0,1\}}$ for distributions over one bit $x_j \in \{0,1\}$, and $\sB\sox{n} \ifFullArticle= (S^{\{0,1\}})\sox{n}\fi \cong S^{\{0,1\}^n}$ for distributions over $x \in \{0,1\}^n$.
	The \emph{$S$-distribution space} over $\{0,1\}^\ast$ is then
	\begin{equation}
		\mathscr D \;:=\; S \;\oplus\; \sB \;\oplus\; \sB\sox{2} \;\oplus\; \cdots \;\oplus\; \sB\sox{n} \;\oplus\; \cdots  
	\end{equation}
	where $\oplus$ denotes a direct sum.
	For $\rA \subset \N$ an index-set and $a \in \{0,1\}^\rA$, $a$ is \emph{possible} (or a possible value) for $\ket{\psi}$ on $\rA$ if there is $x \in \supp(\psi)$ such that the substring $x_\rA$ is well-defined, and equal to $a$.
\end{definition}
\noindent
We define the distribution-space using a tensor algebra for the sake of brevity when referring to arbitrary distributions $\ket{\psi} \in \sD$, as in the definition of a state-space below:
\begin{definition}
	\label{def:stateSpace}
	For a semiring $S$, a \emph{state space} over $\{0,1\}^\ast$ is a subset $\sS \subset \sD$ such that:
	\begin{itemize}
	\item 
		$\sS$ is closed under taking tensor products and extracting tensor factors: that is, for $\ket{\psi} \in \sS$ and $\ket{\phi} \in \sD$, we have $\ket{\phi}\ox\ket{\psi}, \ket{\psi}\ox\ket{\phi} \in \sS$ if and only if $\ket{\phi} \in \sS$;
	\item
		$\sS$ contains $\ket{x} \in \sD$ for each $x \in \{0,1\}^\ast$, and does not contain $\vec 0 \in \sD$.
	\end{itemize}
\end{definition}
\noindent
We regard distributions in general as potentially-underspecified descriptions of an ``ensemble'' of strings $x \in \{0,1\}^\ast$ (with the null distribution $\vec 0$ representing a completely unspecified ensemble).
States represent completely specified ensembles, such as normalised probability distributions.
We require states to be closed under tensor products, so  that joint distributions of independently distributed variables can be well-defined.
We also require them to be closed under extraction of tensor factors in order to be able to discuss the states of two subsystems when their distributions are independent.
\begin{definition}
	\label{def:modality}
  Let $\sD$ be a distribution space, $\sS$ be a state space, and let $\ket\psi \in \sD$.
	For a set $\rA \subset \N$ of indices, $a$ is \emph{necessary} for $\ket{\psi}$ on $\rA$ if $\ket{\psi} \in \sS$, and if for all $x \in \supp(\psi)$, $x_\rA$ is well defined and equal to $a$.
\end{definition}
\noindent
Loosely following the terminology of Schumacher and Westmoreland~\cite{SW10}, who called the non-zero distributions over $\F_k$ ``modal quantum states'', we refer to $\sD$ for arbitrary semirings $S$ as a \emph{modal distribution space}, the state-spaces $\sS$ as \emph{modal state spaces}, and the distributions $\ket{\psi} \in \sD$ and $\ket{\psi} \in \sS$ as \emph{modal distributions} or \emph{modal states}.

Note that if $\rA$ is necessarily $a \in \{0,1\}^A$ for some state $\ket{\psi} \in \sS$, we have $\ket{\psi} = \ket{a}_\rA \ox \ket{\psi'}_{\rB}$ for some state $\ket{\psi'} \in \sS$, where $\rB = \{1,2,\ldots,n\} \setminus \rA$.
Let $x \in \{0,1\}^n$ be an element of $\supp(\psi)$: by hypothesis we then have $x_\rA = a$. 
We may then decompose $\ket{x} = \ket{a}_{\rA} \ox \ket{x_\rB}_{\rB}$.
Collecting all of the components of $\ket{\psi}$ in the standard basis together, we have $\ket{\psi} = \ket{a}_\rA \ox \ket{\psi'}$.
Furthermore, because $\sS$ is closed under extraction of tensor factors and $\ket{a} \in \sS$, we have $\ket{\psi'} \in \sS$ as well.



	For computation on modal states, we consider only transformations which preserve the state-space.
	We also restrict ourselves to linear transformations, which therefore transform ``possible values'' $x \in \supp(\psi)$ independently of one another for $\ket{\psi} \in \sD$.

\begin{definition}
  Let $S$ be a semiring, and $\sS$ an $S$-state space.
	A \emph{valid transformation} for $\sS$ is a linear transformation $T: \mathscr D \to \mathscr D$ such that $T\ket{\psi} \in \sS$ for any $\ket{\psi} \in \sS$, and furthermore such that $(T \ox \id_{\sB}\sox{n}) \ket{\psi} \in \sS$ as well for any $n \ge 0$ and $\ket{\psi} \in \sS$ (where we interpret the latter map as performing the identity on $\sB\sox{m}$ for $m < n$).
\end{definition}
\noindent
\ifFullArticle
We restrict ourselves to linear transformations, as we are interested in the case where each possible value of a system is transformed independently of other possibilities, as with probabilities and quantum amplitudes.\fi
The condition involving $T \ox \id_{\sB}\sox{n}$ ensures that $T$ is meaningful as an operation on a subsystem, so that its validity is not dependent on the context in which it is applied.
\ifscholium
Restricting the set of valid transformations to linear transformations of distributions corresponds to the usual notion that a ``distribution'' is transformed according to a \emph{superposition principle}: distributions represent ``possibilities'' over computational basis states, and the way that the different possible states transform are independent of one another, so that the transformation of arbitrary states is linear.
\fi

More than one possible state space may exist for a given semiring $S$, and determines the set of valid transformations.
\ifFullArticle
For instance: while the set of $\R$-distributions $\ket{\psi}$ for which $\psi_x \ge 0$ and $\sum_x \psi_x = 1$ are preserved by stochastic transformations (which suffices to simulate randomized computation), the set of $\R$-distributions such that $\bracket{\psi}{\psi} = 1$ is instead preserved by orthogonal transformations (which suffice to simulate quantum computations~\cite{Shi03,Ahar03}).
\else 
For instance: $\R$-distributions $\ket{\psi}$ for which $\psi_x \ge 0$ and $\| \psi \|_1 = 1$ are preserved by stochastic transformations, and represent randomized computation.
The $\R$-distributions such that $\bracket{\psi}{\psi} = 1$ are instead preserved by orthogonal transformations, which suffice to simulate quantum computation~\cite{Shi03,Ahar03}.
\fi

\ifscholium
We may consider states $\ket{\psi}$ in either $\sS_2$ and $\sS_\ast$, whose amplitudes include opposite pairs $a, -a \in S$.
The linearity of the valid transformations may then result in \emph{destructive interference} of amplitudes, which limits the set of ``classical'' transformations included among the valid transformations.
A \emph{classical transformation} is a linear transformation $T: \sD \to \sD$ which maps
\begin{equation}
  T \ket{x} = \ket{f(x)}	\qquad\text{for some $f: \{0,1\}^\ast \to \{0,1\}^\ast$};
\end{equation}
however, if $f(y) = f(z)$ for some distinct $y,z \in \{0,1\}^n$, it follows that
\begin{equation}
  T\Bigl(\ket{y} - \ket{z}\Bigr) = \ket{f(y)} - \ket{f(z)} = \vec 0 \notin \sS.
\end{equation}
Because $\ket{y} - \ket{z} \in \sS_\ast$ for any ring $S$, it follows that $T$ is not valid for the generic state-space; similar problems arise for the $\ell_2$ state space.
\ifscholium
More generally, even for valid transformations $T$ of an $S$-state space, it may be that $\bra{r} T \ket{x} \ne 0$ and $\bra{r} T \ket{y} \ne 0$, but that there exists a state $\ket{\psi}$ for which $\psi_x, \psi_y \ne 0_S$ such that $\bra{r} T \ket{\psi} = 0$.
We say in this case that $T$ causes (or gives rise to) \emph{destructive interference} at $r$, for the input $\ket{\psi}$.
(This phenomenon is familiar in quantum computation, where it is used \eg~in phase estimation~\cite{K95}.)
\fi
Thus, irreversible functions $f: \{0,1\}^n \to \{0,1\}^m$ on the computational basis will generally not extend linearly to valid transformations $\ket{x} \mapsto \ket{f(x)}$.
However, by the remarks on reversible computation in Section~\ref{sec:simulatingNTM}, the valid transformations on $\sS_2$ and $\sS_\ast$ are still capable of simulating universal computation on $\{0,1\}^\ast$, because permutations of $\{0,1\}^\ast$ induce permutation operators on $\sD$, which are valid transformations of any state-space.
\fi

\vspace*{-2ex}
\paragraph{An aside: on measurement.}
Readers familiar with quantum computation may wonder why ``possibility'' and ``necessity'' are defined with respect to the standard basis.\label{discn:measurementBasis} 
We do so in contrast to Schumacher and Westmoreland~\cite{SW10}, who define a basis-dependent notion of measurement.
Our aim here is to present an abstract theory of distributions on $\{0,1\}^\ast$, which also includes the cases $S = \R_+$ of probability distributions and $S = \N$ of histograms (though they are not the main subject of this article).
The general theory therefore only defines possibility and necessity of definite boolean strings $x \in \{0,1\}^\ast$, which are distinguishable from each other in every such model.
In special cases such as $S = \F_k$ or $S = \C$, ``measurement bases'' may be recovered by considering which outcomes are possible or necessary after an invertible transformation of the state (a treatment which is not uncommon in the theory of quantum computation~\cite{NC}).\footnote{%
	With regards to quantum measurement, different ``measurement bases'' arise from different couplings of a system to measurement devices.
	Each such coupling imposes relationships between easily perceived degrees of freedom of the measurement device, with degrees of freedom in the measured system.
	This coupling determines the ``measurement basis'' of the system; however, the actually observed outcome is presented in the easily perceived degrees of freedom of the measurement device.
	Indeed, whatever one's interpretation of the process, the presentation of information about a system in easily perceived degrees of freedom is what makes for a useful ``measurement''.	
	We therefore feel justified in not providing an explicit role for ``measurement bases'' (or an explicit discussion of conditional distributions which should give rise to a notion of ``measurement collapse'') in the foundations of the general framework of modal distributions.
	}

\subsubsection{Specific state spaces of interest}
\label{sec:specialStateSpaces}

For the sake of concreteness, we now introduce state-spaces which correspond to randomized computation, to quantum computation, and to the models of computation which are the main subject of this article.
Following Refs.~\cite{SW10,WS11,HOSW11}, and also in analogy to probabilistic computation, we consider three different sorts of state-spaces in the case that $S$ is a finite ring:
\begin{definition}
	\label{def:generalizedGeneric}
  Let $S$ be a non-trivial ring.
  \begin{romanum}
  \item
		The \emph{generic state space} $\sS_\ast \subset \sD$ is the set of distributions $\ket{\psi} \in \sD$ for which there exists $\ket{\phi} \in \sD$ such that $\bracket{\phi}{\psi} = 1$; in other words, the set of distributions $\ket{\psi} \in \sD$ for which the ideal $I_\psi = \sum_x \psi_x S$ is the ring $S$ itself.\footnote{%
			Note that if $S$ is a field, then $\sS_\ast = \sD \setminus \{ \mathbf 0 \}$ is the set of states considered by Refs.~\cite{SW10,WS11}.
		}
	\item
		The \emph{$\ell_1$ state space} $\sS_1 \subset \sS_\ast$ is the set of distributions $\ket{\psi}$ for which $\sum\limits_x \psi_x = 1$.
	\item
		The \emph{$\ell_2$ state space} $\sS_2 \subset \sS_\ast$ is the set of distributions $\ket{\psi}$ for which $\bracket{\psi}{\psi} = 1$.
  \end{romanum} 
\end{definition}
\noindent
It is not difficult to show that $\mathscr S_\ast$, $\sS_1$, and $\sS_2$ are state-spaces: we prove this in Lemmas~\ref{lemma:genericStateSpace}--\ref{lemma:ell-2-stateSpace} (starting at page~\pageref{lemma:genericStateSpace}) in Appendix~\ref{apx:validTransformations}. 

The state-spaces described above motivate the study of three classes of operators on $\mathscr D$: invertible operators, \emph{affine operators} (operators which preserve the sum of the coefficients of the distributions they act upon), and \emph{unitary embeddings} (operators $U$ for which $U\herm U = \id_{\sD}$, which preserve inner products):
\begin{proposition}
\label{lemma:validTransformations}
Let $S$ be a Galois ring~\cite{Wan-2003}: a finite ring such that $p^r \cdot 1_S = 0_S$ for some prime $p$ and integer $r \ge 1$, and whose non-units are the set $p S$.
\begin{romanum}
\item
	The valid transformations of $\sS_\ast$ are all left-invertible transformations of $\mathscr D$.
\item
	The valid transformations of $\sS_1$ are all transformations $T: \mathscr D \to \mathscr D$ for which $\sum\limits_y \bra{y} T \ket{x} = 1_S$ for each $x \in \{0,1\}^\ast$.
\item
	The valid transformations of $\sS_2$ are all transformations $T: \sD \to \sD$ such that $T\herm T \equiv \id_\sD \pmod{p^{\lceil r/2 \rceil}}$ if $p$ is odd, or $T\herm T \equiv \id_\sD \pmod{2^{\lceil (r-1)/2 \rceil}}$ if ${p = 2}$.
	For $S$ any finite field or cyclic ring of odd order, we in fact have $T\herm T = \id_\sD$.
	Conversely, all operators $T: \sD \to \sD$ for which $T\herm T = \id_\sD$ are valid transformations of $\sS_2$.
\end{romanum}
\end{proposition}
\noindent
The proof of this proposition is technical, and is deferred to Appendix~\ref{apx:validTransformations} (Lemma~\ref{lemma:validTransformations-a} on page~\pageref{lemma:validTransformations-a}).
These families of operators preserve the state-spaces $\sS_\ast$, $\sS_1$, and $\sS_2$ respectively, and suggest related but distinct models of computation.

In the case of the unitary embeddings for $k > 2$, the valid transformations of $\sS_2$ may not all satisfy $T\herm T = \id_\sS$ in the Galois ring $S$, but only satisfy $T\herm T \equiv \id_\sS \pmod{\kappa}$ for a prime power $\kappa = p^t$.
However, the equivalence $T\herm T \equiv \id_\sS \pmod{\kappa}$ in the ring $S$ amounts to equality (of the equivalence classes) of $T\herm T$ and $\id_\sD$ in the smaller Galois ring $S' = S/\kappa S$.
Using a formalism of bounded error described in Sections~\ref{sec:bddErrorComputation} and~\ref{sec:reduceBddErrorToExact}, we may then use the unitary transformations of one Galois ring to simulate all valid $\sS_2$ transformations of a larger Galois ring.
Thus, even if not all valid transformations are unitary in some Galois rings $S$, they motivate the study of unitary transformations on other Galois rings $S'$.

\subsection{Modal circuits}
\label{sec:specialModalCircuitComplexity}

Any state-space $\sS$ (as defined in Definition~\ref{def:stateSpace}) is associated with a set of valid transformations which preserve it.
If the standard basis states $\ket{x}$ for $x \in \{0,1\}^\ast$ represent pieces of information, the valid transformations determine ways in which that information may be transformed.
A modal theory of computability describes how valid transformations of modal distributions may be decomposed into simpler, finitely-described transformations.

We now describe such decompositions of valid transformations, in analogy to circuit models for randomized computation and quantum computation.
We take the opportunity to describe conventions and gates of interest for our analysis, with the state-spaces of Definition~\ref{def:generalizedGeneric} in mind.
The material of this section is not substantially different from standard concepts of probabilistic computation or quantum computation.
Readers familiar with quantum computation can expect to be familiar the material in this section, and may skip ahead to Section~\ref{sec:modalCircuitComplexity} (page~\pageref{sec:modalCircuitComplexity}). 

\subsubsection{Modal computability}

For a fixed state-space $\sS$, we are interested in decomposing valid operations into ``primitive'' operations, to consider the computability and complexity of transformations of modal states.
We decompose them into operations of the following two sorts:
\begin{itemize}
\item
	\emph{Preparation operations}: injections $P: \sB\sox{n} \to \sB\sox{n+1}$ of the form $P \ket{\psi} = \ket{\psi}\ox\ket{\alpha}$ for some constant state $\ket{\alpha} \in \sB$\ifFullArticle, adjoining a bit prepared in some distribution which is independent of the input.
	We may represent these injections simply as state-vectors $\ket{\varphi} : \sB\sox{0} \to \sB\sox{1}$, in a tensor product with the identity operation acting on all other bits\fi\ (particularly $\ket{0}$, \ie~a ``fresh'' bit).
\item
	\emph{Bounded-arity transformations} (or \emph{gates}): a transformation $M: \sB\sox{h} \to \sB\sox{\ell}$ acting on an $h$-bit subsystem of the entire system, and producing an $\ell$-bit system as output.
	When considering $M$ among a collection $\{T_1, T_2, \ldots \}$ of transformations, we require that a finite representation of each coefficient $\bra{y} M \ket{x}$ be computable from $x \in \{0,1\}^h$, $y \in \{0,1\}^\ell$, and the index $j$ such that $M = T_j$.
\end{itemize}
\noindent
These operations compose to form more complex transformations $C: \sB\sox{n} \to \sB\sox{N}$.

The tensor product allows us to describe primitive operations $M_1, M_2, \ldots, M_m$ performed in parallel, when the sets of bits $\rA_1, \rA_2, \ldots, \rA_m$ on which they act are disjoint.
This is done by taking the tensor product of the operators $M_j$ (together with the identity operator acting on those bits not affected by the operations $M_j$).
In particular, if $\idop := \id_{\sB}$ is the identity operator on a single bit, we allow transformations $M: \sB\sox{h} \to \sB\sox{\ell}$ to be performed on contiguous subsets of bits by considering the tensor product operator $(\idop \ox \cdots \ox \idop \ox M \ox \idop \ox \cdots \ox \idop)$, for any finite number of identity operators on either side.
We may then multiply several such global operators to represent operations which are performed in sequence.

We describe compositions $C$ of the gates as \emph{$S$-modal circuits}, in analogy to boolean circuits and quantum circuits.

\subsubsection{Examples and notation}

For any semiring $S$, the gates of classical boolean circuit complexity preserve the $\sS_1$ state space.
Consider for instance {\small AND}, {\small OR}, and {\small NOT} gates, with fanout of bits explicitly represented by another gate {\small FANOUT}.
These gates act on $S$-modal distributions with the transformations
\begin{equation}
	\label{eqn:classicGates}
	\arraycolsep=2pt\def\arraystretch{0.9}
	\textrm{\small AND} = \mbox{\footnotesize$\left[\:\!\begin{matrix} 1 & 1 & 1 & 0 \\ 0 & 0 & 0 & 1 \end{matrix}\:\!\right]$}, ~~
	\textrm{\small OR} = \mbox{\footnotesize$\left[\:\!\begin{matrix} 1 & 0 & 0 & 0 \\ 0 & 1 & 1 & 1 \end{matrix}\:\!\right]$}, ~~
	\textrm{\small NOT} = \mbox{\footnotesize$\left[\:\!\begin{matrix} 0 & 1 \\ 1 & 0 \end{matrix}\:\!\right]$}, ~~\text{and}~~
	\textrm{\small FANOUT} = \mbox{\footnotesize$\left[\:\!\begin{matrix} 1 & 0 \\ 0 & 0 \\ 0 & 0 \\ 0 & 1 \end{matrix}\:\!\right]$},
\end{equation}
each acting on one or two bits.
\begin{subequations}
That is, we have
\begin{align}
		\textrm{\small AND} \bigl( a_0 \ket{00} + a_1 \ket{01} + a_2 \ket{01} + a_3 \ket{11} \bigr)
	&=
		(a_0 + a_1 + a_2) \ket{0} + a_3 \ket{1},
  \\[0.5ex]
		\textrm{\small OR} \bigl( a_0 \ket{00} + a_1 \ket{01} + a_2 \ket{01} + a_3 \ket{11} \bigr)
	&=
		a_0 \ket{0} + (a_1 + a_2 + a_3) \ket{1},
  \\[0.5ex]
		\textrm{\small NOT} \bigl( a \ket{0} + b \ket{1} \bigr)
	&=
		b \ket{0} + a \ket{1},
  \\[0.5ex]
		\textrm{\small FANOUT} \bigl( a \ket{0} + b \ket{1} \bigr)
	&=
		a \ket{00} + b \ket{11}.
\end{align}
\end{subequations}
From the above, we may see that these are affine transformations.
We may use these to describe more complex operations: for example, the logical formula $f(x_1,x_2,x_3) = (x_1 \mathbin{\&} x_2) \vee (x_2 \mathbin{\&} x_3)$ corresponds to an operator $F: \sB\sox{3} \to \sB\sox{1}$ expressed by
\begin{equation}
	\label{eqn:simpleComposition}
	\arraycolsep=2pt\def\arraystretch{0.9}
  F = 
	\mbox{\footnotesize$\left[\:\!\begin{matrix} 1 & 1 & 1 & 0 & 1 & 1 & 0 & 0 \\ 0 & 0 & 0 & 1 & 0 & 0 & 1 & 1 \end{matrix}\:\!\right]$}
	= \textrm{\small OR}\,\bigl(\textrm{\small AND} \ox \textrm{\small AND})\, (\idop \ox \textrm{\small FANOUT} \ox \idop).
\end{equation}
For a semiring $S$ in which $2$ is invertible (such as $\R_+$ or $\Z_k$ for odd $k$), we may also consider a single-bit gate
\begin{equation}
	\label{eqn:UNIF}
	\arraycolsep=2pt\def\arraystretch{0.9}
	\textrm{\small UNIF} = \mbox{\footnotesize$\left[\:\!\begin{matrix} 2^{-1} & 2^{-1} \\ 2^{-1} & 2^{-1} \end{matrix}\:\!\right]$},   
\end{equation}
which maps any state to a uniform distribution over $\ket{0}$ and $\ket{1}$, essentially representing the outcome of a fair coin flip (albeit only by formal analogy when $S \not\subset \R_+$).
It follows that for any $\ket{\psi} \in \sB$,
\begin{equation}
  \textrm{\small FANOUT} \cdot \textrm{\small UNIF} \ket{\psi} = 2^{-1} \bigl( \ket{00} + \ket{11} \bigr),
\end{equation}
which is not a standard basis state, nor a tensor product.
This is a consequence of the fact that the two bits produced as output are correlated (perfectly correlated, in this case).
Further transformations may be performed independently on each bit, though the results of those transformations may remain correlated.

Note that $\textrm{\small AND}$, $\textrm{\small OR}$, and $\textrm{\small UNIF}$ are all non-invertible: a straightforward calculation shows that ${\textrm{\small AND}\bigl(\ket{01} - \ket{10}\bigr)} = {\textrm{\small OR}\bigl(\ket{01} - \ket{10}\bigr)} = \vec 0$ and ${\textrm{\small UNIF}\bigl(\ket{0} - \ket{1}\bigr) = \vec 0}$.
Thus they are (usually\footnote{%
	The one exception is in the case $S = \Z_2$, in which case the $\sS_1$ and $\sS_2$ state-spaces are the same, and valid transformations of $\sS_2$ may be non-invertible.
	})
not valid operations for the generic and $\ell_2$ state spaces.
For transformations $M: \sB\sox{h} \to \sB\sox{\ell}$ to be invertible, we require that $h \le \ell$.
To simplify the study of transformations of $\sS_\ast$ and $\sS_2$ in this article, we limit ourselves to bijective gates $M: \sB\sox{h} \to \sB\sox{h}$ for $h \ge 1$.

Using bijective gates, we may simulate the traditional classical gates of Eqn.~\eqref{eqn:classicGates} via elementary techniques of reversible computation~\cite{Toffoli-1980,NC}, as follows.
We define single-bit \text{\small NOT} gates, two-bit \text{\small CNOT} gates, and three-bit \textsc{\small TOFFOLI} gates, which realize the following transformations of standard basis states:
	\begin{subequations}
	\label{eqn:CNOTs}
	\begin{align}
		\textrm{\small NOT} \ket{a} &= \ket{1-a}, \\
		\textrm{\small CNOT} \ket{c}\ket{a} &= \ket{c}\ket{a \oplus c},	\quad\text{ and} \\
		\textrm{\small TOFFOLI}\ket{c}\ket{b}\ket{a} &= \ket{c}\ket{b}\ket{a \oplus bc}   
	\end{align}
	\end{subequations}
	for $a,b,c \in \{0,1\}$, where $\oplus$ here denotes the XOR operation.
	Extending Eqn.~\eqref{eqn:CNOTs} linearly, their actions on arbitrary distributions $\ket{\psi} \in \sD$ are given by the (block) matrices
	\begin{equation}
	  \textrm{\small NOT} = X = \mbox{\small$\left[\begin{matrix} 0 \!&\! 1 \\ 1 \!&\! 0 \end{matrix}\right]$},
	  \quad\!
	  \textrm{\small CNOT} = \mbox{\small$\left[\begin{matrix} \idop \!&\! 0 \\ 0 \!&\! X\end{matrix}\right]$},
	  \quad\!
	  \textrm{\small TOFFOLI} = \mbox{\small$\left[\begin{matrix} \idop \!&\! 0 \!&\! 0 \!&\! 0 \\ 0 \!&\! \idop \!&\! 0 \!&\! 0 \\ 0 \!&\! 0 \!&\! \idop \!&\! 0 \\ 0 \!&\! 0 \!&\! 0 \!&\! X \end{matrix}\right]$}
	\end{equation}
	over the standard basis.
	We may simulate the gates of Eqn.~\eqref{eqn:classicGates} using the gates of Eqn.~\eqref{eqn:CNOTs} and preparation operations (in the case of $\textrm{\small AND}$ and $\textrm{\small OR}$, by producing the outputs of the classical logic gates as the final output bit):
	\begin{subequations}
	\begin{align}
			\ket{a}\ket{b}\ox\Bigl[\textrm{\small AND} \ket{a}\ket{b}\Bigr] &= \Bigl[\textrm{\small TOFFOLI}\,\bigl(\idop \ox \idop \ox \ket{0}\bigr)\Bigr]\ket{a}\ket{b},
		\\[1ex]
			\ket{a}\ket{b}\ox\Bigl[\textrm{\small OR} \ket{a}\ket{b}\Bigr] &= \Bigl[\bigl(X \ox X \ox X\bigr)\:\!\textrm{\small TOFFOLI}\,\bigl(X \ox X \ox \ket{0}\bigr)\Bigr]\ket{a}\ket{b},
	  \\[1ex]
			\textrm{\small FANOUT}\ket{a} &= \Bigl[\textrm{\small CNOT}\,\bigl(\idop \ox \ket{0}\bigr) \Bigr]\ket{a}.
	\end{align}
	\end{subequations}
	If we consider the final bit of a reversible circuit to compute the output of a boolean function $c: \{0,1\}^n \to \{0,1\}$ for some $n \ge 1$, the above equations allow us to simulate classical boolean circuits by reversible circuits.
	(We present an example in the next Section.) 
	
	For the $\ell_1$ state space, the map $\textrm{\small ERASE}: \sB\sox{1} \to \sB\sox{0}$ given by $
	  \textrm{\small ERASE} = \bigl[ 1 \;\;\; 1 \bigr]
	$
	is also a valid transformation. 
	We call this the \emph{erasure gate}, as $(\textrm{\small ERASE} \ox \idop) \ket{a}\ket{b} = \ket{b}$ for all $b \in \{0,1\}$.
	Then we may decompose the gates of Eqn.~\eqref{eqn:classicGates} exactly as
	\begin{subequations}
	\begin{align}
		\textrm{\small AND} &= \bigl(\textrm{\small ERASE} \ox \textrm{\small ERASE} \ox \idop\bigr)\,\textrm{\small TOFFOLI}\,\bigl(\idop \ox \idop \ox \ket{0}\bigr);
	\\
		\textrm{\small OR} &= \bigl(\textrm{\small ERASE} \ox \textrm{\small ERASE} \ox X\bigr)\,\textrm{\small TOFFOLI}\,\bigl(X \ox X \ox \ket{0}\bigr);
	\\
		\textrm{\small FANOUT} &= \textrm{\small CNOT} \bigl( \idop \ox \ket{0} 
		\bigr);
	\\
			\textrm{\small NOT} &= X.
	\end{align}
	\end{subequations}
	Thus, when considering transformations of the $\ell_1$ state space as well, we limit ourselves to gates $M: \sB\sox{h} \to \sB\sox{h}$ with the same number of input and output bits (with the exception of $\textrm{\small ERASE}$).
	
Operations on non-contiguous sets of bits may be allowed if one of the valid transformations is the two-bit {\small SWAP} operation, which re-orders a pair of consecutive bits:
\begin{equation}
	\label{eqn:SWAP}
	\arraycolsep=4pt
  \textrm{\small SWAP} = \mbox{\footnotesize$\left[\begin{matrix} 1 & 0 & 0 & 0 \\ 0 & 0 & 1 & 0 \\ 0 & 1 & 0 & 0 \\ 0 & 0 & 0 & 1 \end{matrix}\right]$}.
\end{equation}
We may simulate a gate $G$ acting on a non-contiguous set of bits, by applying a suitable permutation $P$ which puts the desired bits in a contiguous block, performing $G$, and then undoing the permutation $P$.
Any permutation of bits may be generated by transpositions of pairs of bits.
Then, given a set of primitive operations in a model of computation, we may perform those permutations provided that \text{\small SWAP} can be decomposed into those primitive operations.
While there are models of computation in which the \text{\small SWAP} operation is not a valid transformation of states~\cite{JM-2013}, as a permutation operator it preserves the $\sS_\ast$, $\sS_1$, and $\sS_2$ state-spaces, and so it is a valid operation which we include among our primitive operations with $\textrm{\small NOT}$, $\textrm{\small CNOT}$, and $\textrm{\small TOFFOLI}$.
Together, these gates may simulate any boolean formula~\cite{NC,Toffoli-1980}.

\subsubsection{Circuit diagrams and implicit SWAP operations}

It is helpful to depict modal circuits with diagrams.
We draw these in a way similar to classical logic circuits, with wires representing individual bits.
By convention, we draw them with the inputs on the left and outputs on the right: the order of the bits from top to bottom correspond to the order of the tensor factors in our equations, from left to right.
In the diagrams, gates are represented by labelled boxes (or symbols) placed on the bits on which they act.

As an example, the circuit $F: \sB\sox{3} \to \sB\sox{1}$ described in Eqn~\eqref{eqn:simpleComposition} is illustrated in the top portion of Figure~\ref{fig:simpleDiagrams} (page~\pageref{fig:simpleDiagrams}).
An equivalent circuit using only bijections $M: \sB\sox{h} \to \sB\sox{h}$ is depicted below it.
\begin{subequations}%
\label{eqn:CNOTs-diagrams}%
\allowdisplaybreaks
We visually represent the gates by
\begin{align}
  \textrm{\small NOT}: &&&
  \begin{aligned}
  \begin{tikzpicture}
    \draw (0,0) -- (1,0);
    \node [anchor=west] at (1,0) {$\ket{1-a}$};
    \node [anchor=east] at (0,0) {$\ket{a}$};
    \draw [black] (0.5,0) circle (4pt);
    \draw ($(0.5,0) + (0,4pt)$) -- ($(0.5,0) + (0,-4pt)$);
  \end{tikzpicture}
  \end{aligned}
\\[1ex]
  \textrm{\small CNOT}: &&&
	\begin{aligned}
  \begin{tikzpicture}
    \draw (0,0) -- (1,0);
    \draw (0,0.5) -- (1,0.5);
    \node [anchor=east] at (0,0) {$\ket{a}$};
    \node [anchor=east] at (0,0.5) {$\ket{c}$};
    \node [anchor=west] at (1,0) {$\ket{a \oplus c}$};
    \node [anchor=west] at (1,0.5) {$\ket{c}$};
    \filldraw [black] (0.5,0.5) circle (2pt);
    \draw [black] (0.5,0) circle (4pt);
    \draw (0.5,0.5) -- ($(0.5,0) + (0,-4pt)$);
  \end{tikzpicture}
	\end{aligned}
\\[1ex]
  \textrm{\small TOFFOLI}: &&&
	\begin{aligned}
  \begin{tikzpicture}
    \draw (0,0) -- (1,0);
    \draw (0,0.5) -- (1,0.5);
    \draw (0,1) -- (1,1);
    \node [anchor=east] at (0,0) {$\ket{a}$};
    \node [anchor=east] at (0,0.5) {$\ket{b}$};
    \node [anchor=east] at (0,1) {$\ket{c}$};
    \node [anchor=west] at (1,0) {$\ket{a \oplus bc}$};
    \node [anchor=west] at (1,0.5) {$\ket{b}$};
    \node [anchor=west] at (1,1) {$\ket{c}$};
    \filldraw [black] (0.5,1) circle (2pt);
    \filldraw [black] (0.5,0.5) circle (2pt);
    \draw [black] (0.5,0) circle (4pt);
    \draw (0.5,1) -- ($(0.5,0) + (0,-4pt)$);
  \end{tikzpicture}
	\end{aligned}
\end{align}
where the $\oplus$ denotes negation on the final bit conditioned (for standard basis states) on the bits marked with dots being in the state $\ket{1}$.
When we wish to represent a gate such as those of Eqn.~\eqref{eqn:CNOTs-diagrams} but which acts on non-consecutive bits, we allow the vertical connecting line to cross over any bits not involved (using the dots and $\oplus$ symbols to mark which bits the gate acts on).
This is illustrated by the first three gates in the bottom circuit of Figure~\ref{fig:simpleDiagrams}.
When we need to represent a $\textrm{\small SWAP}$ operation (or other permutation) explicitly, we do this by drawing crossing wires representing the interchanged bits.
\end{subequations}

It will be convenient to consider circuits which perform negation operations similar to $\textrm{\small CNOT}$ and $\textrm{\small TOFFOLI}$, but conditioned on more than two bits.
We define $\Lambda^{\!\ell} \!X: \sB\sox{\ell+1} \to \sB\sox{\ell+1}$, which performs the transformation
\begin{equation}
	\label{eqn:multiply-controlled-not}
	\Lambda^{\!\ell} \!X \ket{c_1}\ket{c_2}\cdots\ket{c_\ell}\ket{a} \;\,=\,\; \ket{c_1}\ket{c_2}\cdots\ket{c_\ell}\ket{a \oplus c_1 c_2 \cdots c_\ell}.
\end{equation}
For example, $\textrm{\small CNOT} = \Lambda^{\!1} \! X$ and $\textrm{\small TOFFOLI} = \Lambda^{\!2} \! X$.
Using standard techniques~\cite{Toffoli-1980,NC}, we may simulate these recursively by $\textrm{\small TOFFOLI}$ gates, using additional work-space.
We make use of such gates, and depict them similarly to $\textrm{\small TOFFOLI}$ gates,
\begin{equation}
  \Lambda^{\!\ell} \!X \;\equiv\;
	\begin{aligned}
  \begin{tikzpicture}
    \draw (0,0) -- (1,0);
    \draw (0,0.5) -- (1,0.5);
    \draw (0,1.5) -- (1,1.5);
    \draw (0,2) -- (1,2);
    \filldraw [black] (0.5,2) circle (2pt);
    \filldraw [black] (0.5,1.5) circle (2pt);
    \filldraw [black] (0.5,0.5) circle (2pt);
    \draw [black] (0.5,0) circle (4pt);
    \draw (0.5,2) -- ($(0.5,0) + (0,-4pt)$);
    \node at (0.125,1.1) {$\vdots$};
    \node at (0.875,1.1) {$\vdots$};
    \node at (1.25,1.25) [anchor=west] {\;$\mathllap{\left.\begin{array}{c}\\[9ex]\end{array}\right\}}$ $\ell$ bits.};
  \end{tikzpicture}
	\end{aligned}  
\end{equation}
For the sake of simplicity, this diagram does not depict the additional work space used to implement the $\Lambda^{\!\ell} \!X$ gate from $\textrm{\small TOFFOLI}$ gates.
We can account for the ancilla bits which are used if necessary, but this will not affect the asymptotic measures of circuit size which we consider.

\begin{figure}[t]
	\begin{center}
  	\begin{tikzpicture}[scale=0.75]
	  \node [draw,
						isosceles triangle, isosceles triangle apex angle=45, shape border rotate=180,
						minimum height=4ex, inner sep=3pt, fill=white] (t) at (1.5,0.75) {};
	  \node [draw, and gate US, inner sep=1pt] (u) at (3,1.5) {}; 
	  \node [draw, and gate US, inner sep=1pt] (v) at (3,0) {}; 
	  \node [draw, or gate US, inner sep=1pt] (w) at (4.5,0.75) {}; 
	  
	  \node (a2) at (-0.25,0.75) {\small $\ket{x_2}$};
	  \node (a1) at ($(a2 |- u.input 1)$) {\small $\ket{x_1}$};
	  \node (a3) at ($(a2 |- v.input 2)$) {\small $\ket{x_3}$};
	  \node (omega) [anchor=west] at (6,0.75) {\small $\ket{(x_1 \& x_2) \vee (x_2 \& x_3)}$};

		\draw [-, out=0, in=180] (a1) to (u.input 1);
		\draw [-, out=0, in=180] (a2) to (t);
		\draw [-, out=0, in=180] (a3) to (v.input 2);
		\draw [-, out=0, in=180] (t.25) to (u.input 2);
		\draw [-, out=0, in=180] (t.335) to (v.input 1);
		\draw [-, out=0, in=180] (u.output) to (w.input 1);
		\draw [-, out=0, in=180] (v.output) to (w.input 2);
		\draw [-, out=0, in=180] (w.output) to (omega);
	\end{tikzpicture}
	\\[4ex]
	\begin{tikzpicture}
	  \coordinate (x1-0) at (0,0);
	  \xdef\prev{1}
	  \foreach \j in {2,3,4,5,6,7} {%
			\coordinate (x\j-0) at ($(x\prev-0) + (0,-0.5)$);
			\xdef\prev{\j}
		}
		
		\xdef\prev{0}
		\foreach \t/\dt in {1/0.5,2/0.75,3/0.75,4/0.75,5/0.35,6/0.35,7/0.75} {
			\foreach \j in {1,2,3,4,5,6,7} {
				\coordinate (x\j-\t) at ($(x\j-\prev) + (\dt,0)$);
				\draw (x\j-\prev) -- (x\j-\t);
			}
			\xdef\prev{\t}
		}
		
		\node at (x1-0) [anchor=east] {\small $\ket{x_1}$};
		\node at (x2-0) [anchor=east] {\small $\ket{x_2}$};
		\node at (x3-0) [anchor=east] {\small $\ket{x_3}$};
		\foreach \j in {4,5,6,7}
			\node at (x\j-0) [anchor=east] {\small $\ket{0}$};

		\filldraw [black] (x2-1) circle (2pt);
		\draw [black] (x4-1) circle (4pt);
		\draw [black] (x2-1) -- ($(x4-1) + (0,-4pt)$);
		
		\filldraw [black] (x1-2) circle (2pt);
		\filldraw [black] (x2-2) circle (2pt);
		\draw [black] (x5-2) circle (4pt);
		\draw [black] (x1-2) -- ($(x5-2) + (0,-4pt)$);

		\filldraw [black] (x3-3) circle (2pt);
		\filldraw [black] (x4-3) circle (2pt);
		\draw [black] (x6-3) circle (4pt);
		\draw [black] (x3-3) -- ($(x6-3) + (0,-4pt)$);

		\draw [black] (x5-4) circle (4pt);
		\draw [black] (x6-4) circle (4pt);
		\draw [black] ($(x5-4) + (0,4pt)$) -- ($(x5-4) + (0,-4pt)$);
		\draw [black] ($(x6-4) + (0,4pt)$) -- ($(x6-4) + (0,-4pt)$);

		\filldraw [black] (x5-5) circle (2pt);
		\filldraw [black] (x6-5) circle (2pt);
		\draw [black] (x7-5) circle (4pt);
		\draw [black] (x5-5) -- ($(x7-5) + (0,-4pt)$);

		\draw [black] (x5-6) circle (4pt);
		\draw [black] (x6-6) circle (4pt);
		\draw [black] (x7-6) circle (4pt);
		\draw [black] ($(x5-6) + (0,4pt)$) -- ($(x5-6) + (0,-4pt)$);
		\draw [black] ($(x6-6) + (0,4pt)$) -- ($(x6-6) + (0,-4pt)$);
		\draw [black] ($(x7-6) + (0,4pt)$) -- ($(x7-6) + (0,-4pt)$);

		\node at (x1-7) [anchor=west] {\small $\ket{x_1}$};
		\node at (x2-7) [anchor=west] {\small $\ket{x_2}$};
		\node at (x3-7) [anchor=west] {\small $\ket{x_3}$};
		\node at (x4-7) [anchor=west] {\small $\ket{x_2}$};
		\node at (x5-7) [anchor=west] {\small $\ket{x_1 \& x_2}$};
		\node at (x6-7) [anchor=west] {\small $\ket{x_3 \& x_2}$};
		\node at (x7-7) [anchor=west] {\small $\ket{(x_1 \& x_2) \vee (x_3 \& x_2)}$};
	\end{tikzpicture}
	\end{center}
\caption{%
	Two circuits to evaluate $f(x_1,x_2,x_3) = (x_1 \mathbin\& x_2) \vee (x_2 \mathbin\& x_3)$, using different sets of gates/preparation operations.
	Lines running from left to right represent bits which are the inputs/outputs of gates in the circuit.
	Initial states (input values) of the bits are presented on the left, and final states (output values) are presented on the right.
	~\textbf{Top:}~A~boolean circuit involving traditional boolean logic gates.
	From left to right, the gates are $\textrm{\footnotesize FANOUT} : \sB \to \sB\sox{2}$, $\textrm{\footnotesize AND}: \sB\sox{2} \to \sB$, and $\textrm{\footnotesize OR}: \sB\sox{2} \to \sB$.
	~\textbf{Bottom:}~A~circuit involving only invertible operations $M: \sB\sox{h} \to \sB\sox{h}$,
	simulating the irreversible circuit above.
	The $\oplus$ symbols represent logical negation operations on bits, performed either unconditionally ({\footnotesize NOT} gates), or conditioned on one or two bits ({\footnotesize CNOT} and {\footnotesize TOFFOLI} gates).
	Control bits are denoted by solid dots, connected to the $\oplus$ symbol by a solid vertical line. 
	({\footnotesize SWAP} operations are used implicitly, to move the control and target operations adjacent to one another while leaving other bits unaffected.)
	The output $f(x_1, x_2, x_3)$ is computed on the bottom-most bit.	
}
\label{fig:simpleDiagrams}
\end{figure}
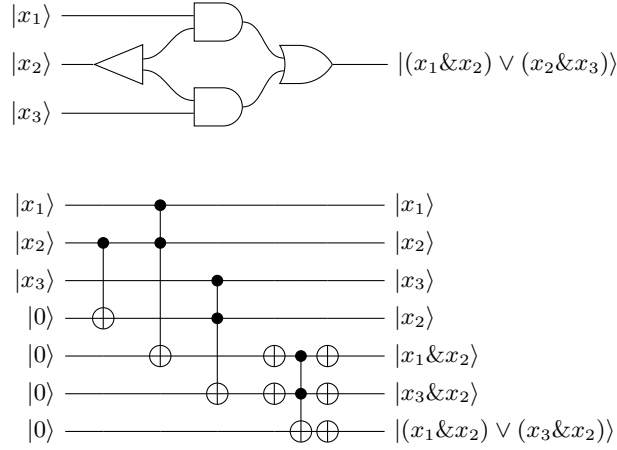

To simplify equations, we omit $\textrm{\small SWAP}$ operations (and tensor products with $\idop$) by using subscripts to denote which bits each operation acts on.
For instance, consider the action of the left-most operation in the bottom circuit of Figure~\ref{fig:simpleDiagrams} on the first four bits, performing the operation $\ket{x_1}\ket{x_2}\ket{x_3}\ket{x_4} \mapsto \ket{x_1}\ket{x_2}\ket{x_3}\ket{x_4 \oplus x_2}$ on standard basis states.
Rather than representing this explicitly as a decomposition such as
\begin{equation}
	(\idop \ox \textrm{\small SWAP} \ox \idop) (\idop \ox \idop \ox \textrm{\small CNOT}) (\idop \ox \textrm{\small SWAP} \ox \idop),
\end{equation}
we instead denote the same operation by $\textrm{\small CNOT}_{2,4}$, which denotes that it is a $\textrm{\small CNOT}$ operation acting only on the second and fourth tensor factor (in that order).

\subsubsection{A sample calculation}

As an example, we now compute the action of the bottom circuit of Figure~\ref{fig:simpleDiagrams}.
We would write the composition of operations in circuit as the product
\begin{equation}
	X_7 \;
	X_6 \;
	X_5 \;
	\textrm{\small TOFFOLI}_{5,6,7} \;
	X_6 \;
	X_5 \;
	\textrm{\small TOFFOLI}_{3,4,6} \;
	\textrm{\small TOFFOLI}_{1,2,5} \;
  \textrm{\small CNOT}_{2,4} \;
\end{equation}
ordered from right to left (acting as linear transformations of column vectors).
We then represent the effect of each operation in the circuit, as follows:
\begin{equation}
	\label{eqn:sampleCalculation}
	\mspace{-18mu}
  \begin{aligned}[b]
			|x_1, x_2, x_3, 0, 0, &0, 0 \rangle
    \\[0.5ex]\mapstoname{\textrm{CNOT}_{2,4}}{}&
			\ket{x_1, x_2, x_3, x_2, 0, 0, 0}
		\\\mapstoname{\textrm{TOFFOLI}_{1,2,5}}{}&
			\ket{x_1,\, x_2,\, x_3,\, x_2,\, x_1 \& x_2,\, 0,\, 0}
		\\\mapstoname{\textrm{TOFFOLI}_{3,4,6}}{}&
			\ket{x_1,\; x_2,\; x_3,\; x_2,\; x_1 \& x_2,\; x_3 \& x_2,\; 0}
		\\\mapstoname{X_6 X_5}{}&
			\ket{x_1,\big.\; x_2,\; x_3,\; x_2,\; \neg(x_1 \& x_2),\; \neg(x_3 \& x_2),\; 0}
		\\\mapstoname{\textrm{TOFFOLI}_{5,6,7}}{}&
			\ket{x_1,\Big.\; x_2,\; x_3,\; x_2,\; \neg(x_1 \& x_2),\; \neg(x_3 \& x_2),\; \neg(x_1 \& x_2) \mathbin\& \neg(x_3 \& x_2)}
		\\\mapstoname{X_7 X_6 X_5}{}&
			\ket{x_1,\Big.\; x_2,\; x_3,\; x_2,\; (x_1 \& x_2),\; (x_3 \& x_2),\; (x_1 \& x_2) \vee  (x_3 \& x_2)}.
  \end{aligned}
  \mspace{-57mu}
\end{equation}
If we performed this circuit on an input $\ket{\psi} \ket{0000}$ for some state $\ket{\psi} = a_{000} \ket{000} + a_{001} \ket{001} + \cdots + a_{111} \ket{111}$, the result would simply be a linear combination of standard basis states as in the final line of Eqn.~\eqref{eqn:sampleCalculation}, for $x_1 x_2 x_3 \in \{0,1\}^3$ running over all possible values.
We could then decompose the output as
\begin{subequations}
\begin{align}
  C \ket{\psi}\ket{0000} &= \ket{\varphi'}\ket{0} + \ket{\varphi''}\ket{1}, \quad \text{where}
  \\[3ex]
  \ket{\varphi'} \;&= \;\; \sum_{\mathclap{\substack{x \in \{0,1\}^3 \\ f(x) = 0}}} \; a_x \ket{x_1,\, x_2,\, x_3,\, x_2,\, x_1\&x_2,\, x_2\&x_3},
  \\[0.5ex]
  \ket{\varphi''} \;&= \;\; \sum_{\mathclap{\substack{x \in \{0,1\}^3 \\ f(x) = 1}}} \; a_x \ket{x_1,\, x_2,\, x_3,\, x_2,\, x_1\&x_2,\, x_2\&x_3}.
\end{align}
\end{subequations}
Such decompositions play a role in our description of how modal circuits are used to solve decision problems.

\subsection{Modal circuit complexity}
\label{sec:modalCircuitComplexity}

\subsubsection{Uniform circuit families and efficient exact computation}

To consider problems that may be ``efficiently'' solved by $S$-modal circuits, we consider families $\{C_n\}_{n \ge 1}$ of circuits $C_n: \sB\sox{n} \to \sB\sox{N}$, for some $N \in \N$ depending on $n$, with the following constraints:
\begin{romanum}
\item
	$\{C_n\}_{n \ge 1}$ is \emph{polynomial-time uniform}~\cite{AB-2009,Papadimitriou-1994}: there is a deterministic Turing machine which, on input $1^n$, computes the construction of $C_n$ as a composition of primitive gates and preparation maps in time $O(\poly n)$, where each gate/preparation map $T_\ell$ may be represented by its label $\ell \in \{0,1\}^\ast$.	
\item
	The gates and preparation maps of $\{ C_n \}_{n \ge 1}$ are \emph{polynomial-time specifiable}: there is a polynomial-time deterministic Turing machine which, for each gate or preparation map $T_\ell$ used in each circuit $C_n$, computes representations of all of the coefficients of $T_\ell$ in the standard basis in total time $\poly(|\ell|)$.
	We also require that there be polynomial-time bounded deterministic Turing machines which compute representations of $a + b \in S$ and $ab \in S$, and to decide $a \mathrel{?\!=} 0_S$, from representations of coefficients $a,b \in S$.
\end{romanum}
Constraint~\parit{ii} allows us to consider circuit families $\{ C_n \}_{n \ge 1}$, in which $C_n$ is completely representable in time $O(\poly n)$, but where the number of distinct gates used by $C_n$ may grow with $n$.\footnote{%
	For models of computation such as deterministic or randomized circuit families, or indeed the $\Z_k$-modal circuits which we study in this article, polytime-specifiable gate-sets may be simulated with only polynomial overhead by constant-sized gate-sets.
	However, we do not consider it extravagant to allow gates which can be specified in polynomial time.
	Furthermore, as we argue in Appendix~\ref{apx:EQPthesis},\label{discn:polytime-specifiable-subst} requiring constant-sized gate sets is an undue restriction on the study of \emph{exact} quantum algorithms.
}
One can study more limited circuit families, such as logspace-uniform circuits with logspace-specifiable gates (modifying the definitions above accordingly), or logspace-uniform circuits with \emph{constant-time} specifiable gates (\ie~a~finite gate set), \etc.

We associate a cost to each gate and preparation map, which bounds the time required to compute any of its coefficients:
\label{discn:gate-cost}
for circuits constructed from a finite gate-set, all gates have an equivalent cost under asymptotic analysis.
In many cases it may also be reasonable to specifically restrict preparation operations to a finite set, \eg~to preparation of the state $\ket{0}$.

We adopt the convention that circuit families $\{C_n\}_{n\ge 1}$ to solve decision problems produce their answer on the last bit of their output.
For a language $L$ and $x \in \{0,1\}^\ast$, we write $L(x) = 0$ for $x \notin L$, and $L(x) = 1$ for $x \in L$.
Thus we are interested in decomposing the final state $\ket{\psi_x}$ of a computation as
\begin{equation}
  C_n \ket{x} \;=\; \ket{\psi'}\ket{0} + \ket{\psi''}\ket{1},
\end{equation}
and deciding on the basis of the conditional distributions $\ket{\psi'}, \ket{\psi''} \in \sD$.
\begin{definition}
  Let $L$ be a language and $\{ C_n \}_{n \ge 1}$ be a circuit family.
	Then this circuit family \emph{efficiently decides $L$ exactly} if $\{C_n\}_{n \ge 1}$ is polynomial-time uniform with polynomial-time specifiable gates, and the final bit of $C_n \ket{x}$ is necessarily $L(x)$ for each $x \in \{0,1\}^n$ (that is, $C_n\ket{x}= \ket{\varphi}\ket{L(x)}$ for some $\ket{\varphi} \in \sS$).
\end{definition}
\noindent

\subsubsection{Efficient bounded-error computation}
\label{sec:bddErrorComputation}

Exact computation is a restrictive condition for certain semirings $S$.
For instance, in the case $S = \R_+$ of randomised computations, problems which are exactly solvable are merely those in $\P$: the probabilities can play essentially no role in exact algorithms.
On the other hand, we are not interested in circuit families $\{ C_n \}_{n \ge 1}$ such that $C_n \ket{x} = \ket{\varphi'}\ket{0} + \ket{\varphi''}\ket{1}$, where $\ket{\varphi''} \ne \vec 0 \,{\iff}\, x \in L$.
This corresponds to determining whether the output $1$ is impossible or merely possible, which 
we interpret 
as decision with unbounded error. 

We may formulate a theory of bounded-error computation over arbitrary semirings through ``significance'' functions $\{0,1\}: S \to \R_+$, similar to metrics or absolute value functions,\footnote{%
	The usual properties of absolute values, such as $\sigma(st) = \sigma(s)\sigma(t)$, would imply for Galois rings that $\sigma(s) = 0$ for all zero divisors $s$.
	We see no reason to require this to be the case, and so define ``significance functions'' to be more general than absolute value functions.
}
which distinguish the significance of various amplitudes.
\begin{definition}
  Let $S$ be a semiring. A \emph{significance function} $\sigma: S \to \R_+$ is a function such that $\sigma(0_S) = 0$, $\sigma(1_S) = 1$, and which is \emph{monotone}: that is, for all $s,t,u \in S$, ${\sigma(u) \le \sigma(t)} \implies {\sigma(su) \le \sigma(st)}$.
  A significance function $\sigma$ is \emph{effective} if there is a polynomial-time deterministic Turing machine which decides,
	from representations of $u \in S$ and $a \in \R_+$, which of $\sigma(u) < a$, $\sigma(u) = a$, or $\sigma(u) > a$ holds. 
\end{definition}
\noindent 
One may consider a computational outcome $\ket{\varphi'}\ket{0} + \ket{\varphi''}\ket{1}$ to have bounded error, by describing conditions for the conditional distributions $\ket{\varphi'}, \ket{\varphi''} \in \sD$ to be ``insignificant''.
In particular, any transformation of an insignificant distribution should again be insignificant.
This motivates the following definition:
\begin{definition}
	\label{def:bddErrorDecision}
  Let ${0 \le a < b \le 1}$ be real constants, $S$ a semiring, $\sigma: S \to \R_+$ be an effective significance function, $L$ be a language, and $\{ C_n \}_{n \ge 1}$ be an $S$-modal circuit family.
	Then this circuit family \emph{efficiently decides $L$ with $\sigma$-error bounds $(a,b)$} if $\{C_n\}_{n \ge 1}$ is polynomial-time uniform with polynomial-time specifiable gates, and for any $x \in \{0,1\}^n$, we have
$		C_n \ket{x} \;=\; \ket{\varphi'}\ket{1{-}L(x)} + \ket{\varphi''}\ket{L(x)}, $
	such that:
	\begin{romanum}
	\item
		for every $y \in \{0,1\}^\ast$ and valid transformation $T$, we have $\sigma\bigl(\bra{y}T\ket{\varphi'}\bigr) \le a$;
	\item
		there is a $y \in \{0,1\}^\ast$ and a valid transformation $T$\!, such that $\sigma\bigl(\bra{y}T\ket{\varphi''}\bigr) \ge b$.	
	\end{romanum}
\end{definition}
\noindent
If such a bounded-error algorithm is a subroutine of some other procedure, final outcomes which depend on an incorrect result from the subroutine will be ``insignificant'' (in the sense that its significance can never surpass the threshold $a$), 
while correct results of a subroutine always admit a way to produce a ``significant'' amplitude.

\subsubsection{Examples}
\label{sec:modalComputationConventionalExamples}

For the sake of concreteness, we illustrate how these definitions of ``exact'' and ``bounded error'' computation are realised in conventional models of computation.

\vspace*{-1ex}
\paragraph{Deterministic computation.}
	Polynomial-time deterministic computation represents a trivial case of modal computation.
	Consider any semi-ring $S$, with a state-space $\sS_{\delta}$ consisting of standard basis states $\ket{x}$ for $x \in \{0,1\}^\ast$.
	This is a state-space for which the valid transformations are all permutation operations.
	The circuit-families over $\sS_{\delta}$ which efficiently solve decision problems are then polytime-uniform circuits composed of permutations which are themselves computable in polynomial time.
	These circuits can be simulated in $\P$ simply by simulating the action of each gate on standard basis states; and conversely, uniform circuit families over boolean logic gates can decide all problems in $\P$.
	By standard arguments in circuit complexity, the same holds for logspace-uniform circuit families transforming $\sS_{\delta}$ using logspace-specifiable gates.	

\vspace*{-1ex}
\paragraph{Randomized computation.}
	We may recover $\BPP$ as the decision problems which are solvable with bounded $\sigma$-error using efficient computation on the $\sS_1$-state space over $\R_+$, where $\sigma(u) = u$.	
	As $\max_{y,T} \sigma(\bra{y} T \ket{\varphi}) = \| \varphi \|_1$, the ``significance'' of a conditional distribution is simply its probability.
	A randomized algorithm may simulate a $\R_+$-modal gate, by computing the distribution of each possible transition, and then sampling from that distribution.
	The sampling process may be approximate (\eg~if some gate coefficients are irrational): it suffices to take a rational approximations to within some small $\varepsilon$ for each coefficient, such that the approximate probabilities add to $1$.
	(These approximations can be computed in time $O(\poly \log(1/\varepsilon))$, as $\sigma$ is effective.)
	A circuit may be simulated, by simulating each gate in turn, and producing the appropriate output bit.
	If there are $T$ gates in the circuit, and there are $N \in O(\poly n)$ bits in the circuit, this simulation yields a distribution which differs from the output distribution of the circuit by at most $2^N T\varepsilon$.
	If our approximations are precise enough that $\varepsilon \in o(1/T 2^N)$, the probabilities of obtaining either output $0$ or $1$ differ from the circuit by at most $o(1)$.
	Thus a polytime-uniform, polytime-specifiable circuit which decides some language $L$ with $\sigma$-error bounds $(\tfrac{1}{4},\tfrac{3}{4})$ may be simulated by a randomized Turing machine in polynomial time with error bounds $(\tfrac{1}{3},\tfrac{2}{3})$.
	Conversely, polytime-uniform $\R_+$-modal circuit families can simulate randomized logic circuits: thus the class decided by such ``modal circuits'' is simply $\BPP$.
	We may similarly characterize $\RP$ or $\co\RP$ by imposing further restrictions on output amplitudes from the modal circuits, for \YES\ or for \NO\ instances.

\vspace*{-1ex}
\paragraph{Quantum computation.}
The unitary circuit model consists of polytime-uniform families of $\C$-modal circuits which preserve the $\ell_2$-state space, constructed from a finite (\ie~constant-time specifiable) gate set .
The definitions of $\EQP$ and $\BQP$~\cite{BV97} are then equivalent to the classes of decision problems which can be decided by such circuits, respectively exactly or with constant probability of error.
We may recover $\BQP$ using polytime-specifiable gates,\footnote{%
	The definition of \EQP\ is unlikely to be similarly robust to a change to polytime-specifiable gate sets.
	We argue in Appendix~\ref{apx:EQPthesis}\label{discn:EQP-nonrobust} that, in fact, a broader definition of efficient exact quantum computation is warranted.
}
with $\sigma$-error bounds $(\tfrac{1}{3},\tfrac{2}{3})$
for $\sigma(u) = |u|^2$, as follows.

A gate set $\mathcal G \subset \mathrm{SU}(2^N)$ is \emph{approximately universal} if products of $T_\ell \in \mathcal G$ (and tensor products with $\idop$) generate a dense subgroup of $\mathrm{SU}(2^M)$ for any $M \ge N$.
By the Solovay--Kitaev theorem~\cite{KSV-2002,DN2006}, any approximately universal unitary gate set which is closed under inverses may efficiently approximate any $U \in \mathrm{SU}(2^L)$, in that there are $W_1, W_2, \ldots, W_T \in \mathcal G$ which can be discovered in time $O(\poly(2^L) \poly \log(1/\varepsilon))$, such that $\| U - W_1 W_2 \cdots W_T \|_\infty < \varepsilon$.
Any polytime-specifiable gate $U \in \mathrm{U}(2^L)$ acts on at most $L \in O(\log n)$ qubits, and is proportional to an operator $U' \in \mathrm{SU}(\poly n)$.
The ``significance'' of a conditional distribution is its probability with respect to the Born rule: that is, $\max_{y,T} \sigma(\bra{y} T \ket{\varphi}) = \| \varphi \|_2^2$\,.
As unit-modulus proportionality factors make no difference to the significance of the outcomes, we may simulate $U$ simply by substituting it with $U'$, which we may simulate (approximately but efficiently) using any approximately universal gate-set by Solovay--Kitaev.

Thus, as with randomized computation, the bounded-error conditions for the modal circuits correspond to bounded error conditions for quantum algorithms; with minor changes to the error bounds, we may simulate any polytime-uniform, polytime-specifiable circuit on the $\ell_2$-state space of $\C$ by a bounded-error quantum algorithm.

\medskip
The above illustrates how several existing notions of bounded-error computation may be described in the framework of modal computation, and hopefully convinces the reader that this framework is well-formulated.
The task of the following Section is to similarly apply this framework to distributions over $\Z_k$ and other Galois rings.

\section{Modal computation on Galois rings}
\label{sec:modalComputationGaloisRings}

Using the general framework of the previous Section, we now present the computational model of our main results: transformations of Galois-ring-valued distributions.
Our motivation is that Galois rings include the special case of fields $\F_k$ considered by Refs.~\cite{SW10,WS11,SW-2012}, as well as the more familiar cyclic rings $\Z_k$ for prime powers $k$.
We present the complexity classes motivated by the definitions of the preceding Section, and state the main results of the article, to be proven in Section~\ref{sec:characterization}.
\pagebreak

Throughout this Section, $k = p^r$ denotes some power of a prime $p$.
We consider mainly the case of cyclic rings $\Z_k$, which includes fields of prime order when $r = 1$.
We then indicate how our results extend to Galois rings in general.

\subsection{Elementary modal gate sets for cyclic rings}

For each state-space $\sS_\ast$, $\sS_1$, and $\sS_2$ as in Definition~\ref{def:generalizedGeneric}, we consider circuits involving preparation of $\ket{0} \in \sB$ and gates acting on at most four bits. 
The latter are represented by matrices over $\Z_k$ of shape ${2 \!\x\! 2}$, ~${4 \!\x\! 4}$, ~${8 \!\x\! 8}$, or ${16 \!\x\! 16}$.
For each of the state-spaces we consider, there are only finitely many valid operations on four or fewer bits, including all of the gates of Eqn.~\eqref{eqn:CNOTs} as well as the \textrm{\small SWAP} gate.
While some smaller gate sets (for each of $\sS_\ast$, $\sS_1$, or $\sS_2$) may generate the same sets of transformations, allowing all gates on four or fewer bits yields at most a constant factor advantage.

Our results do not require the ability to simulate arbitrary valid operations.
However, any valid transformation of $\sS_\ast$, $\sS_1$, or $\sS_2$ can indeed be simulated using four-bit gates.
For instance, invertible gates on four or fewer bits (together with preparation of the state $\ket{0}$) can simulate any valid transformation $T: \sB\sox{n} \to \sB\sox{n}$ of $\sS_\ast$, in that they can be composed to perform an operation $M: \sB\sox{n+m} \to \sB\sox{n+m}$ such that
\begin{equation}
  M \ket{x}\ket{0^m} \;=\; \bigl(T\ket{x}\bigr)\ox\ket{0^m} \quad\text{for all $x \in \{0,1\}^n$}.
\end{equation}
Such an operator $M$ can be found by simulating simple Gaussian elimination on $T$, following the techniques described in Ref.~\cite[Lemma~8.1.4]{KSV-2002} to obtain a decomposition into \textrm{\small TOFFOLI} gates and two-qubit invertible gates.
The gates $\textrm{\small NOT}$, $\textrm{\small CNOT}$, $\textrm{\small TOFFOLI}$, and $\textrm{\small SWAP}$ are used to single out a pair of standard basis states $\ket{a}, \ket{b}$ to act upon for $a,b \in \{0,1\}^n$, performing a permutation of the standard basis on $n+1$ bits such that $\ket{a}\ket{0} \mapsto \ket{11\cdots10}\ket{1}$ and $\ket{b}\ket{0} \mapsto \ket{11\cdots11}\ket{1}$.
One may then simulate an elementary row-operation between the rows $a$ and $b$ by performing a suitable operation on the final two bits, and then undoing the permutation of the standard basis states.
The row-operation is chosen to yield a matrix with no coefficient in row $b$ and column $a$, realising a single step of Gaussian elimination; the result follows by simulating the entire Gaussian elimination.
This result is technical, and a straightforward modification of known results for quantum circuits; similar techniques exist to decompose affine operations or unitary operations into gates acting on four or fewer bits.\footnote{%
	In the case of affine operations in particular, a few minor changes in the analysis are required to the approach of Ref.~\cite[Lemma~8.1.4]{KSV-2002}, as not all affine operations are invertible.
	The details are straightforward, but we omit them here, as they do not bear on our results.
}

One of our main results is to characterise the power of unitary circuits acting on $\Z_k$-distributions (Section~\ref{sec:simulatingModkPbyModal}).
A corollary of these results is that there is a set of eight gates which suffice to efficiently simulate any polynomial-time uniform circuit with polynomial-time specifiable gates, consisting of invertible or unitary transformations on $\Z_k$-modal distributions.
A different set of eight gates suffice to simulate polynomial-time uniform circuits with polynomial-time specifiable gates over $\Z_k$.


\subsection{Bounded error reduces to exact computation for $\Z_{p^r}$}

\label{sec:signFns}

Consider the special case of a prime-order cyclic ring $\Z_p$ (taking $k = p^1$).
We may show that there is a unique choice of significance function $\sigma_p : \Z_p \to \R_+$:
\begin{equation}
  \sigma_p(s) = \begin{cases}
                  1,	&	s \ne 0; \\
                  0,	& s = 0.
                \end{cases}
\end{equation}
The uniqueness of this function is due to the fact that $\Z_p \setminus \{0\}$ (written $\Z_p^\times$) is a finite multiplicative group.
Every element $a \in \Z_p^\times$ has order at most $p-1$: then if $\sigma(a) \le \sigma(1)$ for some significance function $\sigma : \Z_p \to \R_+$ and $a \in \Z_p^\times$, it follows that
\begin{equation}
  \sigma(1) \ge \sigma(a) \ge \sigma(a^2) \ge \cdots \ge \sigma(a^{p-1}) = \sigma(1) = 1
\end{equation}
by monotonicity; and similarly if $\sigma(a) \ge \sigma(1)$.
As $\sigma(0) = 0$ and $\sigma(1) = 1$ by definition, we then have $\sigma = \sigma_p$.

For cyclic rings $\Z_k$ for prime powers $k = p^r$ where $r > 1$, there is more than one significance function.
However, all of the significance functions can be related to a canonical significance function which generalises $\sigma_p$.
\begin{lemma}
  Let $k = p^r$ for $r > 1$ and a prime $p$.
  If $\sigma: \Z_k \to \Z_+$ is a significance function, there is a strictly increasing function $f: \R_+ \to \R_+$ and an integer $1 \le \tau \le r$ such that $\sigma(s) = f(\sigma_{p^\tau}(s))$, where
  \begin{equation}
    \sigma_{p^\tau}(s) = \begin{cases}
														1/p^t,	&	\text{$s = p^t a$ for a multiplicative unit $a \in \Z_k^\times$ and $t < \tau$};	\\
														0,			&	\text{if $s \in p^\tau \Z_k$}.
                         \end{cases}    
  \end{equation}
\end{lemma}
\begin{proof}
  It is easy to show that $\sigma(s) = 1$ for all multiplicative units $s \in \Z_k^\times$, for the same reason as for $\sigma_p$ above.
  More generally, for any $t \ge 0$ and multiplicative unit $a$, we have $\sigma(p^t) \ge \sigma(p^t a) \ge \sigma(p^t)$.
  As every element of $\Z_k$ is either a unit or a multiple of $p$, the value of $\sigma(s)$ is determined by the smallest value of $t$ such that $s \in p^t \Z_k$.
  
	As $\sigma(1) \ge \sigma(p^k) = \sigma(0)$, it follows that $\sigma(1) \ge \sigma(p) \ge \sigma(p^2) \ge \cdots \ge \sigma(p^k) = 0$.
	Let $\tau > 0$ be the smallest integer for which $\sigma(p^\tau) = 0$: then $\sigma(p^{\tau-1}) > \sigma(p^\tau)$ by construction.
	Let $0 \le t < \tau$, and $\delta = \tau - t - 1$; then we have
	\begin{equation}
	  \sigma(p^\delta p^t) > \sigma(p^\delta p^{t+1}).
	\end{equation}
	By the contrapositive of the monotonicity property, it follows that $\sigma(p^t) > \sigma(p^{t+1})$.
	Let $f: \R_+ \to \R_+$ be a piece-wise linear function such that $f(0) = 0$ and $f(1/p^t) = \sigma(p^t)$ for integers $0 \le t < \tau$: then $f$ is strictly increasing, and $\sigma = f \circ \sigma_{p^\tau}$.
\end{proof}

Extending the Lemma above, for any significance function $\sigma: \Z_k \to \R_+$, there is a non-decreasing function $f$ such that $\sigma = f \circ \sigma_{p^r} = f \circ \sigma_k$: it suffices to take $f(x) = 0$ for sufficiently small $x$, and let $f$ otherwise be strictly increasing.
For any complexity class which depends on distinguishing whether $\sigma(s) \le a$ or $\sigma(s) \ge b$, for amplitudes $s \in \Z_k$ and constants $0 \le a < b \le 1$, we may without loss of generality reduce the analysis to the following significance function:
\begin{definition}
  The \emph{canonical significance function} $\sigma_k: \Z_k \to \R_+$ for a prime power $k = p^r$ is the function satisfying
	\begin{equation}
	\label{eqn:canonicalSignificanceFn}
  \sigma_k(s) = \begin{cases}
                  1/p^t,	&	\text{if $s \in p^{t\!\;} \Z_k \:\!\setminus\:\! p^{t\raisebox{0.25ex}{$\scriptscriptstyle +$}\!\!\;1\!\;}\!\!\: \Z_k$ for $0 \le t < r$}\;; \\
                  0,	&	\text{if $s = 0$}.
                \end{cases}
	\end{equation}
\end{definition}
\noindent
This function satisfies $0 \le \sigma_k(st) \le \sigma_k(s)\sigma_k(t)$ for all $s,t \in \Z_k$.

\label{sec:reduceBddErrorToExact}
We now show that bounded $\sigma_k$-error computation may be reduced to exact computation on \mbox{$\Z_\kappa$-modal} distributions, where $\kappa = p^\tau$ for some $0 < \tau \le r$, by the classification of significance functions above for $\Z_k$ where $k = p^r$.
This holds for each of the three models of computation we consider: by circuits composed of invertible transformations, affine transformations, or unitary transformations.

For error bounds $0 \le a < b \le 1$, let $\{ C_n \}_{n \ge 1}$ be a circuit family which efficiently decides a language $L$ with $\sigma_k$-error bounds $(a,b)$, for the canonical significance function $\sigma_k$\,.
Because $\sigma_k$ takes only the values $0$ and $1/p^t$ for ${0 \le t < r}$, there is an integer $0 < \tau \le r$ for which $\{ C_n \}_{n \ge 1}$ decides $L$ with $\sigma_k$-error bounds $(a_k,b_k)$, where
\begin{equation}
	a_k = \sigma_k(p^\tau)  \qquad\text{and}\qquad b_k = \sigma_k(p^{\tau-1}).
\end{equation}
Consider an input $x \in \{0,1\}^n$ for some $n \ge 1$, and let $C_n \ket{x} = \ket{\varphi'}\ket{1{-}L(x)} + \ket{\varphi''}\ket{L(x)}$.
\begin{subequations}
It follows that
\begin{align}
\label{eqn:significanceBounds-mod-k}
	\!\!\!\!
	\bracket{y}{\varphi'} &\in \,p^\tau \Z_k \qquad\!\!\!\text{for all $y \in \{0,1\}^\ast$};
	\\
	\!\!\!\!
	\bra{y} T \ket{\varphi''} &\notin \,p^\tau \Z_k \qquad\!\!\!\text{for some $y \in \{0,1\}^\ast$ and some valid transformation $T$}.
\end{align}
\end{subequations}
For any choice of state-space, if $\bracket{y}{\varphi'} \in \kappa \Z_k$ for all $y \in \{0,1\}^\ast$, it follows that $\ket{\varphi'} \in p^\tau \sD$: in other words, $\ket{\varphi'} \equiv \vec 0 \pmod{p^\tau}$.
Conversely, there is a transformation $T$ for which $T \ket{\varphi''} \not\equiv \vec 0 \pmod{p^\tau}$, so that $\ket{\varphi''}$ is \emph{not} equivalent to the zero vector.
Thus we have
\begin{equation}
  C_n \ket{x} \;\equiv\; \ket{\varphi''} \ket{L(x)}	\pmod{p^\tau}, \quad \text{for $\ket{\varphi''} \not\equiv \vec 0 \bmod{p^\tau}$}.
\end{equation}
Note that $\Z_{p^\tau} \cong \Z_k / p^\tau \Z_k$, so it only remains to show that the circuit family $\{ C_n \}_{n \ge 1}$ on $\Z_k$-modal distributions can be used to obtain a valid (\ie,~an invertible, affine, or unitary) circuit on $\Z_{p^\tau}$-distributions as well.
This easily follows by the fact that invertibility of transformations over $\Z_k$, or congruences $\| \varphi \|_1 \equiv 1 \pmod{k}$ or $\bracket \varphi \varphi \equiv 1 \pmod{k}$, imply the same properties evaluated modulo $p^\tau$.

Thus, bounded-error computation in $\Z_k$ is equivalent to exact computation modulo $\kappa = p^\tau$ for some $1 \le \tau \le r$.
In particular, if $k$ is prime, $\Z_k$-modal computation does not admit any notion of bounded-error computation except for exact computation.

\subsection{Modal complexity classes for cyclic rings}
\label{sec:galoisModalComplexity}


For cyclic rings $\Z_k$, we now define notions of complexity for exact $\Z_k$-modal computation, for circuits composed of invertible, affine, or unitary transformations (as described following Proposition~\ref{lemma:validTransformations}).
Motivated by the observations of the preceding sections, we consider polytime-uniform circuit families $\{ C_n \}_{n \ge 1}$, consisting of arbitrary gates on four or fewer bits and preparation of bits in the state $\ket{0}$, subject to the constraints of the corresponding circuit model:
\begin{definition}
	\label{def:modalCircuitComplexity}
  Let $k > 1$ be an integer and $L \subset \{0,1\}^\ast$ be a language.
  \begin{romanum}
  \item
		$L \in \GLP[\Z_k]$ if and only if there is an invertible circuit family $\{ C_n \}_{n \ge 1}$ which efficiently decides $L$ exactly: specifically, a polytime-uniform family of circuits $C_n$ which consist of $m \in O(\poly n)$ preparation operations and $O(\poly n)$ invertible $\Z_k$-modal gates, such that $C_n \ket{x}\ket{0^m} = \ket{\psi'} \ket{L(x)}$ for all $x \in \{0,1\}^n$.
  \item
		$L \in \AffineP[\Z_k]$ if and only if there is an affine circuit family $\{ C_n \}_{n \ge 1}$ which efficiently decides $L$ exactly: specifically, a polytime uniform family of circuits $C_n$ which consist of $m \in O(\poly n)$ preparation operations and $O(\poly n)$ affine $\Z_k$-modal gates, such that $C_n \ket{x}\ket{0^m} = \ket{\psi'} \ket{L(x)}$ for all $x \in \{0,1\}^n$.
  \item
		$L \in \UnitaryP[\Z_k]$ if and only if there is a unitary circuit family $\{ C_n \}_{n \ge 1}$ which efficiently decides $L$ exactly: specifically, a polytime uniform family of circuits $C_n$ which consist of $m \in O(\poly n)$ preparation operations and $O(\poly n)$ unitary $\Z_k$-modal gates, such that $C_n \ket{x}\ket{0^m} = \ket{\psi'} \ket{L(x)}$ for all $x \in \{0,1\}^n$.
  \end{romanum}
\end{definition}
\noindent
These classes capture \emph{exact polynomial time computation} by \parit{i}~invertible, \parit{ii}~affine, and \parit{iii}~unitary operations over $\Z_k$, respectively.
(These definitions may be readily generalized to any ring $R$, by replacing $\Z_k$ in each instance by $R$, albeit limiting to polynomial-time specifiable gates in case $R$ is infinite.)

With these classes, we may describe the main problems posed by this article as follows.
The main result of Ref.~\cite{WS11} is to show $\text{UNIQUE-SAT} \in \GLP[\Z_{2}]$, suggesting that general ``modal quantum'' computation in the style of Refs.~\cite{SW10,WS11,SW-2012} is a powerful model of computation.
Can we characterize the power of such models in terms of traditional complexity classes?
Furthermore, as every unitary circuit is invertible, we have $\UnitaryP[\Z_k] \subset \GLP[\!\Z_k]$.
Refs.~\cite{WS11,HOSW11} conjecture in effect that this containment is strict, at least for primes $k \equiv 3 \pmod{4}$: does this conjecture hold?

Using standard techniques in counting complexity, it is not difficult to show that
\begin{equation}
\label{eqn:affineCharn}
  \AffineP[\Z_k] = \Mod[k]\P
\end{equation}
for $k$ a prime power.
This equality summarizes certain robust intuitions regarding oracle simulation techniques for \Mod[k]\P\ by describing them in terms of affine transformations of distributions.
(We sketch these techniques to justify Eqn.~\eqref{eqn:affineCharn} in Section~\ref{sec:affineCharn}.)
This motivates the question of whether there is a similar relationship for invertible and unitary transformations.
The main technical results of this article (Lemmas~\ref{lemma:modularContainsModal} and~\ref{lemma:modalContainsModular}) are to prove that in fact, for $k$ a prime power,
\begin{equation}
	\label{eqn:classEqualities}
	\GLP[\Z_k] \subset \Mod[k]\P \subset \UnitaryP[\Z_k],
\end{equation}
so that these classes are in fact equal.
Furthermore, $\GLP[\Z_k] = \UnitaryP[\Z_k] = \AffineP[\Z_k]$ for arbitrary integers $k > 1$ (including $k$ divisible by multiple primes), and these are equal to $\Mod[k]\P$ if and only if $\Mod[k]\P$ is closed under oracle reductions.

\subsection{Extending to Galois rings in general}
\label{sec:galoisRingBasics}

We now show how to reduce the analysis of modal computation on Galois-ring-valued distributions, to the case of cyclic rings $\Z_k$ for prime powers $k$.
We do so in to address ``modal quantum'' computation in general, which considers distributions whose amplitudes may range over a finite field, but also simply for the sake of generality.
Readers who are only interested in computation with $\Z_k$-valued distributions may safely proceed to Section~\ref{sec:characterization} (page~\pageref{sec:characterization}).

\subsubsection{Elementary remarks on finite fields and Galois rings}

A \emph{finite field} (or \emph{Galois field}) is a finite ring in which every non-zero element has an inverse.
The simplest examples are $\Z_p$, for $p$ prime.
For any prime power $k = p^r$,  there is a finite field with size $k$, which has structure of the vector space $\Z_p^r$ together with a multiplication operation which extends the scalar multiplication over this vector space.
This field is unique up to isomorphism, and is denoted $\F_k$.
The non-zero elements of $\F_k$ form a cyclic group under multiplication, generated by a single element.
For the general theory of finite fields and their extensions, see \eg~Refs.~\cite{Hungerford-1974,Wan-2003}.

Finite fields are an example of \emph{Galois rings} (see Wan~\cite{Wan-2003} for an elementary treatment).
A Galois ring $R$ is a finite commutative ring, which is both local and a principal ideal ring --- that is, the set of non-units $Z = {\{ z \:\!{\in}\:\! R \;|\;  1_R \notin z R \}}$ form an ideal, and furthermore is of the form $Z = p R = \{ p r \,|\, r {\in} R \}$ for some $p \in \N$. 
By standard number theoretic arguments, one may show that $p$ must be prime.
Elementary techniques (see \eg~Lemmas~14.2 and~14.4 of Ref.~\cite{Wan-2003} respectively) then suffice to show that $k = \Char(R)$ is a power of $p$, and that $|R| = k^e$ for some $e \ge 1$.
As with finite fields, there is a unique Galois ring $\mathrm{GR}(k,k^e)$ of character $k$ and cardinality $k^e$,
up to isomorphism.
Examples include the cyclic rings $\Z_k = \mathrm{GR}(k,k)$ for $k = p^r$ a prime power, and finite fields $\F_k = \mathrm{GR}(p,p^r)$.

\subsubsection{Relationship to prime-power order cyclic rings}

The following standard results (see Ref.~\cite{Wan-2003}) allow us to reduce the analysis of modal computation for any Galois ring $R = \mathrm{GR}(k,k^e)$ in terms of the cyclic ring $\Z_k$.
This also allows us to represent a linear transformation of an $R$-valued distribution to linear transformations of $\Z_k$-valued distributions.
\begin{itemize}
\item
	The ring $C = \{c \cdot 1_R \mid c \in \Z \}$ is isomorphic to the cyclic ring $\Z_k$.
	By convention we identify $C$ with $\Z_k$.
	We may construct $R$ as an extension $R = \Z_k[\tau]$, where $\tau$ is the root of some monic irreducible polynomial $f$ over $\Z_k$ of degree $e$,
	\begin{equation}
	  f(x) = x^e - f_{e-1} x^{e-1} - \cdots - f_1 x^1 - f_0,
	\end{equation}
	for coefficients $f_j \in \Z_k$.
	(If $k$ is prime, any such $f$ suffices.
	If $k = p^r$ for $r > 1$, the construction is similar but subject to additional constraints on $f$; in particular we require that $f_0$ be a unit modulo $k$.)
	As $f_0 = \tau(\tau^{e-1} - f_{\!\!\:e\:\!\text{--}\:\!1} \tau^{e-2} - \cdots - \tau_1)$ is a unit in $\Z_k$, it follows in particular that $\tau$ is a unit in $R$.
\item
	Any element $a \in R$ may be presented in the form $a_0 + a_1 \tau + \cdots + a_{e-1} \tau^{e-1}$ for coefficients $a_j \in \Z_k$, as higher powers of $\tau$ may be reduced via the relation
	\begin{equation}
		\label{eqn:reduceTau}
			\tau^e = f_{e-1} \tau^{e-1} + \cdots + f_1 \tau + f_0,
	\end{equation}
	which holds in $R$ by virtue of $f(\tau) = 0_R$.	
	Addition of elements of $R$ may be performed term-wise modulo $k$.
	Multiplication is performed by taking products of the polynomials over $\tau$, and simplifying them according to Eqn.~\eqref{eqn:reduceTau}.
\item
	Any element $a \in R$ can be represented by a vector in $\Z_k^e$.
	The standard basis vectors $\vec e_0, \vec e_1, \ldots, \vec e_{e-1}$ represent the monomials $1_{\!\!\:R\:\!}$, $\tau^1$, \ldots, $\tau^{e-1}$.
	Addition of elements of $R$ is then represented by vector addition, and multiplication by any constant $r \in R$ may be represented as a linear transformation $T_r: \Z_k^e \to \Z_k^e$ of these vectors, with columns representing how $\tau^\ell r$ decomposes into a linear combination of $\tau^j$ for $0 \le j \le e-1$.
\item
	Any linear transformation $M$ acting on vectors $\vec v \in R^d$ can be represented by a transformation of $\Z_k^{ed}$, where the coefficients $M_{i,j}$ are represented by $e \x e$ blocks.
	One decomposes $M$ as a linear combination of matrices $M^{(j)} \tau^j$, where $M\sur{j}$ is a $d \x d$ matrix over $\Z_k$.
	Multiplication of the coefficients of $\vec v$ by $\tau^j$ can be represented by a $\Z_k$-linear transformation of each coefficient $v_i \in R$ independently; one then transforms the resulting vector by $M\sur{j}$.
	Summing over all $j$, we may represent $M$ as an $ed \x ed$ matrix, acting on $e$-dimensional blocks.
\end{itemize}
These remarks show that $R$-distributions, and transformations of them, may be reduced to linear algebra over the cyclic ring $\Z_k$.

\subsubsection{Reduction of modal computation over Galois rings to cyclic rings}

In some cases, the above simulation techniques immediately suffice to reduce the complexity of $R$-modal circuits to the case of cyclic rings.
Because $\Z_k \subset R$, we trivially have the containments
\begin{align}
		\GLP[\Z_k] &\subset \GLP[\!\!\;R]\:\!,
	&
		\AffineP[\Z_k] &\subset \AffineP[\!\!\;R]\:\!,
	&
		\UnitaryP[\Z_k] &\subset \UnitaryP[\!\!\;R].
\end{align}
Furthermore, for invertible and affine transformations, the reverse containments also hold.
(Any matrix over $\Z_k$ which simulates an invertible matrix over $R$, is itself invertible; and if $\tilde M$ is a matrix over $\Z_k$ which simulates an affine matrix $M$ over $R$, then all columns of $\tilde M$ sum to $1$, as a simple corollary to the same property of $M$.)
Thus $\GLP[\Z_k] = \GLP[R]$ and $\AffineP[\Z_k] = \AffineP[R]$\,: using coefficients over the Galois ring $R$ provides at most a polynomial savings in work.

For unitary circuits, the situation is more complicated.
While every linear transformation $M$ over $R = \mathrm{GR}(k,k^e)$ can be easily simulated by a linear transformation $\tilde M$ over $\Z_k$, this does not mean that $\tilde M$ is a unitary transformation whenever $M$ is.
We consider a concrete example.
We may construct the finite field $\F_{25}$ as a ring extension $\Z_5[\sqrt 3]$, consisting of elements $r + s\sqrt 3$ for $r,s \in \Z_5$, where $\sqrt{3}{\:\!}^{\:\!2} = 3 \in \Z_5$. 
This ring admits a conjugation operation of the form $\overline{(r + s\sqrt 3)} = r - s\sqrt{3}$ for $r,s \in \Z_5$.
Then the matrix
\begin{equation}
  U = \begin{bmatrix}
        2 + \sqrt 3 & 0 \\[1ex] 0 & 2 - \sqrt 3
      \end{bmatrix}
\end{equation}
is unitary over $\F_{25}$\,, as $(2 + \sqrt 3)(2 - \sqrt 3) = 4 - 3 = 1$.
The simulation technique of the preceding Section would have us represent bit-vectors over $\F_{25}$ by vectors over $\Z_5$ as follows for $r_j,s_j \in \Z_5$:
\begin{equation}
  \Bigl(r_0+s_0\sqrt 3\Bigr) \ket{0} + \Bigl(r_1 + s_1\sqrt 3\Bigr) \ket{1} \;\mapstoname\quad\; r_0 \ket{00} + s_0 \ket{01} + r_1 \ket{10} + s_1 \ket{11}.
\end{equation}
In this representation, the matrix over $\Z_5$ which simulates the effect of $U$ on the standard basis is
\begin{equation}
  \tilde U = \begin{smatrix}
        2 & \!\phantom{-}3 & \!\phantom{-}0 & \!\phantom{-}0
       \\
				1 & \!\phantom{-}2 & \!\phantom{-}0 & \!\phantom{-}0
			\\
				0 & \!\phantom{-}0 & \!\phantom{-}2 & -3
			\\
				0 & \!\phantom{-}0 & -1 & \!\phantom{-}2
      \end{smatrix}.
\end{equation}
However, $\tilde U$ is not unitary over $\Z_5$.
Only the identity function satisfies the conditions $\overline 0 = 0$, $\overline 1 = 1$, and $\overline{(a+b)} = \overline a + \overline b$ on $\Z_5$; then we have $\tilde U\herm = \tilde U\:\!\trans$, so that
\begin{equation}
  \tilde U\herm \tilde U
  = \begin{smatrix}
        2 & \!\phantom{-}1 & \!\phantom{-}0 & \!\phantom{-}0
       \\
				3 & \!\phantom{-}2 & \!\phantom{-}0 & \!\phantom{-}0
			\\
				0 & \!\phantom{-}0 & \!\phantom{-}2 & -1
			\\
				0 & \!\phantom{-}0 & -3 & \!\phantom{-}2
      \end{smatrix}
		\begin{smatrix}
        2 & \!\phantom{-}3 & \!\phantom{-}0 & \!\phantom{-}0
       \\
				1 & \!\phantom{-}2 & \!\phantom{-}0 & \!\phantom{-}0
			\\
				0 & \!\phantom{-}0 & \!\phantom{-}2 & -3
			\\
				0 & \!\phantom{-}0 & -1 & \!\phantom{-}2
      \end{smatrix}
  = \begin{smatrix}
        0 & \!\phantom{-}2 & \!\phantom{-}0 & \!\phantom{-}0
       \\
				2 & \!\phantom{-}3 & \!\phantom{-}0 & \!\phantom{-}0
			\\
				0 & \!\phantom{-}0 & \!\phantom{-}0 & \!\phantom{-}2
			\\
				0 & \!\phantom{-}0 & \!\phantom{-}2 & \!\phantom{-}3
      \end{smatrix}.
\end{equation}
While the unitary matrices over $\F_{25}$ are represented by some group of linear transformations over $\Z_5$, that group is not a subset of the unitary transformations over $\Z_5$.

Despite this barrier to simulation, our main results $\GLP[\Z_k] \subset \Mod[k]\P \subset \UnitaryP[\Z_k]$ for prime powers $k$ (Lemmas~\ref{lemma:modularContainsModal} and~\ref{lemma:modalContainsModular}) imply that
\begin{equation}
	\UnitaryP[R] \subset \GLP[R] = \GLP[\Z_k] \subset \Mod[k]\P \subset \UnitaryP[\Z_k],
\end{equation}
so that $\UnitaryP[\Z_k] = \UnitaryP[R]$ nevertheless.
Computing with amplitudes in $R$ instead of $\Z_k$ then provides at most a polynomial size advantage, for unitary circuits as well.

Thus, for the computational models of this article, the above remarks serve to reduce the complexity of $R$-modal computation to $\Z_k$-modal computation.
This allows us to simplify our analysis by singling out the case of distributions over cyclic rings.

\section{The computational power of $\Z_k$-modal distributions}
\label{sec:characterization}

The preceding Section defined models of computation on $\Z_k$-valued distributions, along the lines of circuit complexity.
In this Section, we characterize the power of ``efficient'' computation (in the sense of Section~\ref{sec:modalCircuitComplexity}) in those models for all integers $k \ge 2$.
In particular, our main result is that the computational power of efficient algorithms in those models is exactly \Mod[k]\P\ when $k$ is a fixed prime power.

In Section~\ref{sec:affineCharn}, we demonstrate that $\AffineP[\Z_k] = \Mod[k]\P$ for prime powers $k$.
As we remarked in Section~\ref{sec:galoisModalComplexity}, this result is already implicit in counting complexity, and is part of the motivation of our investigation of the case for $\GLP[\Z_k]$ and $\UnitaryP[\Z_k]$.
We then show how to supplement these techniques with standard results of number theory and reversible computation, to prove the containments ${\GLP[\Z_k] \subset \Mod[k]\P \subset \UnitaryP[\Z_k]}$ in Sections~\ref{sec:simulatingModalByModkP} and~\ref{sec:simulatingModkPbyModal} respectively.

\subsection{Intuitions from the case of affine circuits}
\label{sec:affineCharn}

Nondeterministic Turing machines can in a sense simulate linear transformations of exponentially large size,
using standard techniques of counting complexity.
Each possible configuration (state + head position + tape contents) of the nondeterministic Turing machine represents a standard basis vector $\ket{x}$ representing a particular binary string $x \in \{0,1\}^\ast$.
From some initial configuration, the machine branches into a distribution over these configurations, weighted according to the number of branches ending at each configuration.
Each transition, deterministic or non-deterministic, governs how the distribution of configurations transforms with time.
(We describe how this is done in greater detail, as a part of the detailed proof of Lemma~\ref{lemma:modularContainsModal}.)


The principal differences between nondeterministic Turing machines and $\Z_k$-modal circuits is in the fact that \textbf{(a)}~nondeterministic Turing machines (in effect) simulate transformations of $\N$-distributions rather than $\Z_k$-distributions, and \textbf{(b)}~the two models have different conditions for distinguishing between \YES/\NO\ instances.
For prime powers $k$, these distinctions may effectively be removed when we compare $\Mod[k]\P$ algorithms to affine $\Z_k$-modal circuits.
For $k \ge 2$, $\Mod[k]\P$ algorithms only distinguish between a number of accepting branches which is either a multiple of $k$ (for \NO\ instances) or not a multiple of $k$ (for \YES\ instances).
For these algorithms, the distributions over Turing machine configurations are then in effect $\Z_k$-valued rather than $\N$-valued.
By Ref.~\cite[Theorems~23 and~30]{BGH90}, we may assume that a $\Mod[k]\P$ algorithm accepts with one branch mod $k$ for \YES\ instances, and with zero branches mod $k$ for \NO\ instances:\footnote{%
	For $k$ a prime, applying Fermat's Little theorem, this can be easily done with a Turing machine which simulates a \Mod[k]\P\ algorithm $k-1$ times in parallel and which accepts only if each simulation accepts.
	For $k = p^r$ a prime power, one performs more elaborate simulations to simulate testing whether the first $r$ digits of the $p$-adic expansion of the number of accepting paths is non-zero.
}
this implies in particular that $\Mod[k]\P$ is closed under negation.
These standard results motivate the claim of Eqn.~\eqref{eqn:affineCharn} on page~\pageref{sec:galoisModalComplexity}:
\begin{lemma}
	\label{lemma:affine}
  For any prime power $k \ge 2$, $\AffineP[\Z_k] = \Mod[k]\P$.
\end{lemma}
\begin{proof}[Proof (sketch).]
For any polytime-uniform $\Z_k$-affine circuit family $\{ C_n \}_{n \ge 1}$ and input $x \in \{0,1\}^\ast$, consider a nondeterministic Turing machine $\mathbf T$ which simulates each gate of $C_n$ in sequence, branching non-deterministically according to the gate amplitudes and re-writing the tape contents simulating the bits of the circuit as required.
(This technique is standard; we describe it in more detail in the proof of Lemma~\ref{lemma:modularContainsModal} for completeness.)
Suppose that $\{C_n\}_{n\ge 1}$ efficiently decides some language $L \in \AffineP[\Z_k]$ exactly, and that $\mathbf T$ accepts in each branch only when the output bit is $1$.
Then $\mathbf T$ accepts on a non-zero number of branches modulo $k$ if and only if $x \in L$, so that $L \in \Mod[k]\P$.

Conversely, for any language $L \in \Mod[k]\P$ and input $x \in \{0,1\}^\ast$, consider a non-deterministic Turing machine $\mathbf T$, which simulates non-deterministic machines $\mathbf T_0$ and $\mathbf T_1$ in parallel to test whether $x \in \bar L$ or $x \in L$ (respectively) by $\Mod[k]\P$ algorithms.
We allocate a particular tape cell $A$ on which $\mathbf T$ writes a $0$ in those branches where $\mathbf T_0$ accepts and/or $\mathbf T_1$ rejects, and $1$ in those branches where the reverse occurs.
Either $\mathbf T_0$ or $\mathbf T_1$ (but not both) accept in one branch modulo $k$, indicating whether $x \notin L$ or $x \in L$.
Then $\mathbf T$ accepts in $1$ branch modulo $k$ for all $x \in \{0,1\}^\ast$, writing either $0$ or $1$ onto $A$ in exactly one branch modulo $k$ according to whether $x$ is a \NO\ or a \YES\ instance.

We may represent the computational branches of the machine $\mathbf T$ above by a boolean string $b \in \{0,1\}^{B}$, for $B \in O(\poly n)$.
To simulate the branching, we prepare each bit $b_i$ of the branching string (together with some auxiliary bit $s_i$) in a distribution $\ket{\rho}_{b_i s_i} = \ket{01} + \ket{11} + (k-1) \ket{00}$, which is a $\Z_k$-affine operation on two bits.
Determining whether $\mathbf T$ accepts in any particular branch is a problem in $\P$, and so may be decided by a polytime-uniform boolean circuit family using $\textrm{\small AND}$, $\textrm{\small OR}$, $\textrm{\small NOT}$, $\textrm{\small FANOUT}$, and $\textrm{\small SWAP}$.
This may then be simulated by a polytime-uniform $\Z_k$-affine circuit, conditioned on $s_1 s_2 \cdots s_B = 11\cdots1$.
By construction, the bit representing the cell $A$ contains $1$ in one branch mod $k$, if and only if $x \in L$.
To obtain an exact $\Z_k$-affine algorithm, in which branches containing the incorrect answer have no effect, we may apply the $\textrm{\small ERASE} = [1\;\;\;1]$ operator on all bits aside from the bit $A$.
The output will then be simply $\ket{0}$ if $x \notin L$ and $\ket{1}$ if $x \in L$.  
\end{proof}
This result formalises well-known intuitions for the counting classes $\Mod[k]\P$, for $k$ a prime power.
Our contribution is to obtain similar results for invertible and unitary circuits as well, by solving the following problems:
\begin{itemize}
\item 
	The decision criterion of $\Mod[k]\P$ only counts the number of accepting branches of a nondeterministic Turing machine.
	How can it make the finer distinction, for a final computational state $C_n \ket{x} = \ket{\psi_0}\ket{0} + \ket{\psi_1}\ket{1}$ of an invertible modal circuit family $\{ C_n \}_{n \ge 1}$, whether $\ket{\psi_0}$ or $\ket{\psi_1}$ are zero?
\item
	For unitary circuit families $\{ C_n \}_{n \ge 1}$ in particular, how can it simulate the branching and the counting of accepting branches of a nondeterministic Turing machine without access to non-invertible operations such as $\text{\small ERASE} = [ 1 \;\;\; 1 ]$?
\end{itemize}

\subsection{Simulation of invertible $\Z_k$-circuits by $\mathsf{Mod_{\mathit k}P}$ algorithms}	
\label{sec:simulatingModalByModkP}

To simulate exact $\Z_k$-modal algorithms, a nondeterministic Turing machine must somehow detect whether there are non-zero amplitudes for the output $\ket{1}$ (corresponding to an answer of \YES), even if the sum of these amplitudes is a multiple of $k$.
For invertible circuits, it suffices to apply the technique of ``uncomputation'' from reversible computation, which in the exact setting produces a standard basis state as output.

\begin{lemma}
	\label{lemma:modularContainsModal}
  For any $k \ge 2$, $\GLP[\Z_k] \subset \Mod[k]\P$.
\end{lemma}
\begin{proof}
  Let $L \in \GLP[\Z_k]$ be decided by an invertible polytime-uniform $\Z_k$-circuit family $\{C_n\}_{n \ge 1}$  such that $C_n \ket{x} = \ket{\psi_x}\ket{L(x)}$ for each $x \in \{0,1\}^n$.
	Suppose that $C_n$ requires $m$ preparation operations: we may suppose that these are all performed at the beginning of the algorithm, in parallel.
	Using $\textrm{\small SWAP}$ operations, we may suppose that they are initially used to prepare a contiguous block of bits in the state $\ket{0}$.
	Abusing notation slightly, we may then describe each $C_n$ as an invertible operation $C_n: \sB\sox{n+m} \to \sB\sox{n+m}$ such that $C_n \ket{x}\ket{0^m} = \ket{\psi_x}\ket{L(x)}$.
  
	Let $T \in O(\poly(n))$ be the number of gates contained in $C_n$.
	Then we may construct a circuit $C'_n$ as illustrated in Figure~\ref{fig:modalUncompute}, consisting of $2T+1$ gates such that $C'_n \ket{x}\ket{0^{m+1}} = \ket{x}\ket{0^{m}}\ket{L(x)}$, as follows:
	\begin{figure}
	\begin{center}
	  \begin{tikzpicture}
	    \coordinate (x1-0) at (0,0);
	    \xdef\prev{x1}
	    \foreach \j in {x2,x3,x4} {%
				\coordinate (\j-0) at ($(\prev-0) + (0,-0.2)$);
				\xdef\prev{\j}
			}
			
	    \coordinate (w1-0) at ($(\prev-0) + (0,-0.6)$);
	    \xdef\prev{w1}
	    \foreach \j in {w2,w3,w4,w5} {%
				\coordinate (\j-0) at ($(\prev-0) + (0,-0.2)$);
				\xdef\prev{\j}
			}
			
			\coordinate (a-0) at ($(\prev-0) + (0,-0.6)$);
			
			\xdef\prev{0}
			\foreach \t in {1,2,3,4} {%
				\foreach \j in {x1,x2,x3,x4,w1,w2,w3,w4,w5,a} {%
					\coordinate (\j-\t) at ($(\j-\prev) + (1,0)$);
					\draw (\j-\prev) -- (\j-\t);
				}
				\xdef\prev{\t}
			}
			
			\node (x1-1) at (x1-1) {};
			\node (w5-1) at (w5-1) {};
			\node (C) at ($(x1-1)!0.5!(w5-1)$) {\large$C_n^{\phantom{-1}}$};
			\node (C) [draw,fill=white,inner sep=3pt,fit=(x1-1)(w5-1)(C)] {};
			\node at (C) {\large$C_n$};
			
			\filldraw [black] (w5-2) circle (2pt);
			\draw [black] (a-2) circle (4pt);
			\draw (w5-2) -- ($(a-2) + (0,-4pt)$);
			
			\node (x1-3) at (x1-3) {};
			\node (w5-3) at (w5-3) {};
			\node (C) at ($(x1-3)!0.5!(w5-3)$) {\large$C_n^{-1}$};
			\node (C) [draw,fill=white,inner sep=3pt,fit=(x1-3)(w5-3)(C)] {};
			\node at (C) {\large$C_n^{-1}$};

			\coordinate (x1'-0) at (-4,0);
	    \xdef\prev{x1'}
	    \foreach \j in {x2',x3',x4'} {%
				\coordinate (\j-0) at ($(\prev-0) + (0,-0.2)$);
				\xdef\prev{\j}
			}
			\node [label=west:$n \;\biggl\{\!\!$] at ($(x1'-0)!0.5!(x4'-0)$)  {};
			
	    \coordinate (w1'-0) at ($(\prev-0) + (0,-0.6)$);
	    \xdef\prev{w1'}
	    \foreach \j in {w2',w3',w4',w5'} {%
				\coordinate (\j-0) at ($(\prev-0) + (0,-0.2)$);
				\xdef\prev{\j}
			}
			\node [label=west:$m \;\Biggl\{\!\!$] at ($(w1'-0)!0.5!(w5'-0)$)  {};
			
			\coordinate (a'-0) at ($(\prev-0) + (0,-0.6)$);
			
			\xdef\prev{0}
			\foreach \t in {1,2} {%
				\foreach \j in {x1',x2',x3',x4',w1',w2',w3',w4',w5',a'} {%
					\coordinate (\j-\t) at ($(\j-\prev) + (1,0)$);
					\draw (\j-\prev) -- (\j-\t);
				}
				\xdef\prev{\t}
			}

			\node (x1'-1) at (x1'-1) {};
			\node (a'-1) at (a'-1) {};
			\node (C) at ($(x1'-1)!0.5!(a'-1)$) {\large$C'_n$};
			\node (C) [draw,fill=white,inner sep=3pt,fit=(x1'-1)(a'-1)(C)] {};
			\node at (C) {\large$C'_n$};
			
			\node at ($(a'-2)!0.5!(x1-0)$) {\Large$\equiv$};
	\end{tikzpicture}
	\end{center}
	  \caption{%
			Schematic of an invertible modal circuit $C'_n$, constructed from another invertible modal circuit $C_n$ and its inverse.
			Given a decomposition of $C_n$, we may decompose $C_n^{-1}$ as inverse of each gate of $C_n$ performed in reverse order.
			The top $n$ bits represent the input, and the remaining $m+1$ bits are workspace bits which are each prepared initially in the state $|0\rangle$.
	  }
	  \label{fig:modalUncompute}
	\end{figure}
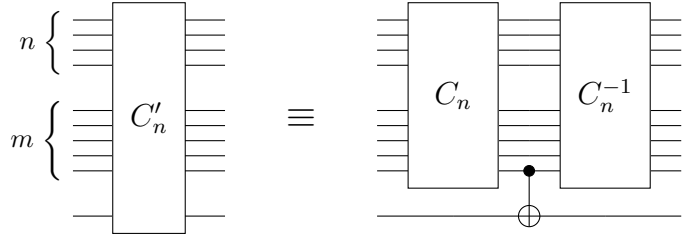	
	\begin{enumerate}
	\item
		Perform $C_n \ox \idop$ on $\ket{x}\ket{0^m}\ket{0}$, obtaining $\ket{\psi_x}\ket{L(x)}\ket{0}$.
	\item
		Perform a \textsc{cnot} gate on the final two bits, obtaining $\ket{\psi_x}\ket{L(x)}\ket{L(x)}$.
	\item
		Perform $C_n^{-1} \ox \idop$ (that is, perform the inverse of each gate in $C_n$ in reverse order), obtaining $\ket{x}\ket{0^m}\ket{L(x)}$.
	\end{enumerate}
	The final state is then a standard basis state 
	with $L(x)$ stored in the final bit.
	It then suffices to consider a nondeterministic Turing machine simulating $C'_n$.

	We may simulate $C'_n$ on a nondeterministic Turing machine $\mathbf N$, in such a way that for $x \in L$ the machine accepts on one branch modulo $k$, and for $x \notin L$ the machine accepts on a number of branches which is a multiple of $k$\iftrue.
	\begin{itemize}
	\item
		Consider a sequence $G_1, G_2, \ldots, G_{2T+1}$ of integer matrices with coefficients ranging from $0$ to $k-1$, which act on $\Z^{\{0,1\}^{n+m+1}}\!$ and whose coefficients are congruent mod $k$ to the action of the gates of $C'_n$ on the computational basis.
		We interpret the matrices $G_t$ as describing a transition function on boolean strings, with an integer weight assigned to each transition.
	\item
		From an initial configuration with $x\,0^{m+1}$ on the tape, we simulate the action of the matrices $G_1$, $G_2$, \ldots\ in sequence by performing nondeterministic transitions.
		\begin{enumerate}
		\item
			Before each matrix $G_t$, we take the contents of the tape $x\sur{t} \in \{0,1\}^{n+m+1}$ in each computational branch as representing a standard basis state.
		\item
			We non-deterministically select a string $x\sur{t+1} \in \{0,1\}^{n+m+1}$.
		\item
			Compute $c = \bra{\smash{x^{(t+1)}}} G_t \ket{\smash{x^{(t)}}} \in \N$, and create a further $c$ branches (\eg~by creating $k$ different branches and immediately halting in a rejecting state in $k-c$ of them).
		\item
			Replace $x\sur{t}$ on the tape with $x\sur{t+1}$, and proceed to the next iteration.
		\end{enumerate}
	\item
		Accept in every branch for which the final contents of the tape has the form $x \,0^m \, 1$, and reject otherwise. 
	\end{itemize}
	It is easy to show by induction that, after simulating the $t\textsuperscript{th}$ gate as above, the number of computational branches in which a given string $x\sur{t} \in \{0,1\}^n$ is written in the index space is given by
	\begin{equation} 
		N(x,t,x\sur{t}) = \bra{\smash{x\sur{t}}} G_t \cdots G_2 G_1 \ket{x}.
	\end{equation}
	Finally we have $N(x,2T+1,y) = \bra{y} G_{2T+1} \cdots G_2 G_1 \ket{x}$.
	By hypothesis, this is equivalent to $\alpha_y := \bra{y} C'_n \ket{x} \pmod{k}$.
	By construction, $C'_n \ket{x}\ket{0^{m-n}}\ket{0} = \ket{x}\ket{0^{m-n}}\ket{L(x)}$, so that
	\begin{equation}
		\alpha_y = \begin{cases} 1 &	\text{if $y = x\,0^{m}\,1$ and $L(x) = 1$}, \\
															0	& \text{otherwise}.
		            \end{cases}
	\end{equation}
	Then $\mathbf N$ accepts on precisely one branch modulo $k$ if $L(x) = 1$, and on zero branches modulo $k$ if $L(x) = 0$.
	Thus $L \in \Mod[k]\P$, so that $\GLP[\Z_k] \subset \Mod[k]\P$.
	\else,
		using standard techniques (see Ref.~\cite{BDHM92}), simply by having the machine accept on all branches of the computation for which $L(x) = 1$: then the number of accepting branches of the computation will be congruence to $1$ mod $k$ if $L(x) = 1$, and congruent to $0$ mod $k$ if $L(x) = 0$.
	\fi
\end{proof}

The above result does not require $k \ge 2$ to be a prime power, or for the gates to act on $O(1)$ bits, and so applies for all moduli.
(We similarly have $\AffineP[\Z_k] \subset \Mod[k]\P$ for all $k \ge 2$.)

\subsection{Simulation of $\mathsf{Mod_{\mathit k}P}$ algorithms by unitary $\Z_k$-modal circuits}	
\label{sec:simulatingModkPbyModal}

For prime powers $k \ge 2$, we may show a strong version of the converse to Lemma~\ref{lemma:modularContainsModal}, in which all gates are unitary modulo $k$.
For any given modulus $k \ge 2$, our proof involves a relatively small set of unitary gates, which is therefore able to efficiently simulate any other finite unitary gate-set (or indeed any polytime-specifiable invertible gate-set).
Following the approach of Lemma~\ref{lemma:affine}, we use the classical gates $\textrm{\small NOT}$, $\textrm{\small CNOT}$, $\textrm{\small TOFFOLI}$, $\textrm{\small SWAP}$, and four more gates which we now describe.

We would like an operation, with which to simulate the branching of a non\-determ\-in\-istic Turing machine well enough to count its accepting branches modulo $k$.	
It will usually not be possible to prepare the uniform superposition $\ket{\varphi} = \ket{0} + \ket{1}$ on each bit individually with unitary gates, as these will not be $\ell_2$ states in the case that ${\bracket{\varphi}{\varphi} \ne 1}$.
As with the $\Z_k$-affine circuits in Lemma~\ref{lemma:affine}, we may circumvent this problem by considering gates which only conditionally creates a uniform distribution.


\begin{lemma}
	\label{lemma:orthogonalBranchingOptor}%
  For any $k \ge 2$, and for $\sB = \Z_p^2$, there is a unitary matrix $K: \sB\sox{3} \to \sB\sox{3}$ such that for some $\ket{\gamma} \in \sB\sox{2}$,
  \begin{equation}
		\label{eqn:unitaryBranching-mod-primes}
		K \ket{000} \,=\, \ket{\gamma}{\ox}\ket{0} \,+\, \Bigl(\ket{0} + \ket{1}\Bigr)\:\!{\ox}\ket{11}  
  \end{equation}
\end{lemma}
\begin{proof}
  Associate an octonian (an element of the 8-dimensional $\star$-algebra arising via the Cayley--Dickson construction~\cite[\SectionSymbol2.2]{Baez-2002} on the quaternions) to each integer vector $\vec v = [ v_0 \;\; v_1 \;\; \cdots \;\; v_7 ]\,\trans \in \Z^8$,
	\begin{equation}
	  \omega_{\vec v}	\,=\,	v_0 \vec e_0 \,+\, v_1 \vec e_1 \,+\, v_2 \vec e_2  \,+\, v_3 \vec e_3 \,+\, v_4 \vec e_4 \,+\, v_4 \vec e_5 \,+\, v_6 \vec e_6 \,+\, v_7 \vec e_7,
	\end{equation}
	where $\mathbf e_0 = 1$ and $\mathbf e_j$ are imaginary units for $0 < j < 8$.
	The integer dot-product of vectors can be evaluated as $\vec v \cdot \vec w = \mathrm{Re}(\bar \omega_{\vec v} \omega_{\vec w})$, where $\bar \omega_{\vec v} = v_0 \vec e_0 - v_1 \vec e_1 - \cdots - v_7 \vec e_7$ and where $\mathrm{Re}(\omega_{\vec v}) = v_0$ extracts the real part.
	Consider an integer solution to $a^2 + b^2 + c^2 + d^2 = k-1$: by the Lagrange four squares theorem, there is a solution for every $k \ge 1$.
	Then we define the vector
	\begin{equation}
			\begin{aligned}[b]
				\vec v_0 \;&=\; \begin{bmatrix} a & 0 & b & 1 & c & 0 & d & 1 \end{bmatrix}\trans
			\\&=\;
			  \Bigl(a\ket{00} + b\ket{01} + c\ket{10} + d\ket{11}\Bigr)\ox\ket{0} \,+\, \Bigl(\ket{0} + \ket{1}\Bigr)\ox\ket{11}.
			\end{aligned}
	\end{equation}
	For each $0 \le j \le 7$ define an octonian $\omega_j = \omega_{\vec v_0} \vec e_j\;\!$.
	Then we have 
	\begin{equation}
	  \bar\omega_j \omega_j \;=\;	a^2 + b^2 + 1 + c^2 + d^2 + 1 \;=\; k+1
	\end{equation}
	for each $0 \le j \le 7$, by construction.
	For each $0 \le h,j \le 7$, we also have
	\begin{equation}
			\mathrm{Re}\bigl(\bar \omega_h \omega_j\bigr)
		=
			- \mathrm{Re}\bigl(\bar \omega_h \omega_h \vec e_h \vec e_j\bigr)
		=
			- \mathrm{Re}\bigl((k-1) \vec e_h \vec e_j\bigr)
		=
			\begin{cases}
				k+1,	&	\text{if $h = j$}; \\ 0, & \text{if $h \ne j$},
			\end{cases}
	\end{equation}
	as $\vec e_h^2 = -1$, whereas $\vec e_h \vec e_j$ is imaginary if $h \ne j$.
	For each $0 \le j \le 7$, let $\vec v_j \in \Z^8$ be the vector such that $\omega_j = \omega_{\vec v_j}$.
	Then these vectors are orthogonal, and if we let $K = {\bigl[\vec v_0 \mid \vec v_1 \mid \cdots \mid \vec v_7 \bigr]}$, $K\,\trans K$ is equivalent to the $8 \x 8$ identity matrix mod $k$.
\end{proof}
For the purposes of our analysis, any such operator $K$ will suffice.
For any integer $k \ge 2$, one may efficiently find solutions to $a^2 + b^2 + c^2 + d^2 = k-1$ by randomized algorithms~\cite{RS-1986}, and then consider circuits with this gate included as a primitive gate.
We also consider a conditionally controlled version of $K$,
\begin{equation}
	\begin{split}
		\Lambda K &: \sB\sox{4} \to \sB\sox{4}
	\\
		\Lambda K &\ket{c} \ket{\psi} = \ket{c} \ox K^c \ket{\psi}	\qquad\text{for $\ket{\psi} \in \sB\sox{3}$ and $c \in \{0,1\}$}
	\end{split}
\end{equation}
which will allow us to simulate $K$ depending on certain conditions, pre-computed in another bit.
We also include the inverses of $K$ and $\Lambda K$ as primitive gates.
Together with the classical reversible gates, these suffice to demonstrate: 
\begin{lemma}
	\label{lemma:modalContainsModular}
  For any prime power $k \ge 2$, $\Mod[k]\P \subset \UnitaryP[\Z_k]$.
\end{lemma}
\begin{proof}
	The acceptance condition of any nondeterministic Turing machine, in a given computational branch, can be computed in polynomial time by a deterministic Turing machine.
	We may represent the choices of transitions made by a nondeterministic Turing machine by a binary string $b$, which we refer to as the \emph{branching string} of that branch.	
	Furthermore, any problem in \P\ may be computed by a polytime-uniform reversible circuit family (consisting of \textrm{\small NOT}, \textrm{\small CNOT}, \textrm{\small TOFFOLI}, and \textrm{\small SWAP} gates)~\cite{Toffoli-1980}.
	We may then represent a nondeterministic Turing machine $\mathbf N$ which halts in polynomial time, by
	\begin{romanum} 
	\item
		a polytime-uniform reversible circuit family $\{R_n\}_{n \ge 1}$, acting on
	\item
		an input string $x \in \{0,1\}^n$, an auxiliary branching string $b \in \{0,1\}^B$, and $m$ work bits for $m,B \in O(\poly n)$, 
	\end{romanum}
  such that $R_n$ outputs $1$ on its final bit if and only if $\mathbf N$ accepts the input $x$ in the computational branch uniquely labelled by $b$.

  From these remarks, it follows for $L \in \Mod[k]\P$
  that there is a polynomial-uniform reversible circuit family $\{R_n\}_{n\ge 1}$ of this sort, acting on $N = n+b+m$ bits, for which $L(x) = 1$ if and only if
  \begin{equation}
		\#\Bigl\{ b \in \{0,1\}^B \,\Big|\, R_n(x,b,0^m)_{N} \;\!{=}\:\! 1\Bigr\} \;\not\equiv\; 0 \pmod{k}.
  \end{equation}
	Furthermore, without loss of generality, we suppose that 
	the number of accepting branches is congruent either to $0$ or to $1$ modulo $k$~\cite[Theorems~23 and~30]{BGH90}, and similarly assume that the total number of branches of the computation is equivalent to $1$ modulo $k$.
	Then either the number of accepting branches is zero modulo $k$ and the number of rejecting branches is one modulo $k$, or vice-versa.

	We consider a polytime-uniform unitary $\Z_k$-modal circuit family $\{C_n\}_{n\ge 0}$ acting on 
	$n + 5B + 3m + 1$ bits, as illustrated in Figure~\ref{fig:orthogonalCircuit}.
	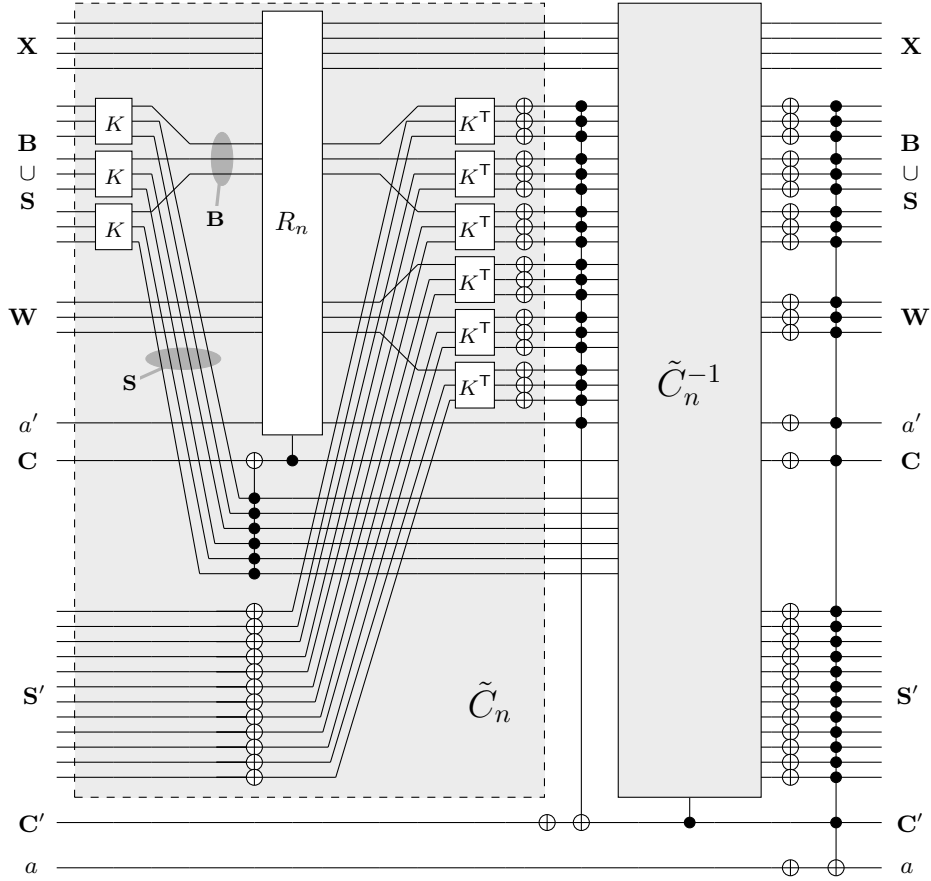
\begin{figure}[!p]
	\begin{center}\iftikz\input{fig-orthogSimul}\else[Place holder for a TiKZ-picture]\fi\end{center}
	\caption{%
		For prime powers $k$, a schematic diagram for a $\Z_k$-unitary circuit to simulate a $\Mod[k]\P$ algorithm.
		The \Mod[k]\P\ algorithm is represented by a reversible circuit $R_n$ (composed \eg~of $\textrm{\small NOT}$, $\textrm{\small CNOT}$, $\textrm{\small TOFFOLI}$, and $\textrm{\small SWAP}$ gates), acting on an input register $\mathbf X$ and a branching register $\mathbf B$, which computes the acceptance condition of a nondeterministic Turing machine on an input provided in $\mathbf X$ in the computational branch labelled by the string in $\mathbf B$.
		The $K$ gates are the three-bit gates described by Lemma~\ref{lemma:orthogonalBranchingOptor}, with inverse $K^{-1} = K^{\mathsf T}$.
		Gates which are connected to black dots on one or more bits $c_1, c_2, \ldots,$ such as the multiply controlled-$\textrm{\small NOT}$ gates of Eqn.~\eqref{eqn:multiply-controlled-not},	are operations which for standard basis states are performed conditional on each bit $c_j$ being $1$.
		The portion of the circuit in a shaded and dashed box defines a subcircuit $\tilde C_n$.
		The subcircuit $\tilde C_n$ is conditionally reversed in the second half of the circuit, depending on the bit $\mathbf C'$.
		All bits except for the input register $\mathbf X$ are initialized to $|0\rangle$, and the output is the final bit $a$.
	}
	\label{fig:orthogonalCircuit}
	\end{figure}
	We group these bits into registers as follows:
	\begin{itemize}[itemsep=0.25ex]
	\item
		An input register $\mathbf X$ on $n$ bits;
	\item
		The ``principal'' branching register $\mathbf B$ on $B$ bits;
	\item
		A ``branching success'' register $\mathbf S$ on $2B$ bits;
	\item
		A ``principal'' work register $\mathbf W$ on $m-1$ bits;
	\item
		A ``simulation accept'' bit $a'$;
	\item
		A ``simulation control'' bit $\mathbf C$;
	\item
		A ``summation'' register $\mathbf S'$ on $2(B+m-1)$ bits;
	\item
		A ``reverse simulation control'' bit $\mathbf C'$; and
	\item
		An answer bit $a$.
	\end{itemize}
	Given an input $x \in \{0,1\}^n$ in the input register $\mathbf X$, the circuit $C_n$ performs the following operations (all conditional operations extend linearly for combinations of standard basis states):
	\begin{enumerate}
	\item
		Prepare all of the bits except for the input bits in $\mathbf X$ in the state $\ket{0}$.
	\item
		Match each bit of $\mathbf B$ with two bits of $\mathbf S$, and act on each triple with the operator $K$ (as described in Lemma~\ref{lemma:orthogonalBranchingOptor}).
	\item
		Conditioned on all bits of $\mathbf S$ being in the state $\ket{1}$, flip the bit $\mathbf C$.
	\item
		Conditioned on $\mathbf C$ being in the state $\ket{1}$, simulate $R_n$ on $(\mathbf X, \mathbf B, \mathbf W, a')$ with $\mathbf B$ as the branching register, $\mathbf W$ as the work register, and $a'$ as the output.
	\item
		Flip all bits in the register $\mathbf S'$.
	\item
		Match each bit of $\mathbf B \cup \mathbf W$ with two bits of $\mathbf S'$, and perform $K\,\trans$ on each triple.
		That is, perform $K\,\trans$ on $(\mathbf B_j, \mathbf S'_{2j-1}, \mathbf S'_{2j})$ for each $1 \le j \le B$, and also on $(\mathbf W_j, \mathbf S'_{2B+2j-1}, \mathbf S'_{2B+2j})$, for each $1 \le j \le m-1$.
	\item
		Flip all bits in the registers $\mathbf B$, $\mathbf W$, and $\mathbf S'$.
	\item
		Flip the bit $\mathbf C'$, and conditioned on all bits of $(\mathbf B, \mathbf W, \mathbf S', a')$ being $\ket{1}$, toggle $\mathbf C'$ back again.
	\item
		Conditioned on $\mathbf C'$ being in the state $\ket{1}$, undo each operation in steps~2--7 in reverse order (using controlled versions of $K$ and $K\,\trans$, as well as controlled versions of each of the classical gates used).
	\item
		Flip all bits in all registers apart from $\mathbf X$ and $\mathbf C'$.
	\item
		Conditioned on every bit in all registers apart from $\mathbf X$ and $a$ being in the state $\ket{1}$, flip the bit $a$, and produce it as output.		
	\end{enumerate}
	The effect of these transformations, described in high-level terms, is as follows.
	\begin{itemize}
	\item 
		Steps~1--3 simulate the preparation of the uniform distribution over all branching strings, in preparation to simulate the nondeterministic machine $\mathbf N$.
		At the end, $\mathbf C$ has the value $1$ if this is successful.
	\item
		Step~4 simulates of $\mathbf N$, conditioned on successfully preparing the uniform distribution on branching strings.
		This creates a number of branches in which $\mathbf N$ accepts or rejects; and $\mathbf C$ serves now to indicate whether the machine $\mathbf N$ was simulated.
	\item
		Steps~5--7 simulate the summation of the amplitudes of all \emph{rejecting} branching strings, and all \emph{accepting} branching strings, into one standard basis state each.
		These two states are $\ket{x\,1^{3(B{+}m{-}1)} \:\! 0 \:\! 1^{2B{+}1}\:\! 00}$ and $\ket{x\,1^{3(B{+}m{-}1)} \:\! 1 \:\! 1^{2B{+}1}\:\! 00}$ respectively.
		Successful simulation and summation of amplitudes is represented by the two blocks of $1$s which are present in both cases.
	\item
		Step~8 sets a bit to indicate failure in summing the amplitudes of all accepting branches, so that $\mathbf C'$ is equal to $1$ conditioned on either having failed to simulate the nondeterministic machine $\mathbf N$, or on $\mathbf N$ rejecting, or on having failed to sum the accepting branches, or on the number of accepting branches being a multiple of $k$.
		When this occurs, step~9 ``undoes'' the simulation of $\mathbf N$.
	\item
		If the number of accepting branches of $\mathbf N$ is a multiple of $k$, then step~9 should undo the entire computation, restoring the initial state $\ket{x 0^{5B{+}3m{+}1}}$.
		In steps~10--11, we attempt to set the final bit to $\ket{0}$ if this is the case, and to $\ket{1}$ otherwise.	
	\end{itemize}
	We now explicitly compute the effect of the circuit.
	For the sake of brevity, we will omit tensor factors which are in the state $\ket{0}$.
	The order of the tensor factors in the development below may differ from that shown in Figure~\ref{fig:orthogonalCircuit}.

	After preparing the non-input registers in step~1 and performing the $K$ operations in step~2, the state of the computation is given by
	\begin{equation}
			\ket{\psi_2}
		\;=\;
			\ket{x}_\bX \ket{\Gamma}_{\bB,\bS} 	\;+\;
			\ket{x}_\bX \left( \;\;\;	\sum_{\mathclap{\;\;b \in \{0,1\}^B}}	\;\ket{b\;\!1^{\!\!\;{2B}}}_{\bB,\bS} \right) ,
	\end{equation}
	where $\ket{\Gamma} \in \sB\sox{3B}$ is a state such that $\bracket{b\;\!1^{\!\!\;2B}}{\Gamma} = 0$ for all $a,b \in \{0,1\}^B$.
	The second term is interpreted as a successful preparation of all possible branching strings: step~3 simply prepares the control bit $\bC$ to indicate this success, yielding
	\begin{equation}
			\ket{\psi_3}
		\;=\;
			\ket{x}_\bX \ket{\Gamma}_{\bB,\bS}\ket{0}_\bC 	\;+\;
			\ket{x}_\bX \! \left( \;\;\;	\sum_{\mathclap{\;\;b \in \{0,1\}^B}}	\;\ket{b\;\!1^{\!\!\;2B}}_{\bB,\bS} \right) \!\ket{1}_\bC.
	\end{equation}
	For the sake of brevity, let $\bW'$ be the joint register $(\bB,\bW,\bS)$, and let $\ket{\Gamma'}_{\bW'}$ denote the tensor product of $\ket{\Gamma}_{\bB,\bS}$ with $\ket{0^{m-1}}_\bW$\,.
	Let $f(x,b)$ be the function computed by $R_n$.
	Step~4 then simulates $R_n$ on $(\bX, \bB, \bW, a')$ conditioned on $\bC$ being $\ket{1}$, yielding
	\begin{equation}
	\begin{aligned}[b]
			\ket{\psi_4}
		\;={}&\;
			\ket{x}_\bX \ket{\Gamma'}_{\bW'}\ket{0}_\bC\ket{0}_{a'}
		\\&
		\;+\;
			\ket{x}_\bX \! \left( \;\;\;	\sum_{\mathclap{\;\;b \in \{0,1\}^B}}	\;\ket{b, w(x,b)}_{\bB,\bW} \ket{1^{\!\!\;2B}}_{\bS} \ket{1}_{\bC} \ket{f(x,b)}_{a'} \right) ,
	\end{aligned}
	\end{equation}
	for some $w : \{0,1\}^{n+B} \to \{0,1\}^{m-1}$ computed on $\bW$.
	\begin{subequations}
	We may re-express the above according to the two possible values of $f(x,b)$ as
	\begin{align}
			\ket{\psi_4}
		\;={}&
		\;
			\ket{x}_\bX \ket{\Gamma'}_{\bW'}\ket{0}_\bC \ket{0}_{a'}
		\notag\\
		\;&+\;
			\ket{x}_\bX \! \left( \; \sum_{\mathclap{\;b \in \mathbf R_x}}	\;\ket{b, w(x,b)}_{\bB,\bW} \right) \ket{1^{\!\!\;2B+1}}_{\bS,\bC} \ket{0}_{a'}
		\notag\\&\quad
		+\;
			\ket{x}_\bX \! \left( \; \sum_{\mathclap{\;b \in \mathbf A_x}}	\;\ket{b, w(x,b)}_{\bB,\bW} \right) \ket{1^{\!\!\;2B+1}}_{\bS,\bC} \ket{1}_{a'}
	\\[2ex]
		&\quad
			\phantom{\mathbf A_x}\text{where }
			\mathbf R_x \!\!\;= \{ b \in \{0,1\}^B \,|\, f(x,b) = 0 \},
		\\&\quad
			\phantom{\text{where }\mathbf R_x\!\!\;}
			\mathbf A_x = \{ b \in \{0,1\}^B \,|\, f(x,b) = 1 \}:
	\end{align}
	thus $\{0,1\}^B = \mathbf R_x \cup \mathbf A_x$.
	\end{subequations}
	Let $S' = B+m - 1$ for the sake of brevity.
	Let $\bW''$ be the joint register ${(\bB, \bW, \bS', \bS, \bC)}$, and let $\ket{\Gamma''}_{\bW''}$ denote the tensor product of $\ket{\Gamma'}_{\bW'}$ with $\ket{\smash{1^{2S'}}0}_{\bS',\bC}$\,.
	Introducing the register $\bS'$ after toggling the value of each of its bits in step~5 then yields the state
	\begin{equation}
	\label{eqn:orthogonalEvolPreSummation}
  	\begin{aligned}[b]
			\ket{\psi_5}
		\;={}&
		\;
			\ket{x}_\bX \ket{\Gamma''}_{\bW'}\ket{0}_{a'}
		\\
		\;&+\;
			\ket{x}_\bX \! \left( \; \sum_{\mathclap{\;b \in \mathbf R_x}}	\;\ket{b, w(x,b)}_{\bB,\bW} \ket{\big.1^{\smash{\!\!\;2S'}}}_{\bS'} \right) \! \ket{1^{\!\!\;2B+1}}_{\bS,\bC} \ket{0}_{a'}
		\\&\quad
		+\;
			\ket{x}_\bX \! \left( \; \sum_{\mathclap{\;b \in \mathbf A_x}}	\;\ket{b, w(x,b)}_{\bB,\bW} \ket{\big.1^{\smash{\!\!\;2S'}}}_{\bS'}\right) \! \ket{1^{\!\!\;2B+1}}_{\bS,\bC} \ket{1}_{a'} .
		\end{aligned}
	\end{equation}
	We collect the first two terms on the right-hand side of Eqn.~\eqref{eqn:orthogonalEvolPreSummation} into a state of the form $\ket{x}_\bX\ket{\Gamma'''}_{\bW''}\ket{0}_{a'}$.
	Let $\tilde K$ be the operation consisting of performing $K$ to the bits of $(\bB,\bW)$ and $\bS'$ in groups of three, and let $\ket{\smash{\tilde\Gamma}}_{\bW''} := \tilde K\,\trans \ket{\Gamma'''}_{\bW''}$.
	\begin{subequations}
	Then the state after step~6 is given by
	\begin{equation}
  	\begin{aligned}[b]
			\ket{\psi_6}
		\;={}\;
			\ket{x}_\bX \ket{\smash{\tilde \Gamma}}_{\bW''}\ket{0}_{a'}
		\,&+\,
			\ket{x}_\bX \ket{\Phi}_{\bB,\bW,\bS'} \ket{1^{\!\!\;2B+1}}_{\bS,\bC} \ket{1}_{a'}
		\\&
		\,+\,
			\alpha \ket{x}_\bX \ket{\smash{0^{\!\!\;3S'}}}_{\bB,\bW,\bS'} \ket{1^{\!\!\;2B+1}}_{\bS,\bC} \ket{1}_{a'} ,
		\end{aligned}
	\end{equation}
	where $\ket{\Phi} \in \sB\sox{3S'}$ is a state such that $\bracket{00\cdots00}{\Phi} = 0$, and where
	\begin{equation}
	\begin{aligned}[b]
				\alpha
			\;:={}&\!
				\sum_{b \in \mathbf A_x} \bra{\big.\smash{0^{\!\!\;3S'}}} \!\!\;\tilde K\,\trans \biggl[ \ket{\big.b,w(x,b)} \ox \ket{\big.\smash{1^{2S'}}} \biggr]
			\\={}&\!
				\sum_{b \in \mathbf A_x}  \biggl[ \bra{\big.b,w(x,b)} \ox \bra{\big.\smash{1^{2S'}}} \biggr] \tilde K \ket{\big.\smash{0^{\!\!\;3S'}}}
			\\={}&
				\sum_{\substack{b \in \mathbf A_x \\ \mathclap{y \in \{0,1\}^{S'}}}}  \biggl[ \bra{\big.b,w(x,b)} \ox \bra{\big.\smash{1^{2S'}}} \biggr] \biggl[ \ket{\big.y} \ox \ket{\big.\smash{1^{2S'}}} \biggr]
			\equiv \bigl| \mathbf A_x \bigr|	\!\!\pmod{k}.
	\end{aligned}\!\!\!\!\!\!
	\end{equation}
	\end{subequations}
	Toggling each of the bits in $(\bB,\bW,\bS')$ yields the state 
	\begin{equation}
		\label{eqn:postSummationToggle}
  	\begin{aligned}[b]
			\ket{\psi_7}
		\;={}\;
			\ket{x}_\bX \ket{\smash{\tilde \Gamma}'}_{\bW''}\ket{0}_{a'}
		\;&+\;
			\ket{x}_\bX \ket{\Phi'}_{\bB,\bW,\bS'} \ket{1^{\!\!\;2B+1}}_{\bS,\bC} \ket{1}_{a'}
		\\[0.5ex]&
		\;+\;
			\alpha \ket{x}_\bX \ket{\smash{1^{\!\!\;3S'}}}_{\bB,\bW,\bS'} \ket{1^{\!\!\;2B+1}}_{\bS,\bC} \ket{1}_{a'} ,
		\end{aligned}\!\!\!
	\end{equation}
	where $\ket{\smash{\tilde \Gamma'}}$ is the result of toggling every bit in $\ket{\smash{\tilde\Gamma}}$ and $\ket{\Phi'}$ is the result of toggling every bit in $\ket{\Phi}$, so that $\bracket{11\cdots11}{\Phi'} = 0$.
	Defining the circuit $\tilde C_n$ consisting of the operations performed in steps~2--7, we may summarize the computation thus far by
	\begin{equation}
	\label{eqn:firstHalfSummary}
	\begin{aligned}[b]
			\tilde C_n \ket{x}_{\bX} \ket{\big.\smash{0^{\!\!\;5B+3m-1}}}_{\bW''\!,a'}
		&=\,
			\ket{\psi_7}
		\\&
		\,=\,
			\ket{x}_\bX \ket{\Psi'}_{\bW''\!,a'} \,+\, \alpha \ket{x}_\bX \ket{\big.\smash{1^{\!\!\;5B+3m-1}}}_{\bW''\!,a'}	\,,
	\end{aligned}
	\!\!\!
	\end{equation}
	where $\ket{x}\ket{\Psi'}$ collects the first two terms on the right-hand side of Eqn.~\eqref{eqn:postSummationToggle}, and in particular has no overlap with any state in which every bit of $\bB$, $\bW$, $\bS'$, and $a'$ is in the state $\ket{1}$.
	It follows that the result of introducing $\bC'$ in the state $\ket{0}$, flipping it to $\ket{1}$, and then conditionally flipping it again in step~8 yields the state
	\begin{equation}
	\begin{aligned}[b]
			\ket{\psi_8}
		\;=\;
			\ket{x}_\bX \,&\! \ket{\Psi'}_{\bW''\!,a'}\ket{1}_{\bC'} \,+\, \alpha \ket{x}_\bX \ket{\big.\smash{1^{\!\!\;5B+3m-1}}}_{\bW''\!,a'}\ket{0}_{\bC'}
			\,.
	\end{aligned}
	\end{equation}
	We now consider the two possible cases of \YES\ or \NO\ instances of $L$.
	
	\paragraph{Soundness.}
	If $x \notin L$, we have $|\mathbf A_x| \equiv 0 \pmod{k}$, so that $\alpha = 0$.
	Then $\ket{\psi_8} = \bigl(\tilde C_n \ket{x}_\bX\ket{00\cdots00}_{\bW''\!,a'}\bigr)\ket{1}_{\bC'}$, so that step~9 simply restores the original state of $\bW''$ and $a'$, step~10 sets all of those bits (and the output bit $a$) to $\ket{1}$, and step~11 flips the output bit $a$ back to $\ket{0}$.
	The final state is then $\ket{x}\ket{\smash{1^{\!\!\;5B+3m}}}\ket{0}$, and in particular, the output is necessarily $0$.

	\paragraph{Completeness.}
	If $x \in L$, we have $|\mathbf A_x| \equiv 1 \pmod{k}$ by hypothesis, so that $\alpha = 1$.
	Conditionally performing $\tilde C_n^{-1}$ on $\bX$, $\bW''$, and $a'$ if $\bC'$ is in the state $\ket{1}$ yields
	\begin{equation}
	\begin{aligned}[b]
			\ket{\psi_9}
		\;&=\;
			\biggl(\tilde C_n^{-1} \ket{x}_\bX \ket{\Psi'}_{\bW''\!,a'}\biggr)\ket{1}_{\bC'} \,+\, \ket{x}_\bX \ket{\big.\smash{1^{\!\!\;5B+3m-1}}}_{\bW''\!,a'}\ket{0}_{\bC'}
		\\&=\;
			\ket{x}_\bX \ket{\Psi''}_{\bW''\!,a'} \ket{1}_{\bC'} \,+\, \ket{x}_\bX \ket{\big.\smash{1^{\!\!\;5B+3m-1}}}_{\bW''\!,a'}\ket{0}_{\bC'}
	\end{aligned}
	\!\!\!\!\!\!
	\end{equation}
	for some state $\ket{\Psi''}$, as $\tilde C_n$ leaves the input register unchanged on all inputs; introducing $a$ and toggling all of the bits except for $\bX$ and $\bC'$ then yields
	\begin{equation}
	\label{eqn:penultimateCompleteness}
	\begin{aligned}[b]
			\ket{\psi_{10}}
		\;&=\;
			\ket{x}_\bX \ket{\Psi'''}_{\bW''\!,a'} \ket{1}_{\bC'} \ket{1}_a \,+\, \ket{x}_\bX \ket{\big.\smash{0^{\!\!\;5B+3m-1}}}_{\bW''\!,a'}\ket{0}_{\bC'}\ket{1}_a
	\end{aligned}
	\!\!\!\!
	\end{equation}	
	where $\ket{\Psi'''}$ is the result of flipping every bit of $\ket{\Psi''}$.	
	In light of the multiply-controlled-not operation performed in step~11, consider the overlap of $\ket{\Psi'''}$ with $\ket{11\cdots11}$.
	As $\tilde C_n$ is orthogonal we have $\tilde C_n^{-1} = \tilde C_n\trans[]$, so that
	\begin{equation}
	\!\!\!\!
	\begin{aligned}[b]
			\bracket{\big.\smash{1^{5B+3m-1}}}{\big.\Psi'''}
		&=\:\!
			\biggl(\bra{x} \ox \bra{\big.\smash{0^{3B+2m}}}\biggr)\tilde C_n\trans[] \biggl(\ket{x} \ox \ket{\big.\Psi'}\biggr)
		\\&=\:\!
			\biggl(\bra{x} \ox \bra{\big.\Psi'}\biggr)\tilde C_n \biggl(\ket{x} \ox \ket{\big.\smash{0^{5B+3m-1}}}\biggr)
		\\&=\:\!
			\biggl(\bra{x} \ox \bra{\big.\Psi'}\biggr)\biggl(\ket{x}\ket{\Psi'} + \alpha \ket{x} \ket{\big.\smash{1^{\!\!\;5B+3m-1}}} \biggr)
		=\:\!
			\bracket{\Psi'}{\Psi'} \!\!\:,
	\end{aligned}\!\!\!\!
	\end{equation}
	by Eqn.~\eqref{eqn:firstHalfSummary} and the description of $\ket{\Psi'}$.	
	However, we find in this case that $\ket{\Psi'}$ is orthogonal to itself: again from Eqn.~\eqref{eqn:firstHalfSummary} we have ${1 = \bracket{\psi_7}{\psi_7}} = \bracket{\Psi'}{\Psi'} + \alpha^2$, and as 
	$\alpha = 1$, we have $\bracket{\Psi'}{\Psi'} = 0$.
	Thus the multiply-controlled-not in step~11 has no effect on the first term on the right-hand side of Eqn.~\eqref{eqn:penultimateCompleteness}; nor does it have any effect on the second term.
	Thus the final state 
	is 
	\begin{equation}
	\label{eqn:ultimateCompleteness}
	\begin{aligned}[b]
			\ket{\psi_{11}}
		\;&=\;
			\ket{x}_\bX \biggl( \ket{\Psi'''}_{\bW''\!,a'} \ket{1}_{\bC'} \,+\, \ket{\big.\smash{0^{\!\!\;5B+3m-1}}}_{\bW''\!,a'}\ket{0}_{\bC'}\biggr)\ket{1}_a\,,
	\end{aligned}
	\end{equation}
	so that the output of $C_n$ is necessarily $1$.
	
	\medskip
	\noindent
	Thus $L \in \UnitaryP[\Z_k]$ for any $L \in \Mod[k]\P$, so that $\Mod[k]\P \subset \UnitaryP[\Z_k]$.
\end{proof}

\begin{corollary*}[Theorem~\ref{thm:mainResult}]
  For any prime power $k$, $\UnitaryP[\Z_k] = \GLP[\Z_k] = \Mod[k]\P$.
\end{corollary*}
\begin{proof}
  We have the sequence of containments
  \begin{equation}
		\UnitaryP[\Z_k] \subset \GLP[\Z_k] \subset \Mod[k]\P \subset \UnitaryP[\Z_k] ,  
  \end{equation}
  where the second and third containments are Lemmas~\ref{lemma:modularContainsModal} and~\ref{lemma:modalContainsModular}, and the rest follow from the remarks following Definition~\ref{def:modalCircuitComplexity}.
\end{proof}

\subsection{A remark on $\Z_k$-modal classes for $k$ not a prime power}
\label{sec:modalNotPrimePower}

In parts of the analysis above, we made use of some techniques which did not depend on $k$ being a prime power, but rather on the fact that all problems in $\Mod[k]\P$ can be made to have only zero or one accepting branch modulo $k$:
\begin{definition}
  For $k \ge 2$ any integer, \UP[k]\ is the set of languages $L \subset \{0,1\}^\ast$ for which there exists a function $f \in \text\#\P$ such that
  \begin{itemize}
  \item 
		$x \notin L$ if and only if $f(x) \equiv 0 \pmod{k}$, and
  \item 
		$x \in L$ if and only if $f(x) \equiv 1 \pmod{k}$.
  \end{itemize}
\end{definition}
\noindent
Our results rely on the fact that $\UP[k] = \Mod[k]\P$ for $k$ a prime power.
More generally, however, the proofs of Lemmas~\ref{lemma:affine} and~\ref{lemma:modalContainsModular} also hold for any \Mod[k]\P\ algorithm which accepts on at most a single branch modulo $k$.
Thus, the same analysis suffices to show that $\UP[k] \subset \AffineP[\Z_k]$ and $\UP[k] \subset \UnitaryP[\Z_k]$, for any integer $k \ge 2$.
Furthermore, because the $\Mod[k]\P$ algorithms for simulating $\Z_k$-affine circuits and $\Z_k$-invertible circuits accept with exactly one branch modulo $k$ when $x \in L$ (and with zero branches modulo $k$ otherwise), we immediately have $\AffineP[\Z_k] \subset \UP[k]$ and $\UnitaryP[\Z_k] \subset \GLP[\Z_k] \subset \UP[k]$.
Thus our results in fact show:
\begin{corollary}
	\label{cor:UPequiv}
  For all $k \ge 2$, we have $\GLP[\Z_k] = \AffineP[\Z_k] = \UnitaryP[\Z_k] = \UP[k]$\,.
\end{corollary}
\noindent
The characterization in terms of $\Mod[k]\P$, for $k$ a prime power, can itself be seen as a corollary of the fact that $\UP[k] = \Mod[k]\P$ in that case.
Thus, despite the differences between the valid transformations for these models, they are polynomial-time equivalent for any fixed integer $k\ge 2$.

While it is not known whether $\UP[k] = \Mod[k]\P$ when $k$ is not a prime power, this seems unlikely.
Consider the factorization $k = p_1^{e_1} p_2^{e_2} \cdots p_\ell^{e_\ell}$ for distinct primes $p_i$: it is easy to see that that $\UP[k] \subset \UP[p_1] \cap \UP[p_2] \cap \cdots \cap \UP[p_\ell]$ essentially by definition.
If $\UP[k] = \Mod[k]\P$ for any $k \ge 2$ not a prime power, by Ref.~\cite[Proposition~29]{BGH90} there would be a collapse of all classes $\Mod[p_i\!]\P$ for all primes $p_i$ dividing $k$.
This would require completely new simulation techniques to relate the acceptance conditions of the classes $\Mod[p_i\!]\P$.
We might then expect the $\Z_k$-modal complexity classes of this article to differ from $\Mod[k]\P$, when $k$ is not a prime power.

\section{Remarks}
\label{sec:remarks}

\subsection{Limitations on the power of quantum computation}


Quantum algorithms (either exact or with bounded error) are not expected to be able to efficiently decide all problems in $\Mod[k]\P$, for any modulus $k \ge 2$.
In particular, while ${\text{UNIQUE-SAT} \in \Mod[k]\P}$ for each $k\ge 2$, we also expect that $\text{UNIQUE-SAT} \notin \BQP$.
If the latter containment did hold, one could use Valiant--Vazirani~\cite{VV86} to show that $\NP \subset \BQP$, which is considered very unlikely~\cite{BBBV97,AA09}.
Our results for $S = \Z_k$ could then be taken to demonstrate that destructive interference can be \emph{very} powerful --- more powerful, for instance, than we expect quantum computation to be.
Given that this occurs for \emph{finite} rings of amplitudes, this is a conclusive rebuttal of criticisms of quantum computation on the grounds of somehow exploiting amplitudes with infinite precision~\cite{Lev00}.

The fact that we expect $\text{UNIQUE-SAT} \notin \BQP$ puts us in the ironic position of asking why the power of quantum algorithms should be so \emph{limited}.
If exact computation with $S$-modal circuits can be powerful enough to solve problems such as $\text{UNIQUE-SAT}$, for infinitely many choices of ring $S$, why should this fail for the special case $S = \C$ even when we allow computation with bounded-error?
We conjecture that this may be due to two different restrictions:
\begin{romanum}
\item
	Restricting to a ring of character $0$ (such as the complex numbers) makes it more difficult to arrange for outcomes to interfere destructively.
	This is particularly true for exact algorithms, where the ``infinite precision'' of quantum amplitudes makes exactitude a more difficult constraint for quantum algorithms to fulfil.
\item
	In the characteristic $0$ case, restricting to unitary circuit algorithms is expected to be significant.
	Bounded error \emph{invertible} (\ie~possibly non-unitary) circuits with complex amplitudes can efficiently decide any $L \in \PP$ by Aaronson~\cite{Aar04-post}, whereas we expect $\BQP \prsubset \PP$.
	(This may seem to conflict with our result that $\UnitaryP[\Z_k] \!= \GLP[\Z_k]$ for all $k \ge 2$: we remark on this in the following Section.)
\end{romanum}
Either of these constraints might imply upper bounds to $\UnitaryP[\C]$\,, and further support the intuition that $\BQP$ should not contain difficult problems such as $\text{UNIQUE-SAT}$.
An investigation of the computational power of exact $\C$-modal algorithms using invertible gates would indicate what limits on computation ``exactitude'' alone might impose on a quantum computation.
One could then investigate what further restrictions unitarity might impose beyond this.

\subsection{Complexity of modal computation in the infinite limit?}
\label{sec:infiniteLimit}

Hanson~\etal~\cite{HOSW11} propose to recover quantum computation as the limit of the models of unitary $\F_{\!\!\:p^2}$-modal circuits, for primes $p \equiv 3 \pmod{4}$ and taking the limit as $p \to \infty$. 
While it is unclear how such a limit might be taken, we may consider what other approaches one might pursue, to attempt to determine the power of quantum computation as a limit of $\Z_k$-unitary circuits.

While $\GLP[\Z_k] \subset \UnitaryP[\Z_k]$ for finite $k \ge 2$, different values of the branching gate $K$ are required for each modulus $k$, to simulate a $\Z_k$-invertible circuit by \mbox{$\Z_k$-unitary} circuits.
	The overhead to using $\Z_k$-unitary circuits to simulate $\Z_k$-invertible circuits for exact algorithms then depends on the cost of the branching gate $K$ (see page~\pageref{discn:gate-cost}), and thus on the modulus $k$ itself.
	Note that $\GLP[\C] \subset \PP$, following Adleman~\etal~\cite{ADH-1997}.
	Thus, examining how the cost of simulating $\Z_k$-invertible circuits using $\Z_k$-unitary circuits, and considering exact quantum algorithms as representing a sort of limit-infimum of exact $\Z_k$ unitary algorithms, might suggest approaches to bound $\UnitaryP[\C]$ and $\EQP$ away from $\PP$.
	
\subsubsection{Bounded-error computation over $p$-adics}
\label{sec:p-adics}

	The cyclic rings $\Z_p$, $\Z_{p^2}$, $\Z_{p^3}$, \ldots for $p$ a prime have the \emph{$p$-adic integers $\Z_{(p)}$} as an inverse limit.\footnote{%
		Given a a sequence of rings $R_1, R_2, \ldots, R_i, \ldots$ with morphisms $\varphi_i : R_{i+1} \to R_i$\,, their \emph{inverse limit} is a ring $\mathcal R$ with a collection of maps $\Phi_i: \mathcal R \to R_i$ such that $\Phi_i = \varphi_i \circ \Phi_{i+1}$ (see Mac~Lane~\cite{MacLane-1998} for more details).
		For $R_i = \Z_{p^i}$\,, let $\varphi_i$ be the canonical projection $\Z_{p^{i+1}} \to \Z_{p^i} / p^i \Z_{p^{i+1}}$\,.
		Then the map $\Phi_i: \Z_{(p)} \to \Z_{p^i}$ for the inverse limit is the truncation of infinite power series in $p$ to the $p^{i-1}$ order term.
	}
	The $p$-adic integers consist of formal power series $a = \sum_i a_i p^i$ over the prime $p$, with potentially infinitely many non-zero coefficients $a_i \in {\{0,1,\ldots,p-1\}}$,	but where addition is still performed with ``carries'' as in the usual representation of integers in base $p$.
	(See Robert~\cite{Robert-2000} for an introductory treatment.)

	The integers $\Z$ are contained as a subring of $\Z_{(p)}$.
	However, $\Z_{(p)}$ has a topology which is different from the usual ordering of $\Z$: for distinct $a,b \in \Z_{(p)}$, we define a distance measure $|a-b|_p = 1/p^{\ell}$, where $p^\ell$ is the smallest power for which the $p$-adic expansion of $a-b$ has a non-zero coefficient.
	The distance measure on $\Z_{(p)}$ allows us to consider bounded one-sided error computation for \mbox{$\Z_{(p)}$-modal} circuits.
	This distance is closely related to the canonical significance functions of Eqn.~\eqref{eqn:canonicalSignificanceFn}.
	The classes $\UnitaryP[\Z_k]$ for $k = p^r$ may then be construed to be languages with efficient (two-sided) bounded-error \mbox{$\Z_{(p)}$-unitary} algorithms, with error bound $1/p^r$.
	
	By taking an appropriate closure of the $p$-adic integers~\cite[Chapter~3]{Robert-2000}, we obtain a topological field $\C_{(p)}$ which is isomorphic as a field to the usual complex numbers.
	The difference between the topologies of $\C_{(p)}$ and $\C$ prevent an easy comparison of bounded-error modal computation with amplitudes over these rings.
	However, as $\C_{(p)}$ and $\C$ are isomorphic as fields, there can be no distinction between the two as regards \emph{exact} modal computation.
	As $\EQP \subset \UnitaryP[\C]$, we may then ask whether we may bound (or even characterize) $\EQP$ in terms of the limit of bounded-error $\Z_k$-modal circuit complexity as $k = p^r$ for $r \to \infty$.

\subsubsection{Other limits of cyclic rings}

	Another way to take limits of $\Z_k$ for increasing $k$ is to take the limit as $k$ is divisible by an increasing number of primes.
	For example, we may consider the computational power of $\Z_{k_i}$-modal circuits, for an integer sequence $1 < k_1 < k_2 < \cdots < k_i < \cdots$ such that each $k_{i+1} = k_i p_i$\,, where $p_i$ is a prime which does not divide $k_i$.
	One might then investigate how the $\Z_{k_i}$-modal circuit complexity of problems varies as $i \to \infty$.

	This approach may prove more difficult than the first approach suggested above.
	The theory of $p$-adics relies on residues which are prime powers; not all of the techniques which apply to the $p$-adics generalize to the inverse limit of the rings $\Z_{k_i}$\:\!.
	However, as the inverse limit of this sequence of rings is again a ring of characteristic zero, an analysis of $\Z_{k_i}$-unitary circuits in the $i \to \infty$ limit might provide insights (or formal results) which apply to the quantum circuit model.

\subsection{``Almost quantum'' distributions beyond finite fields}
\label{sec:almostQuantum?}

	The motivation of this work is to consider the computational power of destructive interference, in models different from quantum computation.
	To do so, we examined models of computation for transforming the $S$-valued distributions for semirings $S$, including the special case $S = \F_k$ first considered by Schumacher and Westmoreland~\cite{SW10}.
	The motivation of Ref.~\cite{SW10} is to demonstrate analogues of quantum phenomena, including classic quantum communication protocols such as superdense coding~\cite{BW-1992}, prompting them to describe distributions over $\F_k$ as ``almost quantum'' in later work~\cite{SW-2012}.
	Having formulated an abstract theory of indeterminism, we may informally consider what features make an indeterministic state-space qualitatively ``quantum-like''.
	
	Is there a sense in which a state-space can be ``quantum'', which is distinct from the presence of negatives among the amplitudes?
	Refs.~\cite{SW10,SW-2012} seem to indicate that, in communication complexity, negatives among amplitudes make available more than one basis in which to express correlations; and that this allows analogues of the quantum protocols which take advantage of such correlations.
	Perhaps ``quantumness'' in the sense of Schumacher and Westmoreland (\ie~abstract similarity to quantum mechanics) is a common feature of state-spaces over rings.
	The case of distributions over the semiring $\R_+$ corresponds to probability distributions, which we take as clearly \emph{non}-quantum.
	The case $S = \N$ (in which efficiently preparable distributions are \#\P-functions) only appear to yield powerful computational classes such as $\NP$ in the unbounded-error setting.
	These two cases have in common that $S$ is not a ring, \ie~none of the non-zero elements have negatives.
	On the other hand, any finite ring $R$ contains a copy of a cyclic ring $\Z_k$ for $k = \Char(R)$, so that (affine, invertible, or unitary) $R$-modal computation at least can solve problems in $\UP[k]$ by Corollary~\ref{cor:UPequiv}.

	Are there rings $R$ of characteristic zero (apart from dense sub-rings of $\C$), giving rise to models of bounded-error $R$-modal computation which are likely to be more powerful than randomized computation, but without containing $\UP$?
	How powerful are models of $\Z$-modal computation?
	Is there a form of modal computation apart from deterministic computation, which is less powerful than quantum computation or quantum communication in the exact setting?	
	
	
		
	\subsection*{Acknowledgements}

	I began this work at the University of Cambridge, with support from the European Commission project QCS.
	A first draft of this article was completed at the CWI in Amsterdam, with support from a Vidi grant from the Netherlands Organisation for Scientific Research (NWO) and the European Commission project QALGO.	
	I would like to thank Ronald de Wolf for suggestions on terminology and presentation, Will Jagy for remarks~\cite{Jagy-2014-SE} which contributed to the proof of Lemma~\ref{lemma:orthogonalBranchingOptor}, Jeroen Zuiddam for discussions which contributed to the ideas of Section~\ref{sec:p-adics},  and anonymous reviewers for feedback on an earlier draft of this article.


\bgroup

\egroup



\appendix
\renewcommand\thesubsubsection{\Roman{subsubsection}.}
\section{Unitarity over quadratic extensions of Galois rings}
\label{apx:quadraticExtensions}

We now describe 
how to construct non-trivial conjugation operations for certain Galois rings, including quadratic field extensions.
This allows us to define a generalized notion of unitarity, extending beyond matrices over $\C$.
We also show that for Galois rings, these conjugation operations are the only (self-inverse) ones that exist.
We indicate the concepts involved: the interested reader may consult Ref.~\cite{Wan-2003} for details.
These remarks are not essential to our analysis for cyclic rings, but do allow us to consider the conjecture of Ref.~\cite{HOSW11} regarding unitary transformations over finite fields.

The following repeats the definition in Section~\ref{sec:basicAlgebra} (following Eqns.~\eqref{eqn:innerProducts} on page~\pageref{discn:conjugation}) of ``a conjugation operation'', which may be used to define a sesquilinear inner product:
\begin{definition}
	\label{def:conjugation}
  Given a non-trivial ring $R$, a function $c: R \to R$ is a \emph{conjugation operation} if it is a self-inverse automorphism of $R$: that is, if $c \circ c = \id_R$, and if for all $r,s \in R$, we have $c(r+s) = c(r) + c(s)$ and $c(rs) = c(r) c(s)$.
  (This implies, in particular, that $c(0_R) = 0_R$ and $c(1_R) = 1_R$.)
\end{definition}
\noindent
Readers familiar with finite ring extensions will recognise that, for $R$ a Galois ring (such as  a finite field), this definition coincides with the standard notion of ``conjugation'' arising from a quadratic Galois extension. 

By recognising our notion of conjugation as a special case of that arising from Galois theory, we may characterize the possible conjugation operations for a Galois ring $R = \mathrm{GR}(k,k^e)$ for prime powers $k = p^r$ and integers $e \ge 1$.

\subsubsection{Construction}


If $e$ is even, let $B = \mathrm{GR}(k,k^{e/2})$.
Then $R$ can be obtained not only as an extension of $\Z_k$, but also as an extension $R = B{[\:\!\omega\:\!]}$ for $\omega\in R$ a formal root of an irreducible monic quadratic polynomial $g \in B[x]$.
In particular, we have $g(x) = (x-\omega)(x-\overline \omega) \in R[x]$ for some $\overline \omega \in R$ distinct from $\omega$.
Consider a conjugation operation $r \mapsto \bar r$ for $r \in R$, given by $\overline{(b_1 + b_2 \omega)} = b_1 + b_2 \:\!\overline \omega$ for $b_1, b_2 \in B$, and the resulting inner product of Eqn.~\eqref{eqn:interestingInnerProduct}.
Similarly to the usual inner product over $\C = \R[{\:\!i\:\!}]$, such an inner product satisfies $\langle \vec v, \vec v \rangle \in B$ for all vectors $\vec v$ over $R$.

The fields $R = \F_{p^2}$ considered by Hanson~\etal~\cite{HOSW11} are a special case of such quadratic extensions, in which $B = \F_p$ for $p \equiv 3 \pmod{4}$.
In that case, in analogy to $\C = {\R[\:\!i\:\!]}$, one may take $R = {B[\:\!i\:\!]}$, where $i$ is a formal root to the polynomial $x^2 + 1$.
Then $\overline{(a+bi)} = a - bi$ for $a,b \in B$, and the notion of unitarity as described above bears a strong formal similarity to unitarity over $\C$.
The restriction on $B$ in this case serves to guarantee that $x^2 + 1$ is irreducible in $B[x]$.
In our analysis, if $x^2 + 1$ is not irreducible in $B[x]$, one simply takes $R$ to be an extension $B{[\:\!\omega\!\:]}$ for $\omega$ the formal root of some other quadratic polynomial which is irreducible over $B$.

For any Galois ring $R$, including those which may be obtained as a quadratic extension of some ring $B$, we may also consider the trivial automorphism $\overline r = r$ for all $r \in R$, giving rise to an inner product of the form of Eqn.~\eqref{eqn:boringInnerProduct}.
For some Galois rings, this leads in principle to two different notions of ``unitarity'', and two different models of unitary computation.
For the sake of definiteness, we consider any Galois ring $R$ to come equipped with some self-inverse ring automorphism $r \mapsto \bar r$: this induces a corresponding inner product according to Eqn.~\eqref{eqn:interestingInnerProduct}, and a specific notion of unitarity.

\subsubsection{Characterization}

Conjugation operations on $R$ are elements of the automorphism group of $R$, which are well-understood in the case that $R$ is a Galois ring.
In general, the automorphisms of a Galois ring $R = \mathrm{GR}(k,k^e)$ form a cyclic group of order $e \ge 1$~\cite[Theorem~14.31]{Wan-2003}, generated by some particular automorphism $\phi$.
\begin{itemize}
\item 
	For $e$ even, the automorphism $c: R \to R$ given by $c = \phi^{\!\;e\!\!\;/2}$ is self-inverse.
	This is the only non-trivial automorphism which satisfies $c^2 = \id_R$.
	In particular, $c$ is the conjugation operation arising from the quadratic extension of $B = \mathrm{GR}(k,k^{e/2})$, whose construction we described above.
\item
	For $e$ odd, the only self-inverse automorphism of $R$ is the identity operation $\id_R$.
	Thus, the only possible conjugation operation is the trivial one.
\end{itemize}
For Galois rings $R$, this characterizes all conjugations in the sense of Definition~\ref{def:conjugation}.

\section{Regarding state-spaces for finite rings}
\label{apx:validTransformations}


\subsection{Lemmata concerning state spaces}

The following two Lemmas serve to justify Definition~\ref{def:generalizedGeneric} (page~\pageref{def:generalizedGeneric}):

\begin{lemma}
	\label{lemma:genericStateSpace}
  For a ring $S \ne 0$, let $\sS_\ast$ be the set of distributions $\ket{\psi} \in \sD$ for which there exists $\ket{\phi} \in \sD$ such that $\bracket{\phi}{\psi} = 1_S$.
  Then $\sS_\ast$ is a state-space.
\end{lemma}
\begin{proof} We show that $\sS_\ast$ satisfies the criteria of Definition~\ref{def:stateSpace}:\\[-3ex]
	\begin{itemize}
	\item 
		Clearly $\sS_\ast$ contains $\ket{x}$ for each $x \in \{0,1\}^\ast$, and excludes $\vec 0$.
	\item
		Let $\ket{\alpha},\ket{\beta} \in \sS_\ast$\,, and let $\ket{\alpha'},\ket{\beta'} \in \sD$ satisfy $\bracket{\alpha'}{\alpha} = \bracket{\beta'}{\beta} = 1_S$.
		Then $\ket{\psi} = \ket{\alpha}\ox\ket{\beta} \in \sD$ and $\ket{\phi} = \ket{\gamma}\ox\ket{\delta} \in \sD$ satisfy $\bracket{\phi}{\psi} = \bracket{\alpha'}{\alpha} \bracket{\beta'}{\beta} = 1_S$, so that $\ket{\alpha} \ox \ket{\beta} \in \sS_\ast$\,.
	\item
	Finally, suppose that $\ket{\psi} = \ket{\alpha}\ox\ket{\delta} \in \sS_\ast$ for some $\ket{\alpha} \in \sS_\ast$ and $\ket{\delta} \in \sD$, and let $\ket{\psi'},\ket{\alpha'} \in \sD$ be such that $\bracket{\psi'}{\psi} = \bracket{\alpha'}{\alpha} = 1_S$.
	Consider the distribution
	\begin{equation} 
		\ket{\psi''} = \Bigl(\ket{\alpha'}\bra{\alpha} \ox \id\Bigr) \ket{\psi'},
	\end{equation}
	where $\id$ is an identity operation of the correct width to make this well-defined.
	By hypothesis, we have
	\begin{equation}
	  \bracket{\psi''}{\psi}
	  \;=\;
	  \bra{\psi'} \Bigl( \ket{\alpha}\bra{\alpha'} \ox \id \Bigr) \ket{\alpha}\ket{\delta}
	  \;=\;
	  \bra{\psi'} \Bigl( \ket{\alpha} \ox \ket{\delta} \Bigr)
	  \;=\;
		1_S;
	\end{equation}
	then $\ket{\delta'} = \bigl(\bra{\alpha} \ox \id\bigr) \ket{\psi''}$ satisfies $\bracket{\delta'}{\delta} = 1_S$, so that $\ket{\delta} \in \sS_\ast$ as well. \qedhere
	\end{itemize}
\end{proof}
\begin{lemma}
	\label{lemma:ell-1-stateSpace}
  For a ring $S \ne 0$, let $\sS_1$ be the set of distributions $\ket{\psi} \in \sD$ for which $\sum\limits_x \psi_x = 1_S$.
  Then $\sS_1$ is a state-space.
\end{lemma}
\begin{proof}
  We show that $\sS_1$ satisfies the criteria of Definition~\ref{def:stateSpace}:\\[-3ex]
  \begin{itemize}
  \item
		Clearly $\sS_1$ contains $\ket{x}$ for each $x \in \{0,1\}^\ast$, and excludes $\vec 0$.
	\item
		Let $\ket{\alpha},\ket{\beta} \in \sS_1$.
		Then $\ket{\psi} = \ket{\alpha} \ox \ket{\beta} \in \sD$ satisfies
		\begin{equation}
		  \sum_{\mathclap{x \in \{0,1\}^\ast}} \;\psi_x \;=\;\; \sum_{\mathclap{y,z \in \{0,1\}^\ast}} \; \alpha_y \beta_z \;=\; \biggl(\sum_{y \in \{0,1\}^\ast} \!\!\!\alpha_y\biggr) \biggl(\sum_{z \in \{0,1\}^\ast} \!\!\!\beta_z\biggr) \;=\; 1_S:
		\end{equation}
		where this factorization follows because only finitely many of the $\alpha_y$ and $\beta_z$ are non-zero.
	\item
		Finally, suppose that $\ket{\psi} = \ket{\alpha}\ox\ket{\delta} \in \sS_1$ for some $\ket{\alpha} \in \sS_1$.
		Then 
		\begin{equation}
				\sum_{z \in \{0,1\}^\ast} \!\!\!\delta_z
			\;=\;
				\biggl(\sum_{y \in \{0,1\}^\ast} \!\!\!\alpha_y\biggr) \biggl(\sum_{z \in \{0,1\}^\ast} \!\!\!\delta_z\biggr)
			\;=\;
				\sum_{\mathclap{x \in \{0,1\}^\ast}} \;\psi_x 
			\;=\;
				1_S,
		\end{equation}
		so that $\ket{\delta} \in \sS_1$ as well. \qedhere
  \end{itemize}

\end{proof}

\begin{lemma}
	\label{lemma:ell-2-stateSpace}
  For a ring $S \ne 0$, let $\sS_2$ be the set of distributions $\ket{\psi} \in \sD$ for which $\bracket{\psi}{\psi} = 1_S$.
  Then $\sS_2$ is a state-space.
\end{lemma}
\begin{proof} We show that $\sS_2$ satisfies the criteria of Definition~\ref{def:stateSpace}:\\[-3ex]
	\begin{itemize}
	\item 
		Clearly $\sS_2$ contains $\ket{x}$ for each $x \in \{0,1\}^\ast$, and excludes $\vec 0$.
	\item
		Let $\ket{\alpha},\ket{\beta} \in \sS_2$\,.
		Then $\ket{\psi} = \ket{\alpha}\ox\ket{\beta} \in \sD$ satisfies $\bracket{\psi}{\psi} = \bracket{\alpha}{\alpha} \bracket{\beta}{\beta} = 1_S$, so that $\ket{\alpha} \ox \ket{\beta} \in \sS_2$\,.
	\item
	Finally, suppose that $\ket{\psi} = \ket{\alpha}\ox\ket{\delta} \in \sS_2$ for some $\ket{\alpha} \in \sS_2$.
	Then $\bracket{\delta}{\delta} = \bracket{\alpha}{\alpha} \bracket{\delta}{\delta} = \bracket{\psi}{\psi} = 1_S$, so that $\ket{\delta} \in \sS_2$ as well. \qedhere
	\end{itemize}
\end{proof}

\subsection{Lemma concerning valid transformations}

The following Lemma justifies Proposition~\ref{lemma:validTransformations} (page~\pageref{lemma:validTransformations}):

\begin{lemma}
\label{lemma:validTransformations-a}
Let $S$ be a finite commutative ring with $\Char(S) = k > 0$.
\begin{romanum}
\item
	The valid transformations of $\sS_\ast$ are all left-invertible transformations of $\mathscr D$.
\item
	The valid transformations of $\sS_1$ are all transformations $T: \mathscr D \to \mathscr D$ for which $\sum\limits_y \bra{y} T \ket{x} = 1_S$ for each $x \in \{0,1\}^\ast$.
\item
	If $S$ is a Galois ring with character $k = p^r$, where $p \ge 2$ is prime and $r \ge 1$, the valid transformations of $\sS_2$ are all transformations $T: \sD \to \sD$ such that $T\herm T \equiv \id_\sD \pmod{p^{\lceil r/2 \rceil}}$ if $p$ is odd, or $T\herm T \equiv \id_\sD \pmod{2^{\lceil (r-1)/2 \rceil}}$ if ${p = 2}$.
	For $S$ any finite field or cyclic ring of odd order, we in fact have $T\herm T = \id_\sD$.
	Conversely, all operators $T: \sD \to \sD$ for which $T\herm T = \id_\sD$ are valid transformations of $\sS_2$.
\end{romanum}
\end{lemma}
\begin{proof}
  For part \parit{i}, suppose that $T$ is a valid transformation for $\sS_\ast$, and suppose that $T \ket{\psi} = \vec 0$ for some distribution $\ket{\psi} \in \sD$. 
  Consider $A = \supp(\ket{\psi})$, the set of strings $x \in \{0,1\}^\ast$ for which $\psi_x \ne 0$, with a finite enumeration $A = \{\alpha_1, \alpha_2, \ldots, \alpha_{|A|}\} \subset \{0,1\}^\ast$ in lexicographic order.
  (Note that $A$ will be finite: as $\sD$ is defined as a direct sum in Definition~\ref{def:distributionSpace}, $\ket{\psi}$ has only finitely many non-zero coefficients.)
  Similarly, let
  \begin{equation}    
		B = \bigcup_{x \in A} \supp(T\ket{x})  \,,
  \end{equation}
	with an enumeration $B = \{\beta_1, \beta_2, \ldots, \beta_{|B|}\}$.
  Consider a matrix $M\sur{0}: S^{|A|} \to S^{|B|}$ with coefficients $M\sur{0}_{h,j} = \bra{\beta_h} T \ket{\alpha_j}$.
	We extend this matrix iteratively by adding rows and performing row-reductions, as follows:
	\begin{itemize}
	\item
		For each $\ell \ge 1$, let $M^{\prime\,(\ell)}$ be the matrix $M\sur{\ell-1}$ extended by one row, such that
		\begin{equation}
			M^{\:\!\prime\,(\ell)}_{|B|+\ell,j} = \bra{\phi_\ell} M\sur{\ell-1} \ket{j},   
		\end{equation}
		for any $\ket{\phi_\ell}$ subject to $\bra{\phi_\ell} M\sur{\ell-1} \ket{\ell} = 1$.
		Then, in particular, $M^{\prime\,(\ell)}_{|B|+\ell,\ell} = 1$.
	\item
		Let $M\sur{\ell}$ be the matrix obtained from $M^{\prime\,(\ell)}$ by elementary row operations, in which $M\sur{\ell}_{j,j} = 1$ for each $1 \le j \le \ell$, and all other coefficients of the first $\ell$ columns are zero.
	\end{itemize}
	This ultimately yields a matrix $M\sur{|A|}$ which consists of an $|A| \x |A|$ identity matrix, together with $|B|$ additional rows of zeros.	
	Furthermore, the operations of adjoining rows to each $M\sur{\ell-1}$ and the row operations on each $M^{\prime\,(\ell)}$ involve only left-invertible operations.
	Together, these operations compose to give a left-invertible transformation $L$ such that $LM\sur{0} = M\sur{|A|}$.
	As $\Null(M\sur{|A|}) = \vec 0$, it follows that $\Null(M\sur0) = \vec 0$ as well, so that $\ket{\psi} = \vec 0$.
	Then $\Null(T) = \vec 0$, so that $T$ itself is left-invertible.

	For \parit{ii}, it is clear that if $T: \sD \to \sD$ is a valid transformation, then $\sum_y \bra{y} T \ket{x} = 1_S$ for each $x \in \{0,1\}^\ast$.
	For the converse, suppose that the above equality holds for all $x \in \{0,1\}^\ast$, and let $\ket{\psi} \in \sS_1$.
	Then we have
	\begin{equation}
	\begin{aligned}[b]
	  \sum_{\mathclap{y \in \{0,1\}^\ast}} \;\bra{y} T \ket{\psi}
	\;&=\;
	  \sum_{\mathclap{y \in \{0,1\}^\ast}} \;\bra{y} \biggl( \sum_{x \in \{0,1\}^\ast} \!\!\!\! \psi_x \,T \ket{x} \biggr)
	\\&=\;
	  \sum_{\mathclap{x \in \{0,1\}^\ast}} \;\psi_x \;\biggl(\sum_{y \in \{0,1\}^\ast} \!\!\!\bra{y} T \ket{x} \biggr)
	 \;=\;
	  \sum_{\mathclap{x \in \{0,1\}^\ast}} \;\psi_x \cdot 1_S
	 \;=\;
	  1_S,
	\end{aligned}
	\end{equation}
	where once again we may exchange the sums because the number of non-zero terms in each case is finite.
	Thus $T$ preserves $\sS_1$.

	For \parit{iii}, 
	let $k = p^r$ for some prime $p$ and some $r \ge 1$.
	Let $T$ be a valid transformation of $\sS_2$: then we have $\bra{x} T\herm T \ket{x} = 1$ for all $x \in \{0,1\}^\ast$.
	Let $x,y \in \{0,1\}^n$ be distinct strings for any $n \ge 1$, and write $\ket{\varphi_x} = T\ket{x}$ and $\ket{\varphi_y} = T\ket{y}$ for the sake of brevity.
	Let $\varepsilon = \bra{y} T\herm T \ket{x}$: then $\bar\varepsilon = \bigl( \bra{y} T\herm T \ket{x} \bigr){}\herm = \bra{x} T\herm T \ket{y}$.
	Define a distribution $\ket{\sigma}	\in \sB\sox{n+\:\!\!\lceil \log(k) \rceil}$ by
  \begin{align}
    \ket{\sigma}	\,&=	\;\!\Bigl(\ket{x}\;\!\!\ket{0} + \ket{x}\;\!\!\ket{1} + \ket{x}\;\!\!\ket{2} + \cdots + \ket{x}\;\!\!\ket{k-1} \Bigr)\;\!\! + \ket{y}\;\!\!\ket{0} ,
  \end{align}
 	where the second tensor factor represents integers in binary.
  It is easy to verify that $\bracket{\sigma}{\sigma} = (k+1)1_S = 1_S$, so $\ket{\sigma}\in \sS_2$.
  Then we have
  \begin{align}
		\!\!
				1_S \,&=\, \bra{\sigma}\bigl(T\herm \ox \idop\bigr)\bigl(T \ox \idop\bigr)\ket{\sigma}
      \notag\\&=\,
				k\bracket{\varphi_x}{\varphi_x} + \bracket{\varphi_x}{\varphi_y} + \bracket{\varphi_y}{\varphi_x} + \bracket{\varphi_y}{\varphi_y}
      \,=\,
				1_S + \varepsilon + \bar\varepsilon ,
		\!\!
  \end{align}
  so that $\bar\varepsilon = -\varepsilon$.
  Then, consider the distribution $\ket{\psi}	\in \sB\sox{n+\:\!\!\lceil \log(k) \rceil}$ given by
  \begin{equation}
    \ket{\psi}	\,=\, \varepsilon\Bigl(\ket{x}\;\!\!\ket{1} + \ket{x}\;\!\!\ket{2} + \ket{x}\;\!\!\ket{3} + \cdots + \ket{x}\;\!\!\ket{k-1} \Bigr) + (1_S+\varepsilon)\ket{y}\;\!\!\ket{1}.    
  \end{equation}
  As $\bar\varepsilon = -\varepsilon$, one may verify that $\bracket{\psi}{\psi} = (k-1) \bar\varepsilon \varepsilon + (1_S+\bar\varepsilon\varepsilon) = 1_S$, so that $\ket{\psi} \in \sS_2$.
  Then we have
	\begin{equation}
  \begin{aligned}[b]
				1_S &= \bra{\psi}\bigl(T\herm \!\;{\ox}\;\! \idop\bigr)\bigl(T \!\;{\ox}\;\! \idop\bigr)\ket{\psi}
      \\[1ex]&=
				(k-1)\bar\varepsilon\varepsilon \bracket{\varphi_x}{\varphi_x} + \bar\varepsilon(1_S+\varepsilon)\bracket{\varphi_x}{\varphi_y}
			\\&\phantom{{}= (k-1)\bar\varepsilon\varepsilon \bracket{\varphi_x}{\varphi_x}}
				+ \varepsilon(1_S+\bar\varepsilon)\bracket{\varphi_y}{\varphi_x} + (1_S+\bar\varepsilon\varepsilon)\bracket{\varphi_y}{\varphi_y}
      \\[1ex]&=
				(k-1)\bar\varepsilon\varepsilon + (1_S+\varepsilon)\bar\varepsilon\varepsilon + (1_S-\varepsilon)\bar\varepsilon\varepsilon + (1_S + \bar\varepsilon\varepsilon)
			\\[1ex]&=
				(k+2)\bar\varepsilon\varepsilon + 1_S \;=\; 1_S - 2 \varepsilon^2,
  \end{aligned}
 	\end{equation}
 	so that $2\varepsilon^2 = 0_S$.
 	If $k = 2^r$ is even ($p = 2$), this implies that $\varepsilon^2 \in 2^{r-1} S$, due to the structure of the zero divisors in the Galois ring $S$.
 	Then $\varepsilon \in 2^{\lceil (r-1)/2 \rceil} S$, or equivalently $\varepsilon \equiv 0 \pmod{2^{\lceil (r-1)/2 \rceil}}$.
 	Otherwise, if $k$ is odd ($p \ne 2$), $2$ has a multiplicative inverse; we then have $\varepsilon^2 = 0_S$.
 	Again, due to the structure of the zero divisors in $S$, we have $\varepsilon = p^{\lceil r/2 \rceil} S$, or $\varepsilon \equiv 0 \pmod{p^{\lceil r/2 \rceil}}$.
 	Thus, for all $x,y \in \{0,1\}^\ast$, we have $\bra{x} T\herm T \ket{y} \equiv 0 \pmod{p^\tau}$, where $\tau = \lceil (r-1)/2 \rceil$ if $p = 2$ and $\tau = \lceil r/2 \rceil$ otherwise.
\end{proof}
\noindent

\paragraph{Remark.}
In the statement of \parit{iii} above, rings for which $\Char(S) = 2$ are an important special case in which the Lemma trivializes: we have $T\herm T \equiv \idop_\sD \pmod{2^0}$, which imposes no constraints on the difference $T\herm T - \idop_\sD$.
For instance, in the case of the field $\F_2$, it is easy to show that $\bracket{\psi}{\psi} = \sum_x \psi_x$ for $\ket{\psi} \in \mathscr D$ when $S = \F_2$.
Then $\sS_1 = \sS_2$ in this case, and the non-invertible $\textrm{\small ERASE} = \begin{bmatrix} 1 & 1 \end{bmatrix}$ operation is valid despite not being unitary.
On the other hand, if $k > 2$ is itself prime, we have $\bra{x} T\herm T \ket{y} = 0_S$ for $x \ne y$, as $\lceil r/2 \rceil = r = 1$ in that case.
For such $k$ we then have $T\herm T = \id_\sD$, without any congruences.

\section{\EQP\ versus quantum meta-algorithms}
\label{apx:EQPthesis}

Bernstein and Vazirani~\cite{BV97} defined the class $\EQP$ as those problems which can be solved exactly on a quantum Turing machine in polynomial time.
Quantum Turing machines are a variation on the randomized Turing machine, but where transitions are described by (computable) $\ell_2$-normalized distributions over $\C$ rather than $\ell_1$-normalized distributions over $\R_+$.
Bernstein and Vazirani~\cite{BV97} also show that $\EQP$ is equivalent to the set of problems which (in the language of Section~\ref{sec:modalCircuitComplexity}) are solvable by polytime-uniform circuit families with constant-time specifiable gates: \ie~all circuits in the family may be constructed from a single finite gate set.

Many standard models of computation (\eg~depth-bounded boolean logic circuits) can be adequately defined with a finite gate-set.
We argue 
that limiting the theory of quantum algorithms to circuits constructed from finite (or \emph{constant size}) gate sets is an undue restriction.
Our main contention is that restricting to finite gate-sets introduces a distinction between quantum algorithms and quantum ``meta-algorithms''.
Such a distinction does not exist in classical computation; introducing it for quantum computation neither reflects the practical aspirations for building a quantum computer, nor the purpose of research in computational theory in general.

We argue for \emph{polynomial-time specifiability} of gates in Section~\ref{sec:modalCircuitComplexity} as a reasonable constraint on circuit families in the study of quantum computation, as well as other indeterministic models.
In analogy to Definition~\ref{def:modalCircuitComplexity}, we may define $\UnitaryP[\C]$ to be the analogue of $\EQP$ in which polytime-specifiable gate-sets are allowed rather than only constant-time specifiable gate-sets. %
We advocate $\UnitaryP[\C]$ as a class whose study may bear more fruit than the study of $\EQP$ has.

\subsection{On $\EQP$ versus $\UnitaryP[\C]$}

We first remark on our comment on page~\pageref{discn:EQP-nonrobust} the relation between $\EQP$ and $\UnitaryP[\C]$, in which we predict the following:
\begin{conjecture*}
  The containment $\EQP \subseteq \UnitaryP[\C]$ is strict.
\end{conjecture*}
\noindent
The basis of this conjecture is simply that polytime-specifiable gate-sets compose to form (families of) transformations for which no exact decomposition is possible for finite gate sets.
A simple example is the family of quantum Fourier transforms $\{ F_{2^n} \}_{n \ge 1}$ over the cyclic rings $\Z_{2^n}$,
\begin{equation}
  F_{2^n}	=	\frac{1}{\sqrt{2^n}} \begin{bmatrix}
         	 	                        1 & 1 & 1 & \cdots & 1 \\[0.5ex]
																		1 & \omega_n	&	\omega_n^2 & \cdots & \omega_n^{-1}		\\[0.5ex]
																		1 & \omega_n^2 & \omega_n^4 & \cdots & \omega_n^{- 2}	\\[0.5ex]
																		\vdots & \vdots & \vdots & \ddots & \vdots \\[0.5ex]
																		1 & \omega_n^{-1} & \omega_n^{-2} & \cdots & \omega_n
         	 	                      \end{bmatrix}\!, \quad\text{where $\omega_n = \e^{2\pi i/2^n}$.}
\end{equation}
These may be expressed as a polytime-uniform circuit family over polytime-specifiable gates, using the recursive decomposition due to Coppersmith~\cite{Coppersmith-1994,NC} into the gates
\begin{equation}
	\mspace{-18mu}
		H = \frac{1}{\sqrt 2} \begin{bmatrix} 1 & \!\!\phantom{-}1 \\ 1 & \!\!-1 \end{bmatrix}\!\!,\quad
		\textrm{\small SWAP} = \begin{bmatrix} 1 & 0 & 0 & 0 \\ 0 & 0 & 1 & 0 \\ 0 & 1 & 0 & 0 \\ 0 & 0 & 0 & 1 \end{bmatrix}\!\!,\quad
		\mathrm{CZ}^{1\!\!\:/\!\!\;2^t} \!= \begin{bmatrix} 1 & 0 & 0 & \!\!0 \\ 0 & 1 & 0 & \!\!0 \\ 0 & 0 & 1 & \!\!0 \\ 0 & 0 & 0 & \!\!\omega_{t+1} \end{bmatrix}\!\!.
	\mspace{-12mu}
\end{equation}
However, $F_{2^n}$ cannot be generated for all $n\ge1$ using a single finite gate-set.
Let $\overline\Q$ be the algebraic closure of $\Q$.
Consider any gate-set $\mathcal U$: representing the coefficients by elements of $\overline\Q(\tau_1, \tau_2, \ldots, \tau_t)$ for some finite list of independent transcendentals $\tau_j$, the unitarity constraints can only be satisfied if all contributions from the transcendentals $\tau_j$ formally cancel out.
Similar remarks apply to any composition of gates representing $F_{2^n}$, as the latter has only algebraic coefficients.
By replacing every transcendental $\tau_j$ with zero in the expression of each gate-coefficient in $\mathcal U$, we may obtain an algebraic gate-set $\mathcal U_\alpha$ which also generates $F_{2^n}$.
For a \emph{finite} algebraic gate-set $\mathcal U_\alpha$, consider the finite-degree field extension $\mathbb U$ obtained by extending $\Q$ by each of the coefficients of the gates in $\mathcal U_\alpha$.
Because the unitaries $F_{2^n}$ contain coefficients of unbounded degrees as $n \to \infty$, a finite gate-set can only generate finitely many of the Fourier transforms $F_{2^n}$.
Thus, it seems likely that by using polytime-specifiable gate sets, one might construct circuits to decide languages $L$ which may be difficult to solve using finite gate sets.

As $\P \subset \EQP \subset \UnitaryP[\C] \subset \PSPACE$, separating $\EQP$ from $\UnitaryP[\C]$ would imply a separation of $\P$ from $\PSPACE$, and so might be considered difficult to achieve.
However, given that the quantum Fourier transform plays a celebrated role in quantum information theory (most notably in Shor's algorithm~\cite{Shor97}), it seems very likely that there are problems in $\UnitaryP[\C]$ which can be solved using quantum Fourier transforms, and which have no obvious solutions without them.

\subsection{Infinite gate-sets and the goals of computational theory}

Given the role of quantum Fourier transforms over cyclic rings in quantum algorithms, the fact that $\EQP$ cannot make use of an infinite family of them is a provocative state of affairs.
This may be taken as evidence that limiting quantum algorithms to finite gate-sets closes off what could be a fruitful field of study in exact quantum algorithms.

There is no finite universal gate set for quantum computation, in the sense of providing an exact decomposition of arbitrary unitary operations.
This is equivalent to the fact that there is no universal quantum Turing machine, in contrast to the classical setting.
As a quantum Turing machine $\mathbf Q$ only has a finite number of transitions, the unitary evolutions describing its behaviour can be described by a finite-dimensional field extension of the rational numbers $\Q$.
As $\overline\Q$ is an infinite field extension of $\Q$, there are then infinitely many unitary transformations over the complex algebraic numbers $\overline\Q$ which $\mathbf Q$ cannot exactly simulate.
Equivalently, given any finite unitary gate-set $\mathcal U$, there are unitary transformations $W \in \mathop{\mathrm U}(2^N)$ with algebraic coefficients which cannot be exactly simulated by a circuit over $\mathcal U$, in the sense that there do not exist $U_1, U_2,\ldots, U_T \in \mathcal U$ such that $U_1 U_2 \cdots U_T \ket{x}\ket{0} \propto (W\ket{x})\ox\ket{0}$ for all $x \in \{0,1\}^N$.

The definition of \EQP\ in terms of quantum Turing machines is motivated by traditional concerns of computational complexity.
However, these motivations have consequences which contradict what could be construed to be the objective of computational theory, which is the analysis of computable transformations of information.
Considering only finite gate-sets has the effect of introducing computational distinctions which are meaningless in the classical regime, as we describe in the next Section.
First, however, we must consider whether such a restriction is necessary to the study of computational complexity, \emph{a priori}.

Computational complexity theory is the study of the structure of algorithms, by decomposition into simpler transformations.
We very often require that the decomposition of an algorithm into transformations makes use of only a finite list of simple transformations, and that each primitive transformation requires only a constant amount of any resource (such as time or work space).
However, there are cases where the simple transformations range over an infinite set, and may not be physically realised in constant time.

Consider the circuit complexity classes $\AC^k \subset \P$: this class consists of those functions which may be computed by logspace-uniform circuit families $\{ C_n \}_{n \ge 1}$, where $C_n$ has size $O(\poly n)$ and depth $O(\log^k n)$-depth, and is constructed from NOT gates, and OR and AND gates with unbounded fan-in and fan-out.
The set of gates available to circuits $C_n$ in such a circuit family grows with $n$: in other words, the family does not use a finite gate-set.
We would not expect to be able to physically realise gates with fan-in $\Theta(n)$, or even $\Theta(\log n)$, with a constant amount of resources: however, we can impose modest upper bounds to simulate such gates in other models of computation.
What motivates the definition of the classes $\AC^k$ is not that they are physically realisable, but rather the study of the structure of parallel algorithms.
Thus, computational complexity is not in principle restricted to constructions from finitely many elements, but is concerned with exploring decompositions of complex functions into simpler ones.
What makes $\AC^k$ a reasonable complexity class to consider is that the allowed primitive operations, while perhaps not realistic, are also not extravagant.

Finite gate-sets suffice for a robust theory of bounded-error quantum computation.
However, there is no reason why we should study algorithms in such a way that \emph{exact} quantum computation trivializes, so long as we ensure that the theory of bounded-error computation remains meaningful.
For instance, we may impose costs on each gate corresponding to the effort to compute its coefficients (as proposed in Section~\ref{sec:modalCircuitComplexity}): this is analogous to how gates of unbounded fan-in may be treated as having non-constant depth in the analysis of $\NC^k$ algorithms.
This unitary gate cost directly represents the complexity of simulating such gates by the branching of nondeterministic Turing machines, using existing techniques~\cite{ADH-1997}: this is not a bound in terms of a reasonable model of computation, but does at least impose a bound which already applies to circuit families constructed from \emph{finite} gate-sets.
By restricting to circuit families with total cost $O(\poly n)$, we ensure that the power of such circuit families is not excessive, and leave open the possibility that circuits may still solve problems with bounded-error and with a smaller gate-cost.

Considering quantum algorithms from a theoretical standpoint, we necessarily abstract away some of the practical difficulties in realising quantum computers.
We set aside the conceit that we act only as assistants to heroic engineers, and try to determine just how much computational power we might wrest from a machine, which conforms to specifications of our choosing.
Our ongoing dialogue with engineers informs our choice of the specifications, but it is not the goal of mathematical theory to exclusively adhere to the practical limitations of the day.
We should not limit our theoretical scope only to what might seem practicable in the next five or fifty years; we should instead choose definitions which provide the greatest insight. 

\subsection{On uniformity and quantum ``meta-algorithms''}

To pursue a robust theory of exact quantum algorithms, it appears that we should allow quantum algorithms to use at least some kind of infinite gate set.
It remains to determine what limits on those gate sets are productive to an informative theory.
What sort of quantum operations might we permit, in light of the motivations and existing theory for quantum computation?

While it is conventional to construct unitary circuits over one of a small number of known ``approximately universal'' finite gate-sets, quantum circuit families may make use of any finite computable gate-set, by definition~\cite{BV97}.
Thus, no particular computable gate (or finite set of them) should theoretically be considered extravagant, however awkward to implement.
Absent the spectre of engineering difficulties looming over our definitions, if we do not limit circuit families to gates from a finite list, the simplest limitation to impose on unitary gates is that there be an efficient algorithm to to specify them.
In particular, in a ``polynomial-time'' quantum algorithm, it should not be necessary to take more than polynomial time (in the length of an appropriate input) to express the coefficients of its gates.
If we consider a gate-set $\mathcal U = \{T_1, T_2, \dots, T_\ell, \ldots \}$ and take the input to be the label $\ell$ of the gate $T_\ell$ to be described, we recover the notion that gate-sets should be \emph{polynomial-time specifiable} described in Section~\ref{sec:modalCircuitComplexity}.
Note that bounded-error computation in this model is $\BQP$, following the argument of Section~\ref{sec:modalComputationConventionalExamples}: as allowing polytime-specifiable gate sets does not affect the class of problems efficiently solvable with bounded-error by quantum algorithms, we take this as evidence that doing so is not computationally extravagant.

Whether or not polynomial-time specifiable gate-sets seem powerful, one may object that these still do not represent quantum algorithms, unless the gate-set is actually finite and simulatable on a single quantum Turing machine.
To do so, however, is to introduce a distinction between quantum algorithms and quantum ``meta-algorithms''.

We may describe a meta-algorithm as follows.
For some language $L$, suppose that we have a classical deterministic Turing machine $\mathbf S$ which, on input $x \in \{0,1\}^n$, computes a function $t(n)$ encoding another Turing machine $\mathbf T_n$ which decides whether \mbox{$x \in L$}.
The Turing machine $\mathbf S$ embodies a \emph{meta-algorithm} for $L$: a procedure to determine, for any input $x \in \{0,1\}^\ast$, some other procedure which would suffice to decide whether $x \in L$.
Because there is a universal Turing machine, one may consider a Turing machine $\mathbf U$ which simulates $\mathbf S$, and subsequently $\mathbf T_n$, on any input $(1^n,x) \in \{0,1\}^{2n}$.
Thus, in the theory of classical computation, there is no meaningful distinction between an algorithm and a meta-algorithm.

One may similarly consider a quantum Turing machine $\tilde{\mathbf S}$ which, on input $x \in \{0,1\}^n$, computes a function $q(n)$ which encodes another quantum Turing machine $\mathbf Q_n$ which decides whether $x \in L$ for $x \in \{0,1\}^n$: this represents a \emph{quantum meta-algorithm}.
Precisely because there is no universal quantum Turing machine, it is possible in the theory of quantum computation to introduce a distinction between an algorithm and a meta-algorithm.
A quantum Turing machine $\tilde{\mathbf S}$ may compute a description of another quantum Turing machine, which $\tilde{\mathbf S}$ is unable to simulate: furthermore, it is possible for $\tilde{\mathbf S}$ to compute descriptions of a family quantum Turing machines $\{\mathbf Q_n\}_{n \ge 0}$ which are varied enough that no single quantum Turing machine may simulate them all.
(Indeed, this is true even if $\tilde{\mathbf S}$ is a deterministic Turing machine.)

We argue that to distinguish between ``quantum algorithms'' and ``quantum meta-algorithms'' is introduce a distinction which does not exist elsewhere in the theory of computation, and to ignore a class of efficiently computed specifications of how decision problems may be solved by quantum computers.
By indicating a quantum algorithm to decide a problem for a given instance size, a quantum meta-algorithm provides complete information to solve a computational problem, provided adequate computational resources.

By abstracting the programme to actually construct quantum computers, we obtain a theoretical motivation to admit quantum meta-algorithms as quantum algorithms.
Indeed, the very pursuit of quantum computation as a practical technology presumes that quantum meta-algorithms are reasonable approaches to solving difficult problems.
A programme to construct a quantum computer in order to solve problems can be described, in outline, as follows:
\begin{enumerate}
\item 
	Compute a specification $S$ of some quantum computational device $D$.
\item
	\emph{Construct the device $D$} according to the plan specified by $S$.
\item
	Use the quantum device $D$ to decide instances of a language $L$, up to some size.
\end{enumerate}
One might say that quantum computation is \emph{scalable in practise} if (and only if) each of the steps above can be performed efficiently in practise; and that the main challenge in building scalable quantum computers is to discover how to efficiently compute working specifications $S$.
If one motivates quantum computation on the grounds that one can in principle construct quantum computers to solve difficult problems, for the sake of consistency one should also accept a quantum meta-algorithm as providing a way to solve a problem via quantum computation.
A uniform quantum circuit family describes precisely the process of computing a description for a quantum circuit $C_n$, constructing the circuit, and using $C_n$ to solve a problem.
Uniform quantum circuit families, the standard model of quantum computation, are quantum meta-algorithms --- which under certain conditions are also admitted as quantum algorithms.
It remains only to ask what constraints we demand for the sake of uniformity of the circuit family.

We propose that polynomial-time uniform circuit families, over a polynomial-time specifiable gate-set, represent a reasonable framework in which to study quantum algorithms.
A circuit $C_n$ from a polytime-uniform circuit family $\{C_n\}_{n\ge1}$, constructed from polytime-specifiable gates, can be completely described in time $O(\poly n)$.
Conversely, any algorithm to construct a unitary circuit $C_n$, whose circuit structure and whose gates can be completely expressed as matrices in time $O(\poly n)$, computes a polytime-uniform circuit family $\{C_n\}_{n\ge1}$ with polytime-specifiable gates.
In this sense, we argue that any polynomial-time ``quantum meta-algorithm'', in the form of an efficient algorithm to describe quantum unitary circuits, represents an efficient ``quantum algorithm''.

\subsection{Reasons to move on}

We argue above that computational principles motivate polytime-specifiable gate-sets --- but what of the original computational motivations for gate sets of constant size?
There are two principles which motivate the restriction of quantum algorithms to finite gate sets: \textbf{(a)}~this limitation is imposed by defining quantum algorithms in terms of quantum Turing machines; and \textbf{(b)}~notwithstanding the study of the classes $\AC^k$, finite gate sets suffice for the theory of boolean circuits.
Having presented positive reasons to entertain broader notions of quantum algorithms, we now present reasons to abandon the model of quantum Turing machines as the basis for the theory of quantum computation, and also not to force the analogy to boolean circuits.

A Turing machine which provably halts provides a finite specification of the set of strings which it accepts.
Furthermore, it represents a simple model of what procedures can be achieved by a human operator; and essentially by this very fact, there are universal Turing machines, which can simulate any other Turing machine provided in a suitable representation.
For these reasons, deterministic Turing machines are a model of central importance in computational theory --- and also a model for the design of further computational models.
However, this does not guarantee that all models of computation which are defined in analogy to Turing machines should be the best choice of model to define a computational paradigm.

Quantum Turing machines may be an example of a Turing-like machine which is less useful than its deterministic counterpart at defining models of computation.
The very small amount of research in quantum computation which is actually described in terms of quantum Turing machines may be taken as a form of anthropological evidence of this.
Even in such abstract domains as complexity theory, the quantum Turing machine appears to have been made obsolete as an analytical tool, by quantum meta-algorithms such as uniform circuit families and adiabatically evolving spin systems.
The problems decided by such meta-algorithms are still provided by finite descriptions, \ie~by means of the deterministic Turing machine which embodies the meta-algorithm.
The other advantages of Turing machines --- human simulatability and the existence of a universal Turing machine --- simply do not apply to quantum Turing machines.
While there may be sound \emph{physical} grounds to impose constraints on the meta-algorithms (\eg~restricting to local unitary transformations or local Hamiltonian constraints), it seems spurious to impose \emph{computational} constraints merely to achieve parity with quantum Turing machines, given these limitations of quantum Turing machines as an analytical tool.

We also argue that the analogy of quantum circuits to logic circuits is limited, essentially because the state-space of qubits is richer than that of bits.
The boolean operations \textrm{\small AND}, \textrm{\small OR}, and \textrm{\small NOT} suffice to describe any boolean formula on a finite number of literals.
This could only be a rough analogy for quantum operations, as there is a continuum of valid unitary operations even on a single qubit (countably many of which have computable descriptions), rather than just two.
To suppose that the theory of these transformations should necessarily be fit into the mold of boolean circuit complexity is a curious conceit.
Circuit families constructed with constant-sized gate-sets are certainly a valid and interesting subject of study, and the fact that they include very good approximations to arbitrary unitaries (the Solovay--Kitaev theorem~\cite{KSV-2002,DN2006}) is a seminal result.
But it is not clear that this should mean that circuit families on constant-sized gate-sets should exclude all other circuit models.
If one accepts that the range of quantum operations is substantially different from the range of classical logic operations, it is reasonable to allow the model of quantum circuits to be more nuanced than the model of boolean circuits.
We propose polytime-uniform circuit families with polytime-specifiable gates as such a model of quantum circuits.

\subsection{Summary}

We have argued above that $\EQP$ --- standing in for quantum circuit families on constant-sized gate-sets --- appears to be unnecessarily limited from a theoretical standpoint.
While the reasons for these limitations were historically well-motivated, we find that these motivations end up working against the purpose of computational theory, and introduce distinctions between quantum algorithms and quantum meta-algorithms which are neither productive nor necessary to that purpose.

The theory of quantum computation has two obvious roles: as a crude caricature of engineering projects to build devices which exploit quantum mechanics to perform computation, and as a facet of the theory of computation which takes its inspiration from  quantum mechanics.
Taking the latter role seriously does not exclude the former, just as the study of nondeterministic Turing machines and $\AC^k$ algorithms does not prevent us from considering more easily realised models of computation.
This motivates a more generous theory of quantum algorithms, in which the study of exact quantum computation might prove more interesting.

It is on the basis of these arguments that we propose polytime-uniform circuit families, with \emph{polynomial-time} specifiable gate-sets (rather than \emph{constant-time} specifiable gate-sets), as the basis for the theory of quantum algorithms.
The special case of constant-sized (\ie~finite) gate-sets remains an important special case which is particularly of interest in the study of bounded-error quantum algorithms, but should be understood to simply be a well-motivated special case.
Barring a surprising discovery about uniform quantum circuit families, it seems likely that the containments $\EQP \subset \UnitaryP[\C] \subset \BQP$  are all strict.
Thus $\UnitaryP[\C]$ is likely to be more useful as a lower bound on the power of bounded-error quantum computation, and is more likely to provide for a thriving theory of exact polynomial-time quantum computation.

\end{document}

%% file: modal-k-circuits-preamble.tex
\usepackage[utf8x]{inputenc}
\usepackage{amsmath,amssymb,amsthm,bbm,mathtools,bm}
\usepackage{times}
\usepackage[mathscr]{eucal}
\usepackage{enumitem}
\usepackage[font=small]{caption}
\usepackage{hyperref}
\usepackage{tikz}
\usetikzlibrary{calc,positioning,fit,shapes,shapes.gates.logic.US}

\newif\ifpreamble
\newif\ifpreliminaries
\newif\iftikz
\newif\ifsynopsis
\newif\ifscholium
\newif\ifFullArticle

\hypersetup{
	pdftitle={On computation with 'probabilities' modulo k},
	pdfauthor={Niel de Beaudrap},
	colorlinks,%
	citecolor=black,%
	filecolor=black,%
	linkcolor=black,%
	urlcolor=black
}

\makeatletter
\renewcommand\thesubsubsection{\Alph{subsubsection}} 
\renewcommand\p@subsubsection{\thesubsection~} 
\makeatother

\theoremstyle{plain}
\newtheorem{theorem}{Theorem}
\newtheorem{lemma}[theorem]{Lemma}
\newtheorem{corollary}[theorem]{Corollary}
\newtheorem*{problem*}{Problem}

\theoremstyle{definition}
\newtheorem{proposition}[theorem]{Proposition}
\newtheorem{definition}{Definition}
\newtheorem*{definition*}{Definition}

\newtheoremstyle{proposition*}
  {\topsep}{\topsep}   																	
  {\itshape}{0pt}{\bfseries}{.}{5pt plus 1pt minus 1pt} 
  {\thmname{#1}\thmnote{~(#3)}}														
\theoremstyle{proposition*}
\newtheorem*{corollary*}{Corollary}
\newtheorem*{conjecture*}{Conjecture}

\newcommand\inter\cap
\newcommand\union\cup

\let\prsubset\subset
\renewcommand\subset\subseteq

\renewcommand\ge\geqslant
\renewcommand\le\leqslant

\newcommand\idop{\mathbbm 1}

\newcommand\toname[2][]   {\xrightarrow[\; #1 \;]{\; #2 \;}}
\newcommand\mapstoname    {\DOTSB\mapstochar\toname}

\renewcommand\vec\mathbf

\makeatletter
	\newcommand\ket[1]{\left| #1 \right\rangle\@ifnextchar\bra{\mspace{-4mu}}{}}
	\newcommand\bra[1]{\left\langle #1 \right|}

	\newcommand\mket[1]{\left| #1 \right)\@ifnextchar\mbra{\mspace{-4mu}}{}}
	\newcommand\mbra[1]{\left( #1 \right|}
\makeatother
\newcommand\bracket[2]{\left\langle #1 \!\left| #2 \!\right.\right\rangle}

\newcommand\ox\otimes
\newcommand\x\times
\renewcommand\setminus\smallsetminus

\newcommand\sox[1]{^{\otimes #1}}

\newcommand\parit[1]{\textup{(\kern-0.1ex\textit{#1}\kern0.1ex)}}

\makeatletter
\newcommand\eg{\emph{e.g.}}
\newcommand\ie{\emph{i.e.}}
\newcommand\etc{\@ifnextchar.{\emph{etc}}{\emph{etc.}}}
\newcommand\etal{\@ifnextchar.{\emph{et al}}{\emph{et al.}}}
\makeatother

\DeclareMathOperator\supp{supp}

\DeclareMathOperator\GF{GF}

\DeclareMathOperator\poly{poly}

\DeclareMathOperator\Null{null}

\let\SectionSymbol\S
\newcommand\B{\mathbb B}
\newcommand\N{\mathbb N}
\newcommand\Z{\mathbb Z}
\newcommand\Q{\mathbb Q}
\newcommand\R{\mathbb R}
\newcommand\C{\mathbb C}
\renewcommand\S{\mathbb S}

\newcommand\F{\mathbb F}

\DeclareMathOperator\Char{char}

\newcommand\sB{\mathscr{B}}

\newcommand\sS{\mathscr{S}}
\newcommand\sD{\mathscr{D}}

\newcommand\bB{\mathbf{B}}
\newcommand\bC{\mathbf{C}}
\newcommand\bS{\mathbf{S}}
\newcommand\bW{\mathbf{W}}
\newcommand\bX{\mathbf{X}}

\newcommand\rA{\mathrm A}
\newcommand\rB{\mathrm B}

\makeatletter
	\newcommand\abs{\@ifstar{\@abs\star}{\@abs{}}}
	\newcommand\@abs[2]{\left|#2\right|_{#1}}
	\newcommand\norm{\@ifstar{\@norm\star}{\@norm{}}}
	\newcommand\@norm[2]{\left\|#2\right\|_{#1}}
\makeatother

\renewcommand\P{\ensuremath{\mathsf P}}
\newcommand\PP{\ensuremath{\mathsf{PP}}}
\newcommand\NP{\ensuremath{\mathsf{NP}}}
\newcommand\RP{\ensuremath{\mathsf{RP}}}
\newcommand\co{\ensuremath{\mathsf{co}}}
\newcommand\BPP{\ensuremath{\mathsf{BPP}}}

\newcommand\BQP{\ensuremath{\mathsf{BQP}}}
\newcommand\PSPACE{\ensuremath{\mathsf{PSPACE}}}
\newcommand\EQP{\ensuremath{\mathsf{EQP}}}

\newcommand\Mod[1][]{\ensuremath{\mathsf{Mod}_{#1}}}

\newcommand\NC{\ensuremath{\mathsf{NC}}}
\newcommand\AC{\ensuremath{\mathsf{AC}}}

\newcommand\UP[1][]{\ensuremath{\mathsf{UP}_{#1}}}

\newcommand\pMQP[1][]{\ensuremath{{\mathsf{NGL}_{#1}\mathsf{P}}}}

\newcommand\GLP[1][]{\ensuremath{\mathsf{GLP}_{#1}}}
\newcommand\UnitaryP[1][]{\ensuremath{\mathsf{UnitaryP}_{\!#1}}}
\newcommand\AffineP[1][]{\ensuremath{\mathsf{AffineP}_{\!#1}}}

\DeclareMathOperator\id{id}

\newenvironment{romanum}{\begin{enumerate}[label=\textup{(\textit{\roman*})}]}{\end{enumerate}}

\newcommand\YES{\emph{yes}}
\newcommand\NO{\emph{no}}

\newcommand\e{\mathrm e}

\newcommand\herm{^\dagger}
\newcommand\trans[1][\!]{^{\raisebox{-0.2ex}{$\scriptstyle#1\mathsf T$}}}
\newcommand\sur[1]{^{\mbox{\tiny$(#1)$}}}
\renewcommand\sur[1]{^{(#1)}}

\newenvironment{smatrix}{%
	\mbox\bgroup
		\setlength{\arraycolsep}{2.5pt}%
		\footnotesize
		$\left[\,\begin{matrix}%
}{%
		\end{matrix}\,\right]$\egroup}

%% file: fig-orthogSimul.tex
	~\hspace*{-10em}
	\begin{tikzpicture}
			\coordinate (X1-0) at (-0.1,0.3);
			\coordinate (X2-0) at (-0.1,0.1);
			\coordinate (X3-0) at (-0.1,-0.1);
			\coordinate (X4-0) at (-0.1,-0.3);
			
			\coordinate (B1-0) at (-0.1,-0.8);
			\coordinate (S1a-0) at (-0.1,-1.0);
			\coordinate (S1b-0) at (-0.1,-1.2);
			
			\coordinate (B2-0) at (-0.1,-1.5);
			\coordinate (S2a-0) at (-0.1,-1.7);
			\coordinate (S2b-0) at (-0.1,-1.9);
			
			\coordinate (B3-0) at (-0.1,-2.2);
			\coordinate (S3a-0) at (-0.1,-2.4);
			\coordinate (S3b-0) at (-0.1,-2.6);
		
			\coordinate (W1-0) at (-0.1,-3.4);
			\coordinate (W2-0) at (-0.1,-3.6);
			\coordinate (W3-0) at (-0.1,-3.8);

			\coordinate (a'-0) at (-0.1,-5.0);
			
			\coordinate (C-0) at (-0.1,-5.5);
			
			\coordinate (A1a-0) at (-0.1,-7.5);
			\coordinate (A1b-0) at (-0.1,-7.7);
			\coordinate (A2a-0) at (-0.1,-7.9);
			\coordinate (A2b-0) at (-0.1,-8.1);
			\coordinate (A3a-0) at (-0.1,-8.3);
			\coordinate (A3b-0) at (-0.1,-8.5);
			\coordinate (A4a-0) at (-0.1,-8.7);
			\coordinate (A4b-0) at (-0.1,-8.9);
			\coordinate (A5a-0) at (-0.1,-9.1);
			\coordinate (A5b-0) at (-0.1,-9.3);
			\coordinate (A6a-0) at (-0.1,-9.5);
			\coordinate (A6b-0) at (-0.1,-9.7);
			
			\coordinate (C'-0) at (-0.1,-10.3);
			
			\coordinate (a-0) at (-0.1,-10.9);
			
			\foreach \a in {0,1} {%
			\ifnum\a>0%
				\node (X1-11) at (X1-11) {};
				\node (A6b-half) at ($(A6b-0) + (0.5,0)$) {}; 
				\node [draw, fill=gray!15!white, dashed, inner xsep=4pt, inner ysep=4pt, fit=(X1-11)(A6b-half)] {};
				\coordinate (B1-9half) at ($(B1-9)!0.5!(B1-10)$);
				\node at (A4a-0 -| B1-9half) {\Large$\tilde C_n$};

				\filldraw [gray!60!white] ($(S2a-2)!0.5!(S2b-3)$) node (x) {} ellipse (14pt and 4pt);
				\draw [very thick,gray!60!white] (x) -- ++ (-2em,-2ex) node [black,fill=gray!15!white,inner sep=1pt] {\footnotesize$\mathbf S$};

				\filldraw [gray!60!white] ($(B2-3)!0.5!(B2-4)$) node (x) {} ellipse (4pt and 10pt);
				\draw [very thick,gray!60!white] (x) -- ++ (-0.25em,-5ex) node [black,fill=gray!15!white,inner sep=2pt]
				{\footnotesize$\mathbf B$};
			\fi

			\xdef\prev{0}
				\foreach \t/\dt in {1/0.75,2/0.5,3/0.5,4/0.85,5/0.5,6/0.4,7/0.75,8/0.5,9/0.75,10/0.4,11/0.25,12/0.75,13/0.75,14/1,15/0.75,16/0.25,17/0.6,18/0.6} {
					\foreach \l in {X1,X2,X3,X4,B1,B2,B3,a',S1a,S1b,S2a,S2b,S3a,S3b,W1,W2,W3,C,A1a,A1b,A2a,A2b,A3a,A3b,A4a,A4b,A5a,A5b,A6a,A6b,C',a} {%
						\coordinate (\l-\t) at ($(\l-\prev) + (\dt,0)$);
					}

					\ifnum\t=2%
						\foreach \l/\dx in {S1a/0.1,S1b/0.03,S2a/0.0,S2b/-0.07,S3a/-0.1,S3b/-0.17} {
							\coordinate (\l-2) at ($(\l-2) + (\dx,0)$);
						}
					\fi

					\ifnum\t=3%
						\foreach \l/\dy in {B1/-0.5,B2/0,B3/0.5} {%
							\coordinate (\l-3) at ($(\l-3) + (0,\dy)$);
						}
						
						\foreach \l/\dx/\dy in {S1a/0.15/0,S1b/0.1/0,S2a/0.05/0.3,S2b/0/0.3,S3a/-0.05/0.6,S3b/-0.1/0.6} {
							\coordinate (\l-3) at ($(\l-3) + (0.4,-5.0)$);
							\coordinate (\l-3) at ($(\l-3) + (\dx,\dy)$);
						}
					\fi

					\ifnum\t=4%
					\foreach \l/\dx in {S1a/0.15,S1b/0.1,S2a/0.05,S2b/0,S3a/-0.05,S3b/-0.1} {
						\coordinate (\l-4) at (\l-4 -| C-4);
					}
					\fi
					
					\ifnum\t=5%
						\foreach \l in {A1a,A1b,A2a,A2b,A3a,A3b,A4a,A4b,A5a,A5b,A6a,A6b} {%
							\coordinate (\l-5) at ($(\l-5) + (-1,0)$);
						}
					\fi
					
					\ifnum\t=7%
						\foreach \l/\dx in {A1a/-0.15,A1b/-0.08,A2a/-0.05,A2b/0.02,A3a/0.05,A3b/0.12,A4a/0.15,A4b/0.22,A5a/0.25,A5b/0.32,A6a/0.35,A6b/0.42} {
							\coordinate (\l-7) at ($(\l-7) + (\dx,0)$);
						}
					\fi

					\ifnum\t=8%
						\foreach \l/\dy in {B1/0.5,B2/0,B3/-0.5,W1/0.5,W2/0,W3/-0.5} {%
							\coordinate (\l-8) at ($(\l-8) + (0,\dy)$);
						}
					
						\foreach \l/\a/\dx in {A1a/B1/-0.15,A1b/A1a/0.07,A2a/B2/-0.05,A2b/A2a/0.07,A3a/B3/0.05,A3b/A3a/0.07,A4a/W1/0.15,A4b/A4a/0.07,A5a/W2/0.25,A5b/A5a/0.07,A6a/W3/0.35,A6b/A6a/0.07} {
							\coordinate (\l-8) at ($(\a-8) + (\dx,-0.2)$);
						}
					\fi

					\ifnum\t=9%
						\coordinate (A1a-9) at (S1a-0 -| C-9);
						\xdef\lprev{1a}
						\foreach \l/\dy in {1b/0.2,2a/0.5,2b/0.2,3a/0.5,3b/0.2,4a/0.5,4b/0.2,5a/0.5,5b/0.2,6a/0.5,6b/0.2} {%
							\coordinate (A\l-9) at ($(A\lprev-9) + (0,-\dy)$);
							\xdef\lprev{\l}
						}
					\fi

					\ifnum\t=11%
						\coordinate (C'-11) at (6.35,-7.5);
						\coordinate (C'-11) at (C'-11 |- C'-0);
					\fi

					\ifnum\t=12%
						\coordinate (C'-12) at (6.65,-7.5);
						\coordinate (C'-12) at (X1-12 |- C'-0);
					\fi

					\ifnum\t=14%
					\foreach \l in {B1,B2,B3,a',S1a,S1b,S2a,S2b,S3a,S3b,W1,W2,W3,C,A1a,A1b,A2a,A2b,A3a,A3b,A4a,A4b,A5a,A5b,A6a,A6b} {%
						\coordinate (\l-14) at (\l-14 |- \l-0);
					}
					\fi

					\ifnum\a>0%
						\foreach \l in {X1,X2,X3,X4,B1,B2,B3,a',S1a,S1b,S2a,S2b,S3a,S3b,W1,W2,W3,C,A1a,A1b,A2a,A2b,A3a,A3b,A4a,A4b,A5a,A5b,A6a,A6b,C',a} {%
							\draw (\l-\prev) -- (\l-\t);
						}
					\fi
					\xdef\prev{\t}
					}
			}
			
			\foreach \l in {1a,2a,3a} {%
				\node [draw, rectangle, fill=white, inner sep=2pt, minimum height=1ex, minimum width=1ex]
							at (S\l-1) {\footnotesize$\,\big.K^{\mathclap{\phantom{\big.}}}$};
			}
			
			\foreach \j in {1a,1b,2a,2b,3a,3b} {%
				\filldraw [black] (S\j-4) circle (2pt);
			}
			\filldraw [black,fill=white] (C-4) circle (3pt);
			\draw ($(S3b-4) + (0,-1pt)$) -- ($(C-4) + (0,3pt)$);
			\draw ($(C-4) + (-3pt,0)$) -- ($(C-4) + (+3pt,0)$);
			
			\node (X1-5) at (X1-5) {};
			\node (a'-5) at (a'-5) {};
			\node [draw, rectangle, fill=white, inner sep=1pt] (x) at ($(X1-5)!0.5!(a'-5)$) {$~R_n~$};
			\node (x) [draw, rectangle, fill=white, inner sep=1pt, fit=(X1-5)(a'-5)(x)] {\phantom{$~R_n~$}};
			\node at ($(X1-5)!0.5!(a'-5)$) {$R_n$};
			
			\filldraw [black] (C-5) circle (2pt);
			\draw (C-5) -- (x);

			\foreach \l in {1a,1b,2a,2b,3a,3b,4a,4b,5a,5b,6a,6b} {%
				\coordinate (A\l-6) at (A\l-6 -| X1-4);
				\filldraw [black,fill=white] (A\l-6) circle (3pt);
				\draw ($(A\l-6) + (0,-3pt)$) -- ($(A\l-6) + (0,+3pt)$);
				\draw ($(A\l-6) + (-3pt,0)$) -- ($(A\l-6) + (+3pt,0)$);
			}

			\foreach \a in {A1a,A2a,A3a,A4a,A5a,A6a} {%
				\node [draw, rectangle, fill=white, inner ysep=2pt, inner xsep=1pt, minimum height=1ex, minimum width=1ex]
							at (\a-9) {\footnotesize$\big.K^{\mathsf T\mathclap{\phantom{\big|}}}$};
			}

			\foreach \l in {B1,B2,B3,W1,W2,W3,A1a,A1b,A2a,A2b,A3a,A3b,A4a,A4b,A5a,A5b,A6a,A6b,C'} {%
				\filldraw [black,fill=white] (\l-11) circle (3pt);
				\draw ($(\l-11) + (-3pt,0)$) -- ($(\l-11) + (+3pt,0)$);
				\draw ($(\l-11) + (0,-3pt)$) -- ($(\l-11) + (0,+3pt)$);
			}
						
			\node (B1-12) at (B1-12) {};
			\node (a'-12) at (a'-12) {};

			\foreach \l in {B1,A1a,A1b,B2,A2a,A2b,B3,A3a,A3b,W1,A4a,A4b,W2,A5a,A5b,W3,A6a,A6b,a'} {%
				\filldraw [black] (\l-12) circle (2pt);
			}
			\filldraw [black,fill=white] (C'-12) circle (3pt);
			\draw ($(B1-12)+(0,1pt)$) -- ($(C'-12) + (0,-3pt)$);
			\draw ($(C'-12) + (-3pt,0)$) -- ($(C'-12) + (+3pt,0)$);

			\node (X1-13) at (X1-13) {};
			\node (A6b-14half) at ($(A6b-15) + (-0.4,0)$) {}; 
			\node (x) [draw, rectangle, fill=gray!15!white, inner xsep=4pt, inner ysep=4pt, fit=(X1-13)(A6b-14half)]
				{\Large$\tilde C_n^{-1}$};
			\coordinate (C'-13half) at (x |- C'-0);
			\filldraw [black] (C'-13half) circle (2pt);
			\draw (C'-13half) -- (x);
				
			\foreach \l in {B1,B2,B3,S1a,S1b,S2a,S2b,S3a,S3b,C,W1,W2,W3,a',A1a,A1b,A2a,A2b,A3a,A3b,A4a,A4b,A5a,A5b,A6a,A6b,a} {%
				\filldraw [black,fill=white] (\l-16) circle (3pt);
				\draw ($(\l-16) + (-3pt,0)$) -- ($(\l-16) + (+3pt,0)$);
				\draw ($(\l-16) + (0,-3pt)$) -- ($(\l-16) + (0,+3pt)$);
			}

			\foreach \l in {B1,S1a,S1b,B2,S2a,S2b,B3,S3a,S3b,W1,W2,W3,a',C,A1a,A1b,A2a,A2b,A3a,A3b,A4a,A4b,A5a,A5b,A6a,A6b,C'} {%
				\filldraw [black] (\l-17) circle (2pt);
			}
			\filldraw [black,fill=white] (a-17) circle (3pt);
			\draw ($(B1-17) + (0,1pt)$) -- ($(a-17) + (0,-3pt)$);
			\draw ($(a-17) + (-3pt,0)$) -- ($(a-17) + (3pt,0)$);
			
			\foreach \d/\t in {west/0, east/18} {
				\node [label=\d:\small$\mathbf X$] at ($(X1-\t)!0.5!(X4-\t)$) {};
				\node [label=\d:\small$\begin{matrix}\mathbf B\\\cup\\\mathbf S\end{matrix}$] at ($(B1-\t)!0.5!(S3b-\t)$) {};
				\node [label=\d:\small$\mathbf W$] at ($(W1-\t)!0.5!(W3-\t)$) {};
				\node [label=\d:\small$\:\!a'\!\!\;$] at (a'-\t) {};
				\node [label=\d:\small$\mathbf C$] at (C-\t) {};
				\node [label=\d:\small$\!\mathbf S'\!\!$] at ($(A1a-\t)!0.5!(A6b-\t)$) {};
				\node [label=\d:\small$\!\mathbf C'\!\!$] at (C'-\t) {};
				\node [label=\d:\small$a$] at (a-\t) {};
			}
\end{tikzpicture}
\hspace*{-10em}~